\documentclass[acmtog]{acmart}

\newcommand{\revised}[1]{{\color{black} #1}}
\usepackage{booktabs} 

\usepackage[ruled]{algorithm2e} 

\SetAlFnt{\small}
\SetAlCapFnt{\small}
\SetAlCapNameFnt{\small}
\SetAlCapHSkip{0pt}

\graphicspath{{figures/}}

\usepackage{amsfonts} 
\usepackage{amsmath} 
\usepackage{amssymb}  
\usepackage{overpic}
\usepackage{pgfplots}
\usepackage{pgf,tikz}
\usepackage{bm}
\usepackage{wrapfig}
\usepackage{multirow}
\usepackage{mathtools}

\usepackage{verbatim}

\SetAlFnt{\small}
\SetAlCapFnt{\small}
\SetAlCapNameFnt{\small}
\SetAlCapHSkip{0pt}

\newcommand{\M}{\mathcal{M}}
\newcommand{\A}{A}
\newcommand{\N}{\mathcal{N}}

\renewcommand{\C}{\mathbf{C}}

\renewcommand{\phi}{\varphi}

\DeclareMathOperator*{\argmin}{arg\,min}

\usepackage{soul}

\usepackage{listings}
\usepackage{color} 
\definecolor{mygreen}{RGB}{28,172,0} 
\definecolor{mylilas}{RGB}{170,55,241}

\setcopyright{acmlicensed}
\acmJournal{TOG}
\acmYear{2019}
\acmVolume{38}
\acmNumber{6}
\acmArticle{155}
\acmMonth{11} 
\acmDOI{10.1145/3355089.3356524}

\newcommand{\name}{\textsc{ZoomOut}}

\begin{document}
\title{\name: Spectral Upsampling for Efficient Shape Correspondence}

\author{Simone Melzi}
\authornote{Equal contribution.}
\affiliation{
	\institution{University of Verona}
}
\email{simone.melzi@univr.it}

\author{Jing Ren}
\authornotemark[1]
\affiliation{
	\institution{KAUST}
}
\email{jing.ren@kaust.edu.sa}

\author{Emanuele Rodol\`a}
\affiliation{
	\institution{Sapienza University of Rome}
}
\email{rodola@di.uniroma1.it}

\author{Abhishek Sharma}
\affiliation{
	\institution{LIX, \'Ecole Polytechnique}
}
\email{kein.iitian@gmail.com}

\author{Peter Wonka}
\affiliation{
	\institution{KAUST}
}
\email{pwonka@gmail.com}

\author{Maks Ovsjanikov}
\affiliation{
	\institution{LIX, \'Ecole Polytechnique}
}
\email{maks@lix.polytechnique.fr}
\renewcommand\shortauthors{Melzi. et al}

\begin{abstract}
  We present a simple and efficient method for refining maps or correspondences by iterative upsampling in the spectral domain that can be implemented in a few lines of code.  Our main observation is that high quality maps can be obtained
  even if the input correspondences are noisy or are encoded by a small number of coefficients in a spectral
  basis. We show how this approach can be used in conjunction with
  existing initialization techniques across a range of application scenarios, including symmetry detection, map refinement across complete shapes, non-rigid partial shape matching and function transfer. In each application we demonstrate an improvement with respect to both the quality of the results and the computational speed compared to the best competing methods, with up to two orders of magnitude speed-up in some applications. We also demonstrate that our method is both robust to noisy input and is scalable with respect to shape complexity. Finally, we present a theoretical justification for our approach, shedding light on structural properties of functional maps.
\end{abstract}

\ccsdesc[500]{Computing methodologies~Shape analysis}

\keywords{Shape Matching, Spectral Methods, Functional Maps}

\maketitle

\section{Introduction}
\label{sec:introduction}

Shape matching is a task that occurs in countless applications in computer graphics, including shape interpolation
\cite{kilian2007geometric} and statistical shape analysis \cite{FAUST}, to name a few.

An elegant approach to non-rigid shape correspondence is provided by \emph{spectral techniques}, which are broadly
founded on the observation that near-isometric shape matching can be formulated as an alignment problem in certain
higher-dimensional embedding spaces \cite{jain2006,ovsjanikov2012functional,maron2016point,biasotti2016recent}.
%
%
Despite significant recent advances and their wide practical applicability, however, spectral methods can both be
computationally expensive and unstable with increased dimensionality of the spectral embedding. On the other hand, a
reduced dimensionality results in very approximate maps, losing medium and high-frequency details and leading to
significant artifacts in applications.

In this paper, we show that a higher resolution map can be recovered from a lower resolution one through a remarkably
simple and efficient iterative spectral up-sampling technique, which consists of the following two basic steps:
\begin{enumerate}
\item Convert a $k \times k$-size functional map to a pointwise map.
\item Convert the pointwise map to a $k+1\times k+1$ functional map.
\end{enumerate}

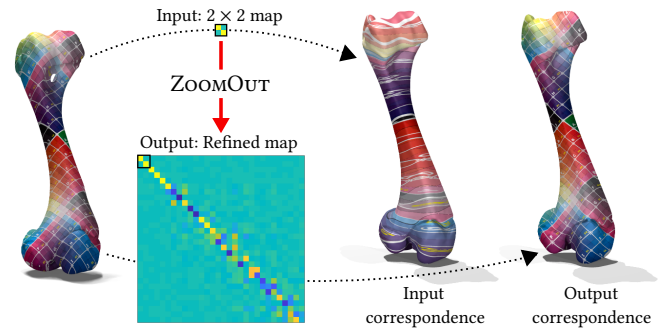
\begin{figure}[t]
\vspace{-3.2cm}
  \centering
  \begin{overpic}
  [trim=0cm 0cm 0cm 0cm,clip,width=\linewidth]{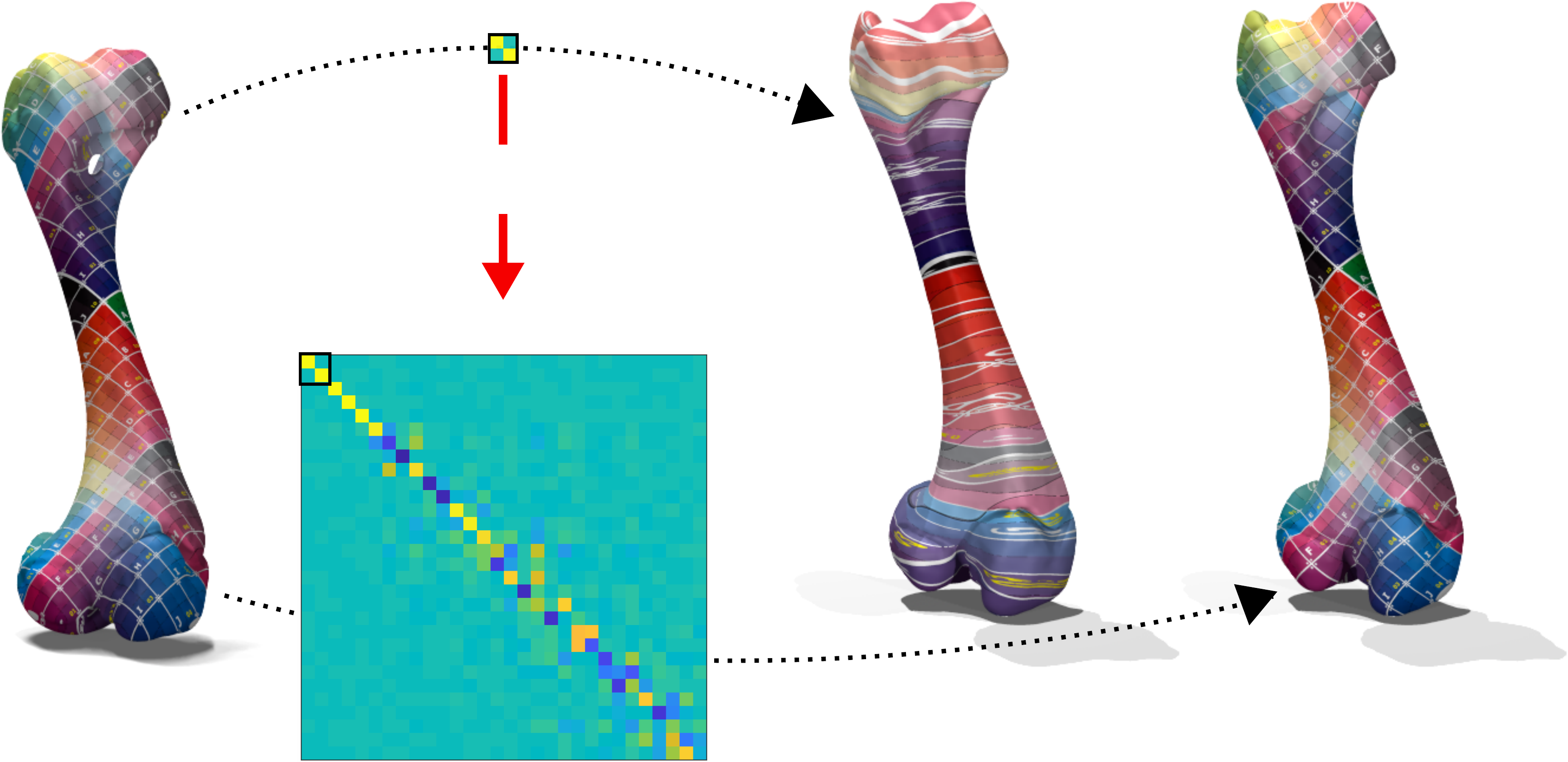}
    \put(22.1,47){\footnotesize Input: $2\times 2$ map}
    \put(19.5,27.2){\footnotesize Output: Refined map}
    \put(60.5,3.6){\footnotesize Input}
    \put(54.6,0){\footnotesize correspondence}
    \put(85.5,3.6){\footnotesize Output}
    \put(80.5,0){\footnotesize correspondence}
    \put(24.2,36){\name}
  \end{overpic}
 \caption{Given a small functional map, here of size $2 \times 2$ which corresponds to a very noisy point-to-point correspondence (middle right) our method can efficiently recover both a high resolution functional and an accurate dense point-to-point map (right), both visualized via texture transfer from the source shape (left).\vspace{-3mm}\label{fig:example}}
 \vspace{-0.1cm}
\end{figure}

Our main observation is that by iterating the two steps above,
starting with an approximate initial map, encoded using a small number
of spectral coefficients (as few as 2--15), we can obtain an accurate
correspondence at very little computational cost.

We further show that our refinement technique can be combined with standard map initialization methods to obtain
state-of-the-art results on a wide range of problems, including intrinsic symmetry detection, isometric shape matching,
non-rigid partial correspondence and function transfer among others. Our method is robust to significant changes in
shape sampling density, is easily scalable to meshes containing tens or even hundreds of thousands of vertices and is
significantly (up to 100-500 times in certain cases) faster than existing state-of-the-art map refinement approaches,
while producing comparable or even superior results. For example, Figure~\ref{fig:example} shows a result obtained
with our method, where starting from an initial $2\times 2$ functional map, we recover a high
resolution functional and an accurate pointwise correspondence.


\paragraph{Contributions.} To summarize:
\begin{enumerate}
\item We introduce a very simple map refinement method capable of improving upon the state of the art in a diverse set of shape correspondence problems; for each problem we can achieve the same or better quality at a fraction of the cost compared to the current top performing methods.
\item We demonstrate how higher-frequency information can be extracted from low-frequency spectral map representations.
\item We introduce a novel variational optimization problem and develop a theoretical justification of our method, shedding light on structural properties of functional maps.
\end{enumerate}


\section{Related Work}
\label{sec:related}
Shape matching is a very well-studied area of computer graphics. Below we review the methods
most closely related to ours, concentrating on spectral techniques for finding
correspondences between non-rigid shapes. We refer the interested readers to recent surveys
including \cite{van2011survey,tam2013registration,biasotti2016recent} for a more in-depth treatment of the area.

\paragraph{Point-based Spectral Methods.} Early spectral methods for shape correspondence were based on directly optimizing pointwise maps between spectral shape embeddings based on either adjacency or Laplacian matrices of graphs and triangle meshes
\cite{umeyama1988,scott1991,jain2006,jain2007,mateus08,ovsjanikov2010}.
Such approaches suffer from the requirement of a good initialization, and rely on restricting assumptions about the type of transformation relating the shapes. An initialization algorithm with optimality guarantees, although limited to few tens of points, was introduced in~\cite{maron2016point} and later extended to deal with intrinsic symmetries in~\cite{dym2017exact}.
Spectral quantities (namely, sequences of Laplacian eigenfunctions) have also been used to define pointwise descriptors, and employed within variants of the quadratic assignment problem in~\cite{dubrovina2010matching,dubrovina2011approximately}. These approaches have been recently generalized by spectral generalized multidimensional scaling~\cite{aflalo2016spectral}, which explicitly formulates minimum-distortion shape correspondence in the spectral domain.

\paragraph{Functional Maps.} Our approach fits within the functional map framework, which was
originally introduced in \cite{ovsjanikov2012functional} for solving non-rigid shape matching problems, and extended significantly in follow-up works, including
\cite{kovnatsky2013coupled,aflalo2013spectral,rodola2017partial,ezuz2017deblurring} among others
(see \cite{ovsjanikov2017computing} for an overview). These methods assume as input a set of corresponding functions, which can be derived from pointwise landmarks, dense descriptor fields, or from 
 region correspondences. They then estimate a functional map matrix that allows to transfer
real-valued functions across the two shapes, which is then converted to a pointwise map.

Although the first step reduces to the solution of a linear system of equations, this last step can be difficult and error prone \cite{rodola-vmv15,ezuz2017deblurring}. As a result, several strong regularizers have been proposed to promote certain desirable properties: see  \cite{rodola2017partial,nogneng17,huang2017adjoint,litany2017fully,burghard2017embedding,wang2018vector}. More recently, several other constraints on functional maps have been proposed to promote continuity of the pointwise correspondence \cite{poulenard18}, map curves defined on shapes \cite{gehre2018interactive}, extract more information from given descriptor constraints \cite{wang2018kernel}, and for incorporating orientation information into the map inference pipeline \cite{ren2018continuous}.

\revised{
In a concurrent work, \cite{Shoham2019hierarchicalFmap} also compute hierarchical functional maps by building explicit hierarchies in the spatial domain using subdivision surfaces. Unlike this work, our method operates purely in the spectral domain, and does not require computing additional shape hierarchies.
}

\paragraph{High-frequency Recovery.} Several approaches have also observed that high-frequency
information can be recovered even if the input functional map is small or noisy. This includes both optimizing an input map using vector field flow \cite{corman2015continuous}, recovering precise (vertex-to-point) maps \cite{ezuz2017deblurring} from low frequency functional ones, and using pointwise products to extend the space of functions that can be transferred \cite{nogneng18}.

\paragraph{Iterative Map Refinement.} We also note other commonly-used relaxations for matching problems based on optimal transport, e.g. \cite{solomon2016entropic,mandad2017variance}, which are often solved through iterative refinement. Other techniques that exploit a similar formalism for solving optimal assignment include the Product Manifold Filter and its variants~\cite{vestner2017product,vestner2017efficient}. Map refinement has also been considered in the original functional maps approach \cite{ovsjanikov2012functional}
where Iterative Closest Point in the spectral embedding has been used to improve input functional maps. Finally, in the context of shape collections
\cite{wang2013,huang2014functional,wang2013exact}, cycle-consistency constraints have been used to iteratively improve input map quality. We further discuss methods most closely-related to ours in Section \ref{sec:relation_others} below.

Although these techniques can be very effective for obtaining high-quality correspondences, methods based purely on optimization in the spatial domain can quickly become prohibitively expensive even for moderate sampling density. On the other hand, spectral techniques can provide accurate solutions for low-frequency matching, but require significant effort to recover a high-quality dense pointwise correspondence; further, such approaches are often formulated as difficult optimization problems and suffer from instabilities for large embedding dimensions.



\begin{figure}[t]
  \centering
  \begin{overpic}
  [trim=4.5cm 3cm 7cm 0cm,clip,width=\linewidth, grid=false]{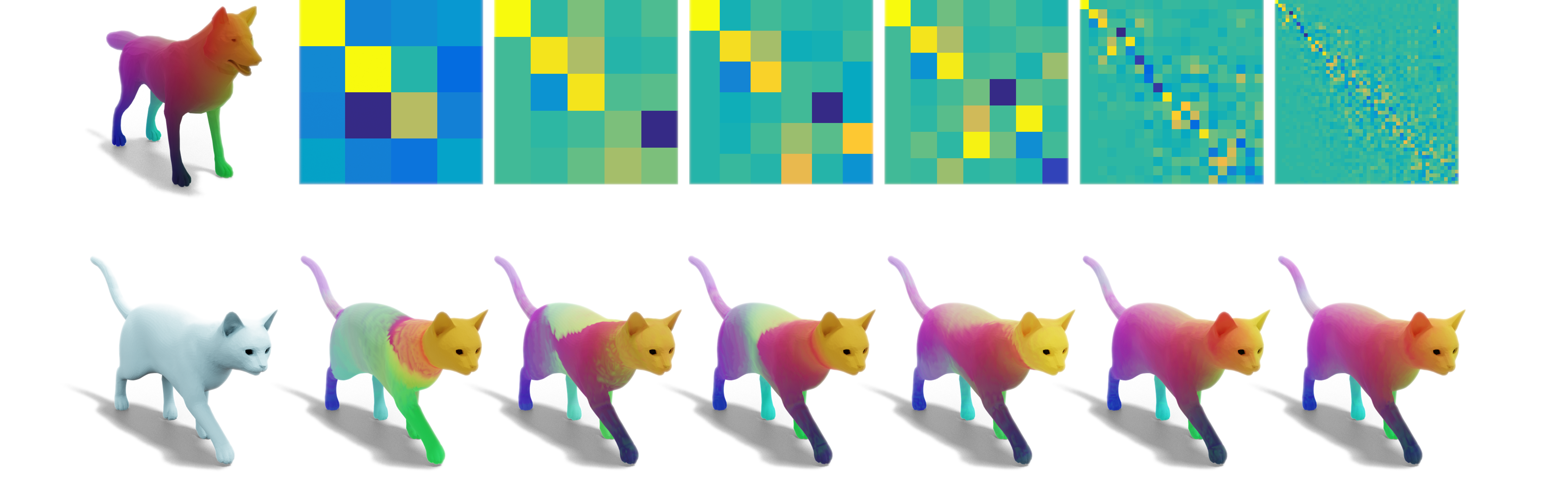}
  \put(4.9,36){\footnotesize{Source}}
  \put(4,33.5){\scriptsize{ $n = 4.3K $}}
  
  \put(5,17){\footnotesize{Target}}
  \put(4,14.5){\scriptsize{ $n = 10K$}}
  
  \put(18,35.75){\footnotesize{Ini: $4\times4$}}
  \put(31.7,36.5){\footnotesize{zoomOut}}
  \put(32.8,34){\footnotesize{to $5\times 5$}}
  
  \put(48,35.75){\footnotesize{$6\times 6$}}
  \put(62,35.75){\footnotesize{$7\times 7$}}
  \put(74,35.75){\footnotesize{$20\times 20$}}
  
  \put(86,36.5){\footnotesize{zoomOut}}
  \put(85.8,34){\footnotesize{to $50\times 50$}}
  \end{overpic}
  \vspace{-0.7cm}
  \caption{\name ~example. Starting with a noisy functional map of size $4 \times 4$ between the two
  shapes we progressively upsample it using our two-step procedure and visualize the corresponding
  point-to-point map at each iteration via color transfer. Note that as the size of the functional
  map grows, the map becomes both more smooth and more semantically accurate. We denote the number of vertices by $n$.
\label{fig:eg:cat_wolf}
\vspace{-0.25cm}}
\end{figure}
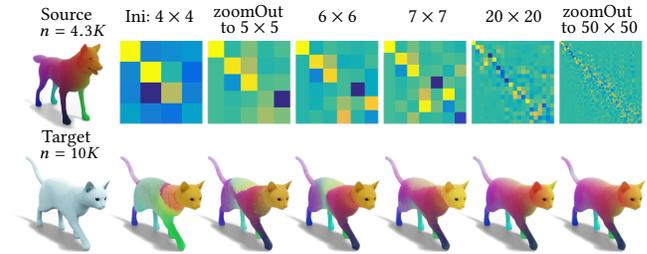

\section{Background \& Notation}
\label{sec:background}
%
In this section we introduce the main background notions and notation
used throughout the rest of the paper.

Given a pair of shapes $\M$ and $\N$, typically represented as
triangle meshes, we associate to them the positive semi-definite
Laplacian matrices $L_\M, L_\N$, discretized via the standard
cotangent weight scheme \cite{pinkall93}, so that
$L_\M = \A_\M^{-1} W_\M$, where $\A_\M$ is the diagonal matrix of
lumped area elements and $W_\M$ is the cotangent weight matrix, with
the appropriate choice of sign to ensure positive
semi-definiteness. We make use of the basis consisting of the first
$k_\M$ eigenfunctions of the Laplacian matrix, and encode it in a
matrix
$\Phi^{k_\M}_{\M} = [ \phi_{1}^{\M}, \phi_{2}^{\M}, \ldots ,
\phi_{k_{\M}}^{\M} ]$ having the eigenfunctions as its columns.
%
%
We define the {\em spectral embedding} of $\M$ as the $k_\M$-dimensional point set $\big\{ \big(\phi_{1}^{\M}(x), \ldots , \phi_{k_{\M}}^{\M}(x)\big) \, | \, x\in\M \big\}$.

Given a point-to-point map $T: \M \to \N $, we denote by $\Pi$ its matrix
representation, s.t. $\Pi(i,j) = 1$ if $T(i) = j$ and $0$ otherwise, \revised{where $i$ and $j$ are vertex indices on shape $\M$ and $\N$, respectively. Note that the matrix $\Pi$ is an equivalent matrix representation of any pointwise map $T$ without extra
  assumptions, such as bijectivity.} The corresponding
\emph{functional map} $~\mathbf{C}$ is a linear transformation taking functions on $\N$ to functions on $\M$; in matrix notation, it is given  by the projection of $\Pi$ onto the corresponding functional basis: 
\begin{align}
\label{eq:fmap_define}
\C = \Phi^{+}_{\M} \Pi \Phi_{\N}\,,
\end{align}
where $^{+}$ denotes the Moore-Penrose pseudo-inverse. When the eigenfunctions are orthonormal with respect to the area-weighted inner product, so that $\Phi_\M^\top \A_\M \Phi_\M = Id$, then Eq.~\eqref{eq:fmap_define} can be written as: $\C = \Phi_\M^\top \A_\M \Pi \Phi_\N$. Note that $\C$ is a matrix of size $k_\M \times k_\N$, independent of the number of
vertices on the two shapes. 

A typical pipeline for computing a correspondence using the functional map representation proceeds as follows \cite{ovsjanikov2017computing}: 1) Compute a moderately-sized basis (60-200 basis functions) on each shape; 2) Optimize for a functional map $\C_{\text{opt}}$ by minimizing an energy, based on preservation of descriptor functions or landmark correspondences and regularization, such as commutativity with the Laplacian operators; 3) Convert $\C_{\text{opt}}$ to a
point-to-point map. The complexity of this pipeline directly depends on the size of the chosen basis, and thus the
dimensionality of the spectral embedding. 
 Smaller bases allow more stable and efficient functional map recovery but result in approximate pointwise correspondences, while larger functional maps
can be more accurate but are also more difficult to optimize for and require stronger priors.

Our main goal, therefore, is to show that accurate pointwise correspondences can be obtained even in
the presence of only small, or approximate functional maps.



\section{\name: Iterative Spectral Upsampling}
\label{sec:method}

As input we assume to be given either a small functional map $\C_0$  or a point-to-point correspondence $T: \M \rightarrow \N$; both may be affected by noise. We will discuss
the role and influence of the input map in detail in the following sections. If it is a point-to-point map, we first convert it to a functional one via
Eq.~\eqref{eq:fmap_define}. For simplicity, we first state our method and then provide its theoretical derivation from a variational optimization problem in Section \ref{sec:theory}.

Given an input  $k_\M \times k_\N$ functional map $\C_0$ our goal is to extend it to a new map $\C_1$ of size
$(k_{\M}+1)\times( k_{\N}+1)$ without any additional information. We do so by a simple two-step
procedure:
\begin{enumerate}
\item Compute a point-to-point map $T$ via Eq.~\eqref{eq:p2p_recovery}, and encode it as a matrix $\Pi$.
\item Set $\C_{1} = (\Phi^{k_\M+1}_\M)^\top \A_\M \Pi \, \Phi^{k_\N+1}_\N.$
\label{eq:iteration}
\end{enumerate} 

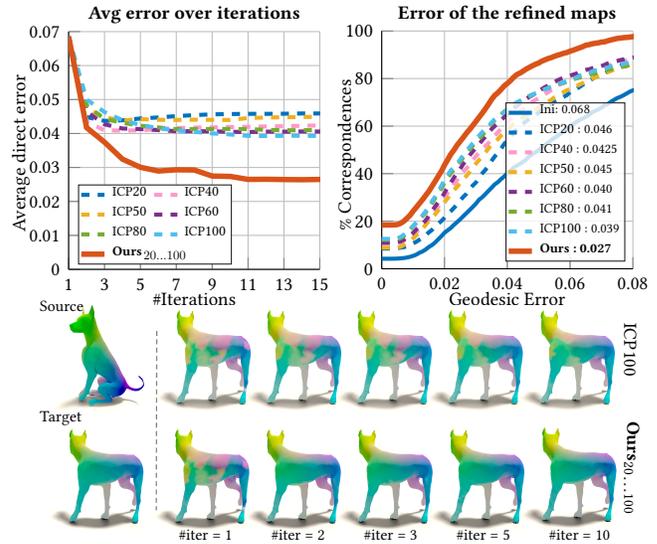
\begin{figure}[!t]
\centering
\includegraphics[width=1\linewidth]{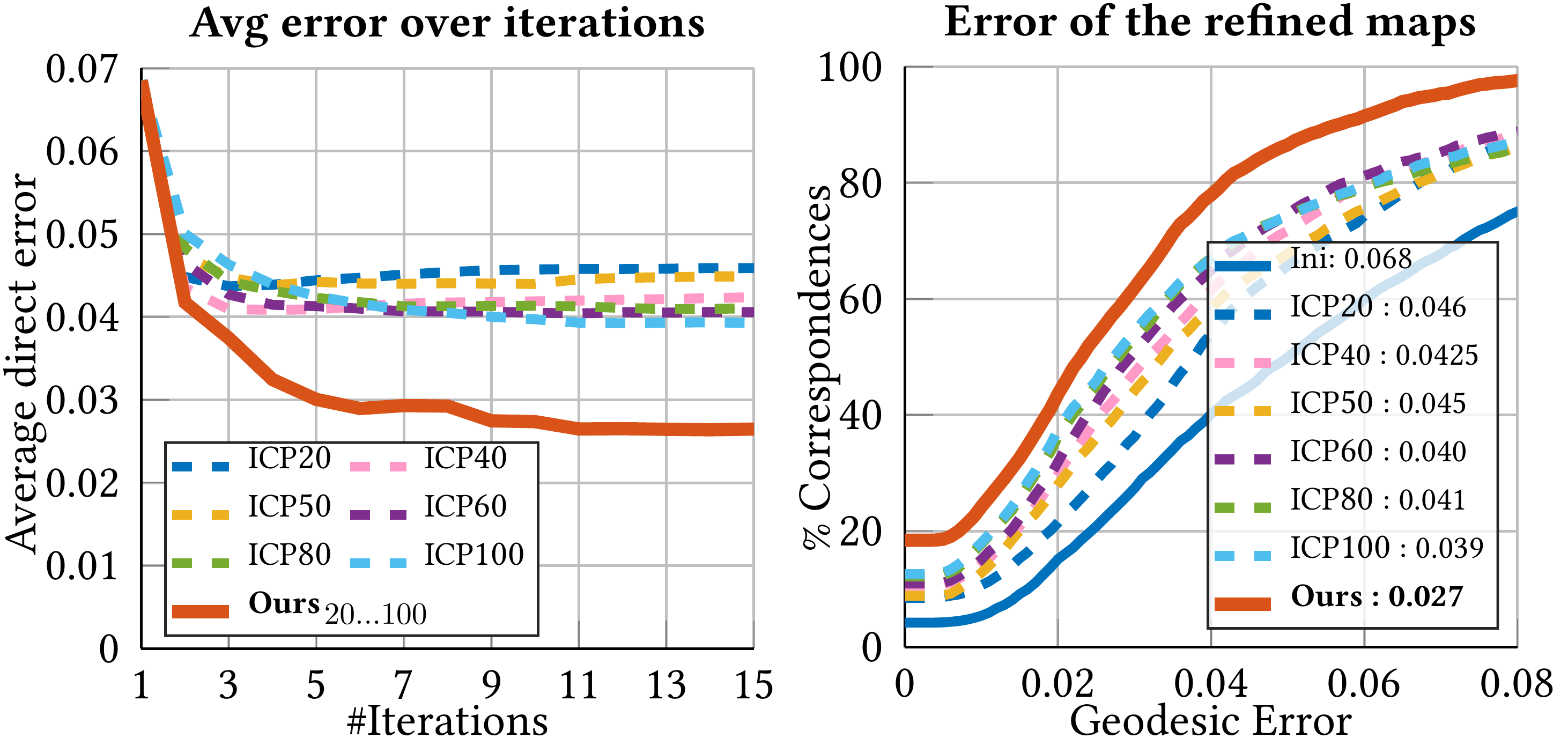}
\begin{overpic}[trim=0cm 0cm 0cm 0cm,clip,width=0.9\linewidth,grid=false]{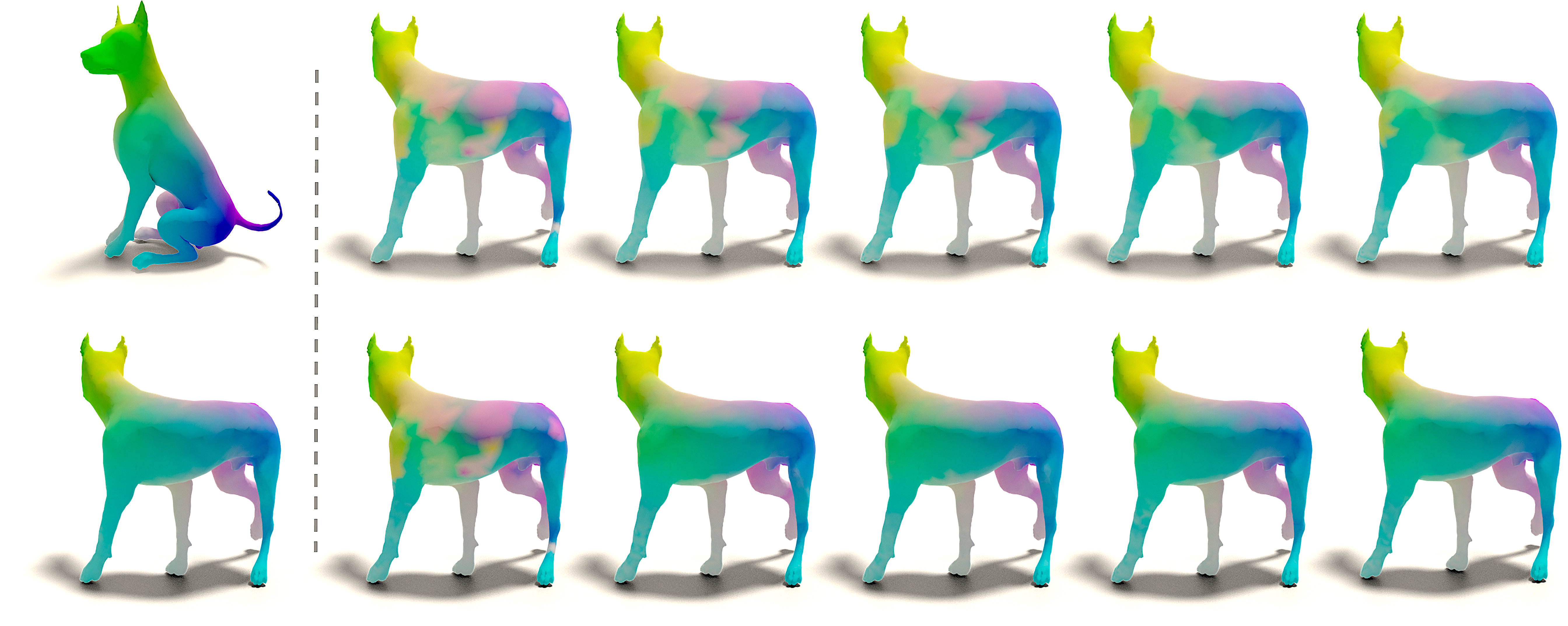}
\put(0,39){\scriptsize Source}
\put(0,20){\scriptsize Target}

\put(24,-0.5){\scriptsize{\#iter = 1}}
\put(40,-0.5){\scriptsize{\#iter = 2}}
\put(56,-0.5){\scriptsize{\#iter = 3}}
\put(72,-0.5){\scriptsize{\#iter = 5}}
\put(88,-0.5){\scriptsize{\#iter = 10}}
\put(101,37){\rotatebox{-90}{\footnotesize{ICP100}}}
\put(101,20){\rotatebox{-90}{\footnotesize{\textbf{Ours$_{20...100}$}}}}
\end{overpic}
\vspace{-0.2cm}
\caption{\label{fig:res:eg:tosca_iter}Comparison of map quality during ICP iterations in different (fixed) dimensions vs. \name\ from 20 to 100 with step 5. 
 Top row: average error of pointwise maps during refinement and the error summary of the refined maps after 15 iterations. Note that regardless of dimension, ICP gets trapped in a local minimum. Bottom row: visualization of the refined maps at iteration 1, 2, 3, 5, and 10 of ICP with dimension 100 vs. our method. 
}
\vspace{-0.2cm}
\end{figure}

We then iterate this procedure to obtain progressively larger functional maps
$\C_{0}, \C_{1}, \C_{2}, ..., \C_{n}$ until some sufficiently large $n$. As we demonstrate below,
this remarkably simple procedure, which can be implemented in only a few lines of code (see Appendix B), can result
in very accurate functional and pointwise maps even given very small and possibly noisy input.
To compute a pointwise map from a given $\C$ in step (1), we solve the following problem:
\begin{align}
\label{eq:p2p_recovery}
T(p) = \argmin_{q} \|\C (\Phi_\N(q))^\top - (\Phi_\M(p))^\top\|_2 , ~\forall ~p \in \M
\end{align}
where $\Phi_\M(p)$ denotes the $p^{\text{th}}$ row of the matrix of eigenvectors $\Phi_\M$. This
procedure gives a point-to-point map $T: \M \rightarrow \N$, and can be implemented via a nearest-neighbor query in $k_\M$-dimensional space. It is also nearly identical,
up to change in direction, to the pointwise map recovery described in the original functional
maps article \cite[Section 6.2]{ovsjanikov2012functional} but differs from other recovery steps, introduced, e.g., in \cite{ezuz2017deblurring} as we discuss below. 

Figure \ref{fig:eg:cat_wolf} shows an example of \name\ on a pair of animal shapes from the TOSCA
dataset \cite{TOSCA}. Starting with a $4\times 4$ functional map, we show both the functional
and point-to-point (visualized via color transfer) maps throughout our upsampling iterations. Note
how the pointwise map becomes both more smooth and accurate as the functional map grows.

We use the term ``upsampling'' in the description of our method to highlight the fact that at every iteration \name\ introduces additional frequencies and thus intuitively adds samples \emph{in the spectral domain} for representing a map.



\subsection{Map Initialization}


We initialize our pipeline by optimizing for a $k_\M \times k_\N$ functional map $\C_0$ using an existing approach; we 
tested recent techniques, including
\cite{rodola2017partial,ren2018continuous} among others, across different settings
described in detail in Section \ref{sec:results}. 

The key parameter for the initialization is the size of the functional map, and in most settings,
we set $k_\M = k_\N = k$ for some small $k$. This value ranges between $4$ and $20$ in all of our
experiments, and allows us to obtain high quality maps by upsampling $\C_0$ to sizes up to $200\times 200$
depending on the scenario. 
 We have observed that the key requirement for the input map $\C_0$ is that although it
can be noisy and approximate, it should generally disambiguate between the possible symmetries
exhibited by the shape. Thus, for example, if $4$ basis functions are sufficient to distinguish left
and right on the animal models shown in Figure \ref{fig:eg:cat_wolf}, then with a functional map of
this size our method can produce an accurate final correspondence. Perhaps the most difficult case
we have encountered is in disambiguating front and back in human shapes which requires approximately
15 basis functions. This is still significantly smaller than typical values in existing functional
map estimation pipelines, which are based on at least 60 to 100 basis functions to compute accurate
maps.

\begin{figure}[!t]
\centering
\hspace{-0.37cm}
%
%
\definecolor{mycolor1}{rgb}{0.00000,0.44700,0.74100}%
\definecolor{mycolor2}{rgb}{0.85000,0.32500,0.09800}%
\definecolor{mycolor3}{rgb}{0.92900,0.69400,0.12500}%
\definecolor{mycolor4}{rgb}{0.49400,0.18400,0.55600}%
\definecolor{mycolor5}{rgb}{0.46600,0.67400,0.18800}%
\pgfplotsset{scaled y ticks=false}
\pgfplotsset{
compat=1.11,
legend image code/.code={
\draw[mark repeat=2,mark phase=2]
plot coordinates {
(0cm,0cm)
(0.15cm,0cm)        
(0.3cm,0cm)         
};%
}
}
\begin{tikzpicture}

\begin{axis}[%
width=0.28\linewidth,
height=0.29\linewidth,
scale only axis,
every x tick label/.append style={font=\color{black}, font=\tiny},
every y tick label/.append style={font=\color{black}, font=\tiny},
xmin=30,
xmax=200,
xlabel={\footnotesize Size of the map $(k)$},
xlabel style={at={(0.5,-0.07)}},
xtick={30, 50, 75, 100, 150, 200},
xticklabels={30, 50, 75, 100, 150, 200},
ymin=0.04,
ymax=0.07,
ytick={0.04,0.05,0.06,0.07},
yticklabels={0.04,0.05,0.06,0.07},
yminorticks=true,
ylabel={\footnotesize Average geodesic error},
ylabel style={at={(-0.15,0.48)}},
axis background/.style={fill=white},
xmajorgrids,
ymajorgrids,
yminorgrids,
legend style={at={(1.83,0.71)},anchor=south east,legend cell align=left, align=left, draw=white!15!black},
legend style={inner sep=0pt}
]
\addplot [color=mycolor2,dashed,line width=1.5pt,mark=o,mark options={solid}]
  table[row sep=crcr]{%
30	0.0630698980782679\\
50	0.064835644894863\\
75	0.0653196899319946\\
100	0.0601279841813541\\
150	0.0658588534582569\\
200	0.0674629489696546\\
};
\addlegendentry{\footnotesize $\text{Fmap}_{k}+\text{ICP}$};

\addplot [color=mycolor5,dashed,line width=1.5pt,mark=o,mark options={solid}]
  table[row sep=crcr]{%
30	0.0577368139540182\\
50	0.0551554185651868\\
75	0.0564840469521786\\
100	0.0519360284370615\\
150	0.0503639961400362\\
200	0.0464930228555058\\
};
\addlegendentry{\footnotesize $\text{\textbf{Ours}}_{10 \ldots k}$};

\end{axis}
\end{tikzpicture}%
\hspace{-1.8cm}\begin{overpic}[trim=0cm 0cm 0cm 0cm,clip,width=0.63\linewidth]{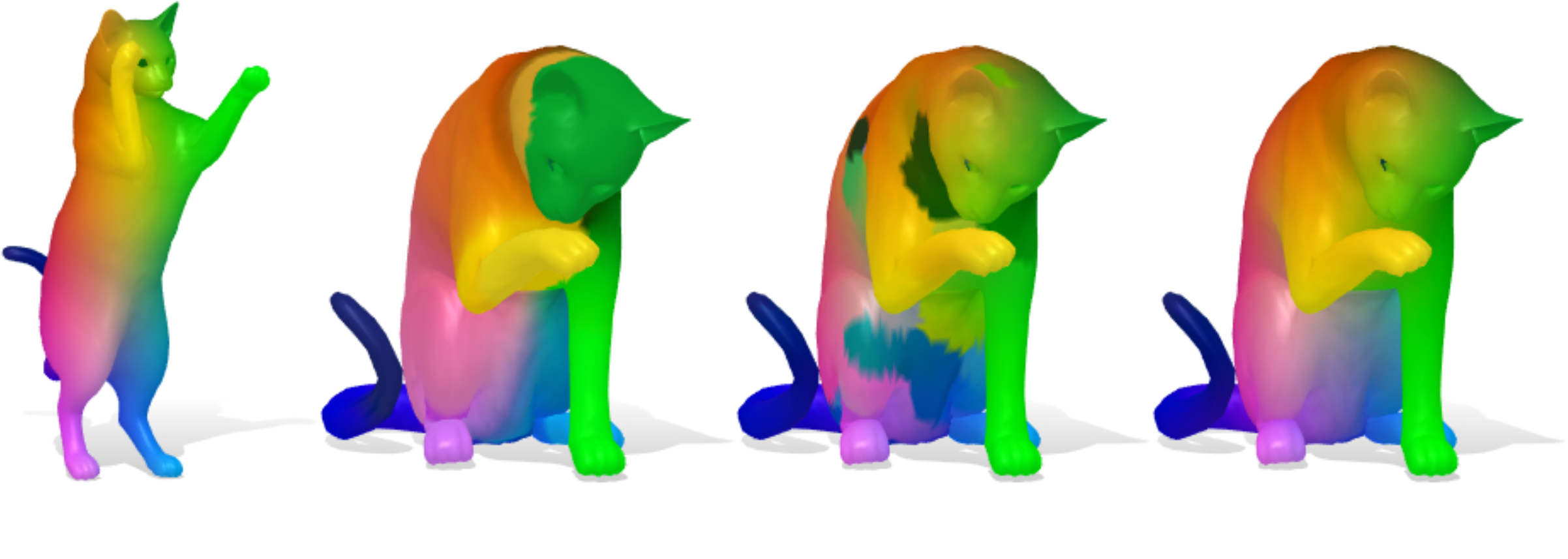}
\put(2,37.5){\footnotesize Source}
\put(33,40){\footnotesize Ini}
\put(26.5,35){\footnotesize ($10\times 10$)}
\put(55,40){\footnotesize Fmap$_{200}$}
\put(57,35){\footnotesize +ICP}
\put(85.4,40){\footnotesize Ini}
\put(75,35){\footnotesize +\textbf{Ours}$_{10 \ldots 200}$}
\end{overpic}
\vspace{-0.75cm}
\caption{\label{fig:res:eg:fMap_diff_size}Impact of the input functional map size. Given a pair of shapes, we use a fixed set of descriptors and the approach of \cite{nogneng17} to compute a functional map of size $k \times k$ and refine it with ICP. Alternatively, we compute a functional map of size $10\times 10$ using the same approach and upsample it to $k\times k$ using our method. Differently from the ICP baseline, our method leads to improvement as $k$ grows. On the right we show a qualitative illustration for $k=200$.
\vspace{-0.25cm}}
\end{figure}
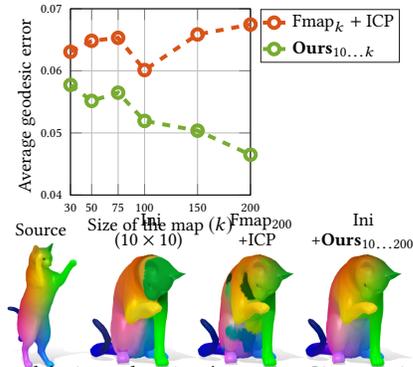

\subsection{Acceleration Strategies}
We propose three ways to accelerate \name.

\subsubsection{Larger step size}
The basic method increases the size of the functional map by one row and one column at each iteration. This choice is supported by our theoretical analysis below, which suggests that increasing by one at each iteration helps to promote isometric maps, when they are present. In practice our method also achieves good accuracy with larger increments ranging between $2$ and $5$ (see supplementary materials for an illustration). 
%
We also note that in some settings (e.g., in the context of partial correspondence or in
challenging non-isometric pairs), it is more reasonable to have rectangular functional maps with
more rows than columns. There, we increase the number of rows with higher increments than that of columns. We point out these explicitly in Section \ref{sec:results}.

\subsubsection{Approximate nearest neighbors} 
We can also use approximate nearest neighbor instead of exact nearest neighbors during upsampling. This is particularly useful in higher dimensions where such queries can become expensive. In practice, we have observed that using the FLANN library~\cite{flann} can lead to a 30x time improvement with negligible impact on final quality ($\sim$0.001\% decrease of average accuracy).

\subsubsection{Sub-sampling}
%
In the presence of perfect information, a functional map $\C$ of size $k \times k$ is fully determined by $k$ point correspondences. Thus, it is possible to sample a small number (typically a few hundred) points on each shape, perform our refinement using the spectral embedding of only those points, and then convert the final functional map to a {\em dense} pointwise correspondence only once. In practice we simply use Euclidean farthest point sampling starting from a random seed point.

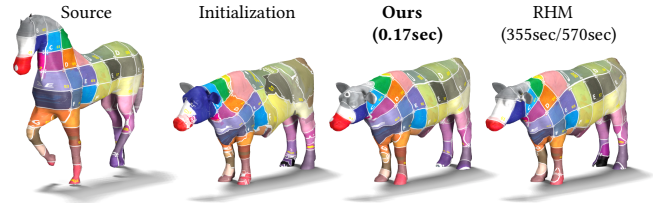
\begin{figure}[bt]
\centering
\begin{overpic}[trim=0cm 0cm 0cm 0cm,clip,width=0.999\linewidth]{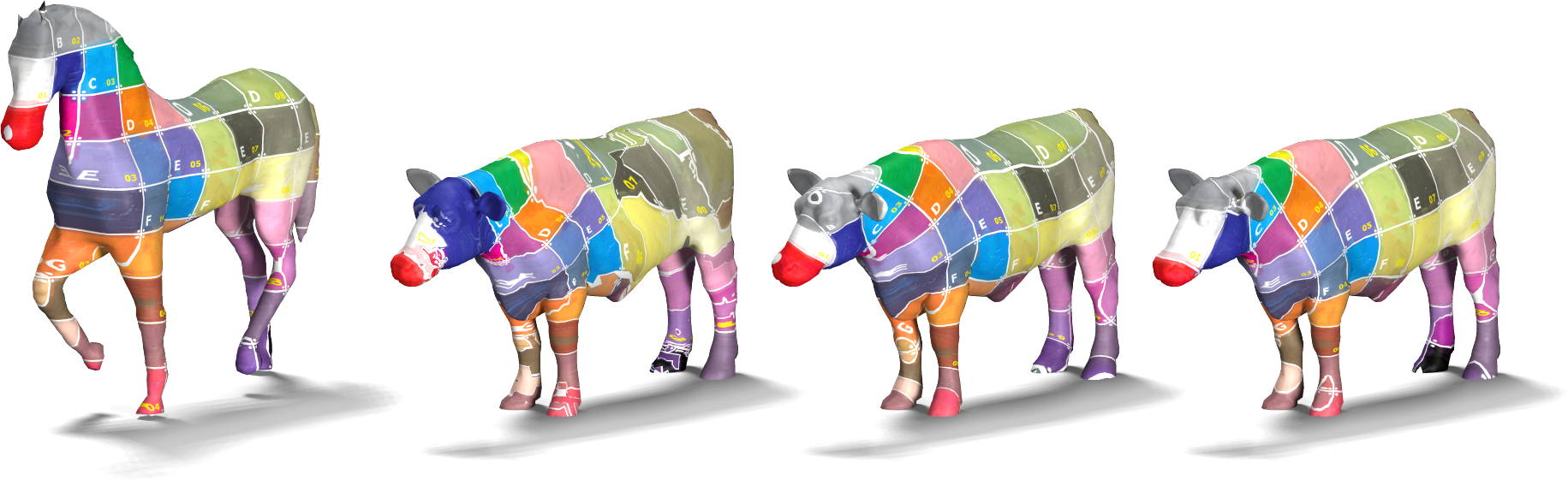}
\put(7.5,30){\footnotesize Source}
\put(29,30){\footnotesize Initialization}
\put(57.5,30){\footnotesize \textbf{Ours}}
\put(56,26.5){\footnotesize \textbf{(0.17sec)}}
\put(81,30){\footnotesize RHM}
\put(76,26.5){\footnotesize (355sec/570sec)}
\end{overpic}%
\vspace*{-0.45cm}
\caption{\label{fig:ezuz}Comparison with RHM \cite{ezuz2018reversible}. Both methods are initialized with a $17\times 10$ functional map provided by the authors of~\cite{ezuz2018reversible}. The reported runtimes (excluding pre-computation) are for a CPU implementation of our method with acceleration, and a (GPU/CPU) implementation of RHM. The runtime of pre-computation for our method is 7s (and 70s for RHM). Our solution has comparable quality and is more than 2 orders of magnitude faster.
}
\vspace{-0.2cm}
\end{figure}

\subsection{Relation to Other Techniques}
\label{sec:relation_others}
While closely related to multiple existing techniques, our method is
fundamentally different in several ways that we highlight below.

\paragraph{Iterative Closest Point} ICP refinement of functional
maps \cite{ovsjanikov2012functional} differs in that our method
progressively {\em increases} the dimension of the spectral embedding
during refinement. This crucial difference allows us to process
smaller initial functional maps, which are easier to compute, and
avoids getting trapped in local minima at higher dimensions,
significantly improving the final accuracy.
\revised{
Figure~\ref{fig:res:eg:tosca_iter} shows the accuracy of our method compared to ICP with different dimensions.
All methods in this figure refine the same initial pointwise map at \#iter = 1, which is computed
using~\cite{ren2018continuous} with the orientation-preserving term.
}
Moreover, differently from ICP our approach does not force the
singular values of functional maps to be 1, and inverts the direction
of the pointwise and functional maps in a way that is consistent with
the directions of a map and its
pull-back. 
\revised{As we show in Section~\ref{sec:theory}, rather than promoting area-preserving pointwise maps as done in ICP, our method implicitly optimizes an energy that promotes full isometries.}
%
In Figure \ref{fig:res:eg:fMap_diff_size} we further illustrate that
our method produces significantly more accurate maps in higher dimensions.
\revised{
We initialize the maps with the approach of \cite{nogneng17} using the WKS descriptors and 2 landmarks.
}


\paragraph{BCICP} \cite{ren2018continuous} is a recent powerful
technique for improving noisy correspondences, based on refining maps
in both the spectral and spatial domains, while incorporating
bijectivity, smoothness and coverage through a series of sophisticated
update steps. While often accurate, this method requires the
computation of geodesic distances, is inefficient, and suffers from
poor scalability. As an extension of the original ICP, this method
also uses spectral embeddings of {\em fixed} size. As we show in our
tests, our very simple approach can achieve similar and even superior
accuracy at a fraction of the time cost.

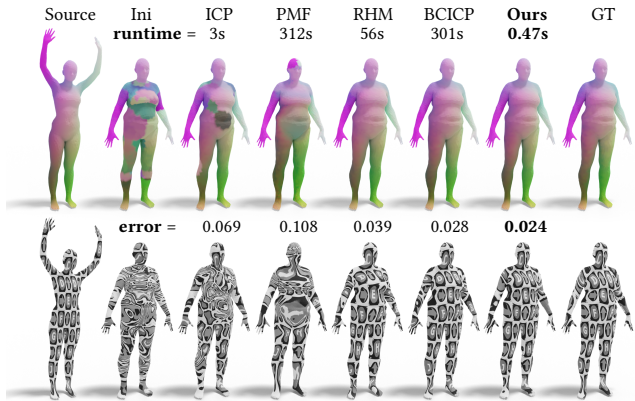
\begin{figure}[!t]
\centering
\begin{overpic}
  [trim=5cm 0cm 8cm 0cm,clip,width=1\columnwidth, grid=false]{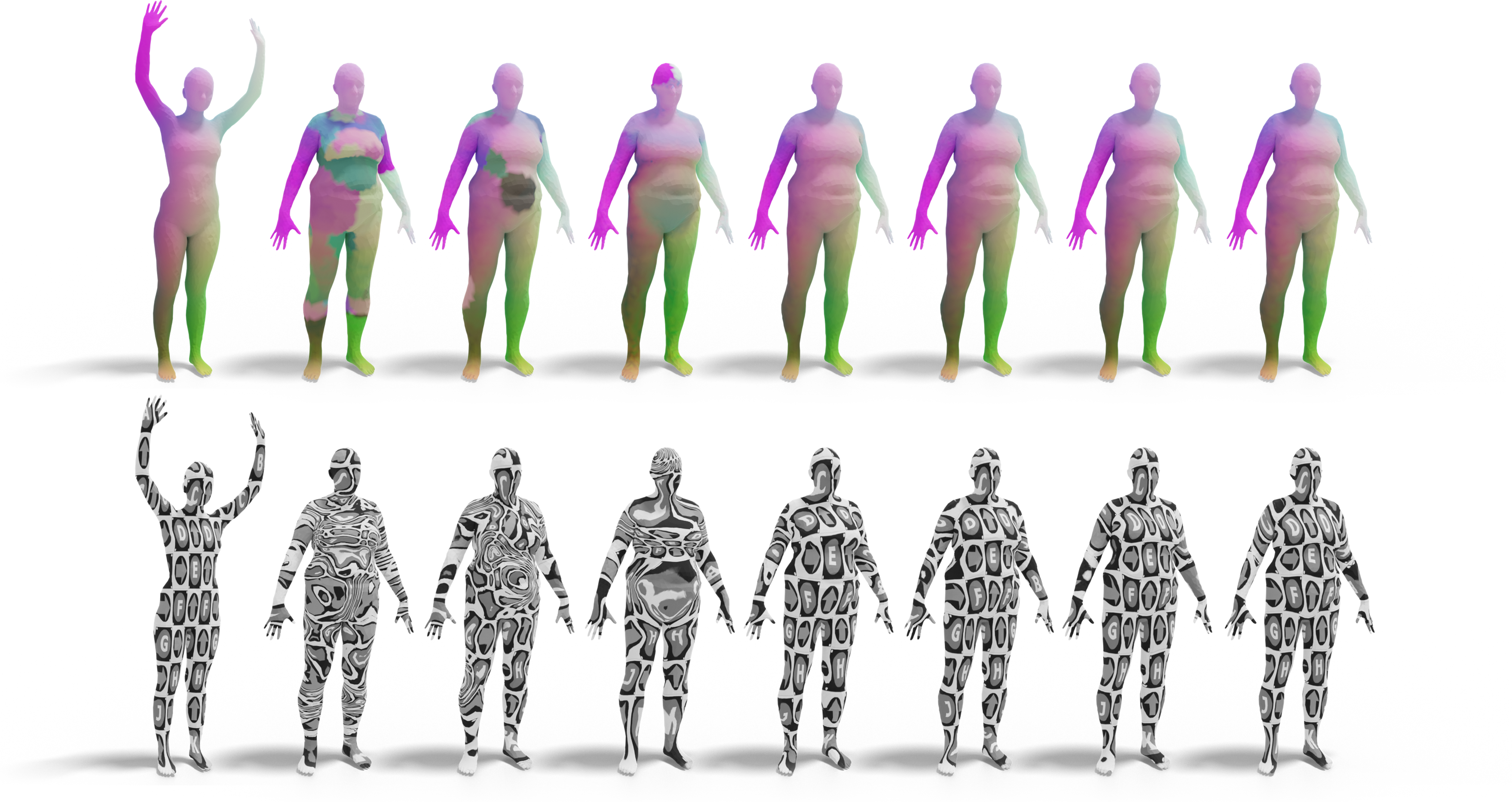}
  \put(6,61){\footnotesize Source}
  
  \put(19.4,61){\footnotesize Ini}
  \put(16.6,58){\footnotesize \textbf{runtime} =}
  \put(17.5,28){\footnotesize \textbf{error} = }
  
  \put(31,61){\footnotesize ICP}
  \put(31.6,58){\footnotesize 3s}
  \put(30.5,28){\footnotesize 0.069}
  
  \put(42,61){\footnotesize PMF}
  \put(42.5,58){\footnotesize 312s}
  \put(42.5,28){\footnotesize 0.108}
  
  \put(54,61){\footnotesize RHM}
  \put(55,58){\footnotesize 56s}
  \put(54,28){\footnotesize 0.039}
  
  \put(65,61){\footnotesize BCICP}
  \put(66,58){\footnotesize 301s}
  \put(66,28){\footnotesize 0.028}
  
  \put(77.9,61){\footnotesize \textbf{Ours}}
  \put(78,58){\footnotesize \textbf{0.47s}}
  \put(77.5,28){\footnotesize \textbf{0.024}}
  
  \put(91,61){\footnotesize GT}
  \end{overpic}
  \vspace{-0.8cm}
  \caption{Refinement example. Given the initialization computed from
    WKS descriptors, we compare our method with existing refinement
    techniques, by visualizing the maps via color transfer (first row)
    and texture transfer (second row). We also report the average
    error and the runtime for each method.  Note that our method is
    120x faster than RMH and 640x faster than BCICP, while resulting in
    lower error.}
\label{fig:res:eg:faust_refinement}
\vspace{-0.7cm}
\end{figure}

\paragraph{Reversible Harmonic Maps (RHM)} \cite{ezuz2018reversible}
is another recent approach for map refinement, based on minimizing the
bi-directional geodesic Dirichlet energy. In a similar spirit to ours,
this technique is based on splitting the alignment in a
higher-dimensional embedding space from the computation of pointwise
maps. However, it requires the computation of all pairs of geodesic
distances, and results in least squares problems with size
proportional to the number of points on the shapes. Furthermore,
similarly to ICP and BCICP, the embedding dimension is fixed
throughout the approach. As a result, our approach is significantly
more efficient (see Figure~\ref{fig:ezuz}), scalable, and, as we show
below, more accurate in many cases.

\paragraph{Spatial refinement methods} Spatial refinement methods 
such as PMF \cite{vestner2017product,vestner2017efficient} operate via an
alternating diffusion process based on solving a sequence of
linear assignment problems; this approach demonstrates high accuracy in
challenging cases, but is severely limited by mesh resolution. Other
approaches formulate shape correspondence by seeking for optimal
transport plans iteratively via Sinkhorn projections, but they either
scale poorly \cite{solomon2016entropic} or can have issues with
non-isotropic meshes~\cite{mandad2017variance}. Interestingly,
although fundamentally different, a link exists between \name\
and PMF that we describe in the supplementary materials. \\

In Figure~\ref{fig:res:eg:faust_refinement} we show qualitative
comparisons with the methods above on pairs of remeshed shapes from
the FAUST \cite{FAUST} dataset. We provide a more complete evaluation
with state-of-the-art refinement methods in Section \ref{sec:results}.

\subsection{Derivation and Analysis}
\label{sec:theory}
In this section we provide a theoretical justification for our method by first formulating a variational optimization problem and
then arguing that \name\ provides an efficient way of solving it.


\subsubsection{Optimization Problem}
%
We consider the following problem:
\begin{align}
  \label{eq:ourprob}
  \min_{\C \in \mathcal{P}} E(\C), \text{ where } E(\C) = \sum_k \frac{1}{k} \big\Vert \C_k^T \C_k - I_k \big\Vert^2_F\,. 
\end{align}
Here $\mathcal{P}$ is the set of functional maps arising from
pointwise correspondences, $\C_k$ is the principal $k \times k$
submatrix of $\C$ \revised{(i.e., the submatrix of $\C$ consisting of the first $k$ rows and columns)}, and $I_k$ is an identity matrix of size $k$. In
other words, Eq. \eqref{eq:ourprob} aims to compute a pointwise map
associated with a functional map in which every principal submatrix
is orthonormal.

The energy in Eq. \eqref{eq:ourprob} is different from the commonly
used penalty promoting orthonormal functional maps
\cite{ovsjanikov2012functional,kovnatsky2013coupled,kovnatsky2016madmm}
in two ways: first we explicitly constrain $\C$ to arise from a
point-to-point map, and second we enforce orthonormality of every
principal submatrix rather than the full functional map of a given
fixed size. Indeed, an orthonormal functional map corresponds to only
a \emph{locally area-preserving} point-to-point correspondence
\cite{rustamov2013map}. Instead, the energy in
Eq. \eqref{eq:ourprob} is much stronger and promotes complete 
isometries as guaranteed by the following theorem (proved in
Appendix A):
\begin{theorem}
  \label{thm:energy}
  Given a pair of shapes whose Laplacian matrices have the same
  eigenvalues, none of which are repeating, a functional map
  $\C \in \mathcal{P}$ satisfies $E(\C) = 0$ if and only if the
  corresponding pointwise map is an isometry.
\end{theorem}

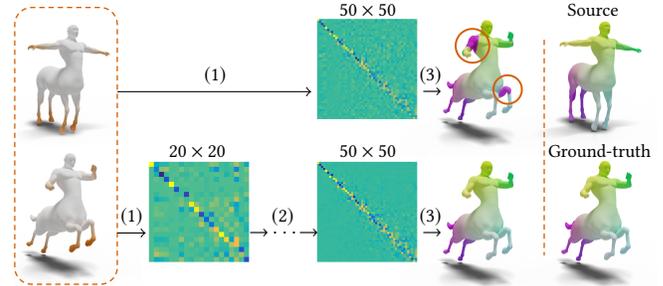
\begin{figure}[!t]
\centering
  \begin{overpic}
  [trim=0cm 0cm 0cm 0cm,clip,width=1\linewidth, grid = false]{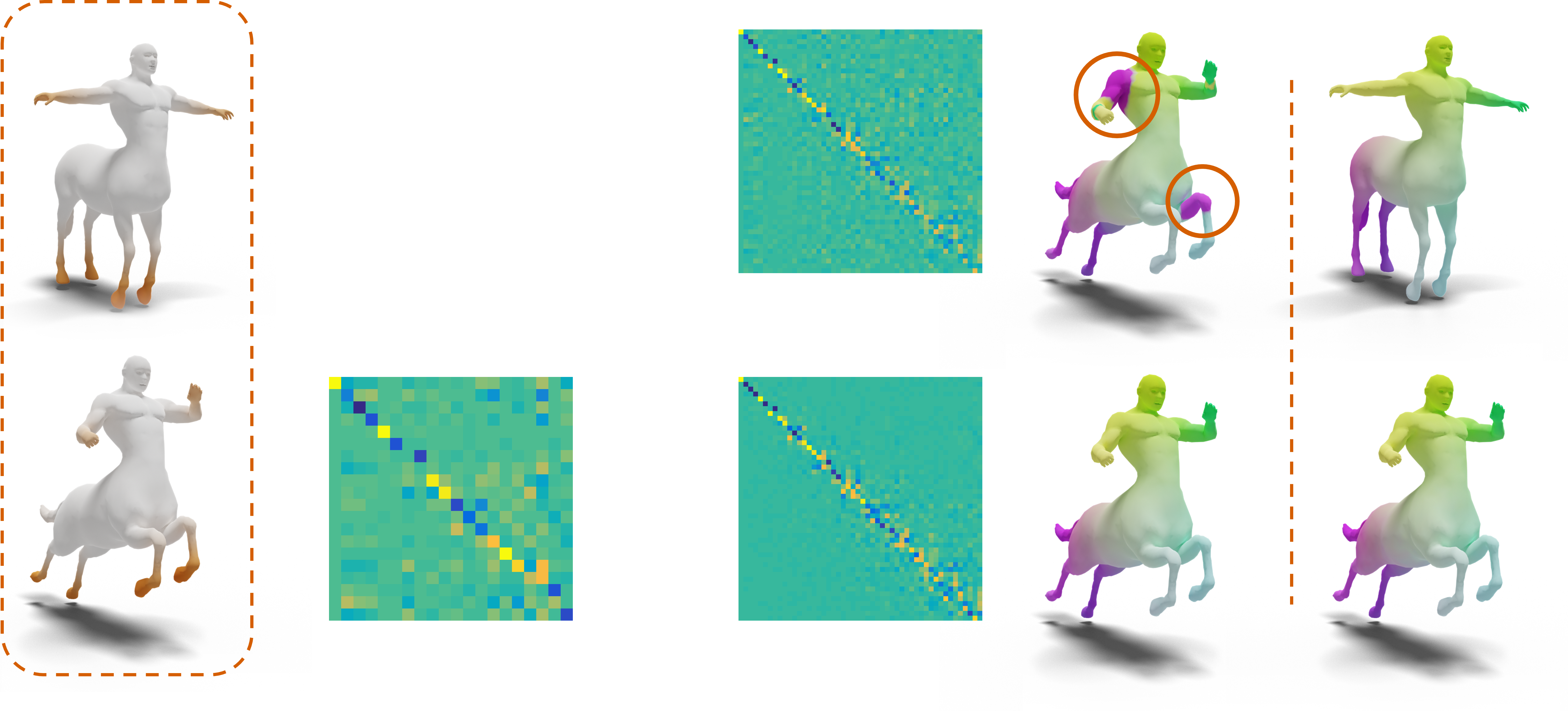} 
  
  \put(29.5,34){\footnotesize $(1)$}
  \put(12.5,31){
    \begin{tikzpicture}[node distance=2cm]
    \node (A) at (0,0) {};
    \node (B) at (2.8,0) {};
    \draw[->](A)--(B);
    (A) edge (B);
    \end{tikzpicture}
  }
  \put(16.5,12){\footnotesize $(1)$}
  \put(12.5,9){
    \begin{tikzpicture}[node distance=2cm]
    \node (A) at (0,0) {};
    \node (B) at (0.6,0) {};
    \draw[->](A)--(B);
    (A) edge (B);
    \end{tikzpicture}
  }
  \put(40,12){\footnotesize $(2)$}
  \put(36.5,9.5){\footnotesize $\to\!\cdots\!\to$}
  \put(63,12){\footnotesize $(3)$}
  \put(63.4,9.5){\footnotesize $\to$}
  \put(63,34){\footnotesize $(3)$}
  \put(63.4,31.5){\footnotesize $\to$}
  \put(24,22){\footnotesize $20\times 20$}
  \put(50.5,22){\footnotesize $50\times 50$}
  \put(50.5,44){\footnotesize $50\times 50$}
  \put(86, 44){\footnotesize Source}
  \put(83, 22){\footnotesize Ground-truth}
 
  \end{overpic}
  \vspace*{-0.8cm}
\caption{\label{fig:res:eg:fMap_size}
We use an existing functional map pipeline (1) to compute either a $50\times 50$ (top row) or $20\times 20$ (bottom row) functional map using the same input descriptors. We then upsample (2) the smaller map to also have size 50x50 using our technique, and convert both to pointwise maps (3). Our approach leads to better results as can be seen, e.g., on the arms and legs.\vspace{-0.2cm}}
\end{figure}

To derive \name\ as a method to solve the optimization problem in Eq. \eqref{eq:ourprob} we first consider a single term
inside the sum, and write the problem explicitly in terms of the binary matrix $\Pi$ representing the pointwise map:
\begin{align}
  \label{eq:ourprob2}
  &\min_{\Pi} \| \C_k^T \C_k - I_k \|^2_F = \min_{\Pi} \| \C_k \C_k^T - I_k \|^2_F\,, \\ &\text{ where } \C_k =
  (\Phi^{k}_{\M})^{+} \Pi \Phi^k_{\N}\,.
\end{align}
This problem is challenging due to the constraints on $\Pi$. To address
this, we use \emph{half-quadratic splitting,} by decoupling $\Pi$ and
$\C_k$. This leads to the following two separate sub-problems:
\begin{align}
  \label{eq:ourprob3}
  \min_{\Pi} \|   (\Phi^{k}_{\M})^{+} \Pi \Phi^k_{\N} \C_k^T - I_k
  \|^2_F\,, \\
  \label{eq:ourprobck}
  \min_{\C_k} \|   \C_k - (\Phi^{k}_{\M})^{+} \Pi \Phi^k_{\N} \|^2_F\,.
\end{align}

Now we remark that Eq. \eqref{eq:ourprob3} does not fully constrain
$\Pi$ since it only penalizes the image of $\Pi$ within the vector
space of $\Phi^{k}_{\M}$. Instead, inspired by a related construction
in \cite{ezuz2017deblurring} we add a regularizer
$\mathcal{R}(\Pi) = \| (I - \Phi^{k}_{\M} (\Phi^{k}_{\M})^{+}) \Pi
\Phi^k_{\N} \C_K^T \|_{A_{\M}}^2$, where we use the weighted matrix
norm $\|X\|^2_ {A_\M} = tr(X^T A_\M X)$ and $A_\M$ is the area matrix
of shape $\M$. This regularizer penalizes the image of
$\Pi \Phi^k_{\N} \C_k^T$ that lies outside of the span of
$\Phi^{k}_{\M}$, which intuitively means that no spurious
high frequencies should be created. Finally, it can be shown (see
proof in the appendix) that solving Eq. \eqref{eq:ourprob3} with the
additional term $\mathcal{R}(\Pi)$ is equivalent to solving:
\begin{align}
  \label{eq:ourprob4}
  \min_{\Pi} \|   \Pi \Phi^k_{\N} \C_k^T - \Phi^{k}_{\M} \|^2_F\,.
\end{align}

The problem in Eq. \eqref{eq:ourprob4} has a closed-form solution, which reduces to the
nearest-neighbor search described in Eq. \eqref{eq:p2p_recovery}
above. Moreover, the problem in Eq. \eqref{eq:ourprobck} is solved simply via 
$\C_k = (\Phi^{k}_{\M})^{+} \Pi \Phi^k_{\N}$ since the minimization is unconstrained.

Finally, in this derivation we assumed a specific value of $k$. In
practice we start with a particular value $k_0$ and progressively
increase it. This is motivated by the fact that if a principal
$k \times k$ submatrix is orthonormal, it provides a very strong
initialization for the larger problem on a $(k+1) \times (k+1)$ matrix
since only a single new constraint on the additional row and column
must be enforced. This leads to our method \name:

\begin{enumerate}
  \item Given $k = k_0$ and an initial $\C_0$ of size $k_0 \times k_0$.
  \item Compute $\argmin_{\Pi} \|   \Pi \Phi^k_{\N} \C_k^T - \Phi^{k}_{\M} \|^2_F.$
  \item Set $k = k+1$ and compute $\C_k =  (\Phi^{k}_{\M})^{+} \Pi \Phi^k_{\N}$.
  \item Repeat the previous two steps until $k = k_{\max}$.
\end{enumerate}


\subsubsection{Empirical Accuracy}
We demonstrate that this simple procedure is remarkably efficient in
minimizing the energy in Eq. \eqref{eq:ourprob}. 
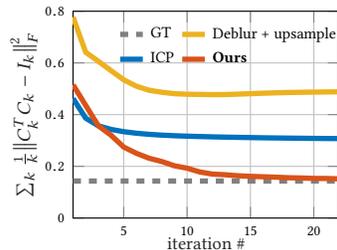
\begin{wrapfigure}{r}{0.46\linewidth}
  \vspace{-0.55cm}
  \begin{center}
  \hspace{-0.8cm}
%
%
\definecolor{mycolor1}{rgb}{0.50196,0.50196,0.50196}%
\definecolor{mycolor2}{rgb}{0.92941,0.69412,0.12549}%
\definecolor{mycolor3}{rgb}{0.00000,0.44706,0.74118}%
\definecolor{mycolor4}{rgb}{0.85098,0.32549,0.09804}%
\pgfplotsset{
compat=1.11,
legend image code/.code={
\draw[mark repeat=2,mark phase=2]
plot coordinates {
(0cm,0cm)
(0.15cm,0cm)        
(0.3cm,0cm)         
};%
}
}
\begin{tikzpicture}

\begin{axis}[%
width=0.9\linewidth,
height=0.7\linewidth,
at={(1.713in,1.987in)},
scale only axis,
every x tick label/.append style={font=\color{black}, font=\tiny},
every y tick label/.append style={font=\color{black}, font=\tiny},
xmin=1,
xmax=22,
xlabel style={font=\color{white!15!black}},
xlabel style={at={(0.5,-0.05)}},
xlabel={\footnotesize iteration \#},
ymin=0,
ymax=0.8,
xmajorgrids,
ymajorgrids,
ylabel style={font=\color{white!15!black}},
ylabel style={at={(-0.1,0.5)}},
ylabel={\footnotesize $\sum_k \frac{1}{k}\big\Vert C_k^TC_k - I_k \big\Vert_F^2$},
axis background/.style={fill=white},
legend style={at={(1.04,0.68)}, anchor=south east, legend cell align=left, align=left, draw=none, fill=white, fill opacity=0, draw opacity=1, text opacity=1,legend columns=2}
]
\addplot [color=mycolor1, dashed, line width=2.0pt]
  table[row sep=crcr]{%
1	0.14319118\\
2	0.14319118\\
3	0.14319118\\
4	0.14319118\\
5	0.14319118\\
6	0.14319118\\
7	0.14319118\\
8	0.14319118\\
9	0.14319118\\
10	0.14319118\\
11	0.14319118\\
12	0.14319118\\
13	0.14319118\\
14	0.14319118\\
15	0.14319118\\
16	0.14319118\\
17	0.14319118\\
18	0.14319118\\
19	0.14319118\\
20	0.14319118\\
21	0.14319118\\
22	0.14319118\\
};
\addlegendentry{\scriptsize GT}

\addplot [color=mycolor2, line width=2.0pt]
  table[row sep=crcr]{%
1	0.7770559\\
2	0.6412428\\
3	0.6062253\\
4	0.5701211\\
5	0.5347045\\
6	0.5099887\\
7	0.4943626\\
8	0.4865527\\
9	0.4810209\\
10	0.4788942\\
11	0.4781905\\
12	0.4773112\\
13	0.4775495\\
14	0.4792343\\
15	0.4811165\\
16	0.4829124\\
17	0.4843963\\
18	0.4856162\\
19	0.4864407\\
20	0.4871846\\
21	0.4877012\\
22	0.4883532\\
};
\addlegendentry{\scriptsize Deblur + upsample}

\addplot [color=mycolor3, line width=2.0pt]
  table[row sep=crcr]{%
1	0.4625489\\
2	0.3863334\\
3	0.3567456\\
4	0.3428626\\
5	0.3340743\\
6	0.3281089\\
7	0.3238292\\
8	0.3207212\\
9	0.3184527\\
10	0.3166507\\
11	0.3153272\\
12	0.3140464\\
13	0.3130412\\
14	0.3121059\\
15	0.3112864\\
16	0.3104206\\
17	0.3097965\\
18	0.309195\\
19	0.3086702\\
20	0.3081268\\
21	0.3077185\\
22	0.3074144\\
};
\addlegendentry{\scriptsize ICP}

\addplot [color=mycolor4, line width=2.0pt]
  table[row sep=crcr]{%
1	0.5152937\\
2	0.4358049\\
3	0.356546\\
4	0.3203655\\
5	0.2745724\\
6	0.2504646\\
7	0.2314047\\
8	0.2186986\\
9	0.20220259\\
10	0.19277947\\
11	0.17824983\\
12	0.16997461\\
13	0.16753245\\
14	0.16356839\\
15	0.16129899\\
16	0.15967849\\
17	0.15801394\\
18	0.15594807\\
19	0.15498768\\
20	0.15327586\\
21	0.15304166\\
22	0.15168794\\
};
\addlegendentry{\scriptsize \textbf{Ours}}

\end{axis}
\end{tikzpicture}%
  \end{center}
  \vspace{-0.4cm}
  \caption{Value of $E(\C)$ across iterations}
  \label{fig:mtd:orthoErr_over_iter}
  \vspace{-0.35cm}
\end{wrapfigure}
For this in Figure~\ref{fig:mtd:orthoErr_over_iter} we plot the value
of the energy during the iterations of \name\ from $k=20$ to $k=120$
with step $5$ on 100 pairs of shapes from the FAUST dataset, and
compare it to the ICP refinement using $k=120$. We also evaluate
a method in which we perform the same iterative spectral upsampling as
in \name\ but use the pointwise map recovery from
\cite{ezuz2017deblurring} instead of Eq. \eqref{eq:p2p_recovery}. For all methods, at
every iteration we convert the computed pointwise map to a functional
map $\C$ of fixed size $120\times120$ and report $E(\C)$.  Our
approach results in maps with energy very close to the ground truth,
while Deblur with upsampling performs poorly, highlighting the
importance of the adapted pointwise recovery method. In the supplementary materials we further detail the differences between the two methods and their relation to PMF.

Finally, in Figure \ref{fig:res:eg:fMap_size} we also show the result of an existing functional map estimation pipeline with orientation preservation \cite{ren2018continuous} for a map of size $50\times 50$ with careful parameter tuning for optimality, which nevertheless leads to noise in the final point-to-point map. Initializing the map to  size $20\times 20$ using exactly the same descriptors and up-sampling it to a larger size with our approach leads to a significant improvement.

\section{Results}
\label{sec:results}
We conducted an extensive evaluation of our method, both in terms of its empirical properties (Section~\ref{sec:perf}) and in relation to existing methods, as we showcase across several applications (Section~\ref{sec:apps}).

\begin{figure}[!t]
\vspace{-0.3cm}
\definecolor{mycolor1}{rgb}{0.00000,0.44700,0.74100}%
\definecolor{mycolor2}{rgb}{0.85000,0.32500,0.09800}%
\definecolor{mycolor3}{rgb}{0.92900,0.69400,0.12500}%
\pgfplotsset{scaled x ticks=false}
\pgfplotsset{minor grid style={dashed}}
\pgfplotsset{
compat=1.11,
legend image code/.code={
\draw[mark repeat=2,mark phase=2]
plot coordinates {
(0cm,0cm)
(0.0cm,0cm)        
(0.3cm,0cm)         
};%
}
}
\begin{tikzpicture}
\begin{axis}[
width=0.51\linewidth,
height=0.45\linewidth,
every x tick label/.append style={font=\color{black}, font=\tiny},
every y tick label/.append style={font=\color{black}, font=\tiny},
xmin = 0, xmax = 20,
ymin = 0, ymax = 100,
xminorgrids,
yminorgrids,
minor ytick={0,50, 100},
ytick={0,50,100},
yticklabels={0,0.05,0.1},
minor xtick={0,5,10,15,20},
xlabel = {\footnotesize Mesh resolution ($\times 10^3$)},
ylabel = {\footnotesize Average error},
ylabel style={at={(-0.12,0.48)}},
legend style={at={(1,1)}, legend columns=1, fill opacity=0.8, text opacity = 1, draw opacity=1, legend columns=3}
]
\addplot [color=mycolor1, line width=1.5pt]
  table[row sep=crcr]{%
1	62.7\\
2	83.4\\
5	76.4\\
10	74.7\\
15	60.1\\
20	58.3\\
};

\addplot [color=mycolor2,line width=1.5pt]
  table[row sep=crcr]{%
1	59.3\\
2	57.5\\
5	57.5\\
10	57.3\\
15	57.2\\
20	57.2\\
};

\addplot [color=mycolor3, line width=1.5pt]
  table[row sep=crcr]{%
1	27.7\\
2	27.4\\
5	26.4\\
10	26.5\\
15	26.2\\
20	26.1\\
};

\end{axis}
\end{tikzpicture}\definecolor{mycolor1}{rgb}{0.00000,0.44700,0.74100}%
\definecolor{mycolor2}{rgb}{0.85000,0.32500,0.09800}%
\definecolor{mycolor3}{rgb}{0.92900,0.69400,0.12500}%
\pgfplotsset{scaled x ticks=false}
\pgfplotsset{minor grid style={dashed}}
\pgfplotsset{
compat=1.11,
legend image code/.code={
\draw[mark repeat=2,mark phase=2]
plot coordinates {
(0cm,0cm)
(0.0cm,0cm)        
(0.3cm,0cm)         
};%
}
}
\begin{tikzpicture}
\begin{axis}[
width=0.51\linewidth,
height=0.45\linewidth,
every x tick label/.append style={font=\color{black}, font=\tiny},
every y tick label/.append style={font=\color{black}, font=\tiny},
xmin = 0, xmax = 20,
ymin = -0.01, ymax = 500,
xminorgrids,
yminorgrids,
minor ytick={0,250,500},
ytick={0,250,500},
minor xtick={0,5,10,15,20},
xlabel = {\footnotesize Mesh resolution ($\times 10^3$)},
ylabel = {\footnotesize Runtime (s)},
ylabel style={at={(-0.12,0.48)}},
legend style={at={(0.55,0.95)}, legend columns=1, fill opacity=0.8, text opacity = 1, draw opacity=1, legend columns=1}
]
\addplot [color=mycolor1, line width=1.5pt]
  table[row sep=crcr]{%
1	0.249\\
2	0.423\\
5	1.082\\
10	2.169\\
15	3.206\\
20	4.033\\
};
\addlegendentry{\scriptsize ICP}

\addplot [color=mycolor2, line width=1.5pt]
  table[row sep=crcr]{%
1	31.58\\
2	31.31\\
5	59.79\\
10	130.2\\
15	272.9\\
20	472.1\\
};
\addlegendentry{\scriptsize BCICP}

\addplot [color=mycolor3, line width=1.5pt]
  table[row sep=crcr]{%
1	1.316\\
2	2.353\\
5	5.64\\
10	11.05\\
15	17.39\\
20	22.32\\
};
\addlegendentry{\scriptsize \textbf{Ours}}
\end{axis}
\end{tikzpicture}
\vspace{-0.3cm}
\caption{Scalability and accuracy test on 6 pairs of scanned bones. The source shape has 5K vertices. We compare to ICP and BCICP on the same target shape with different resolution (ranging from 1K to 20K vertices).}
\label{fig:res:eval:scalability}
\vspace{-0.1cm}
\end{figure}
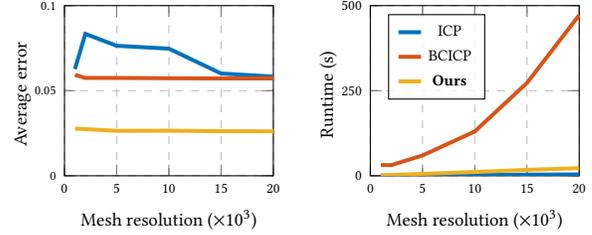

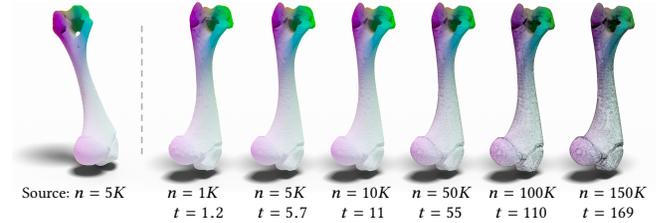
\begin{figure}[!t]
\centering
\begin{overpic}
  [trim=0cm 2cm 2cm 0cm,clip,width=1\columnwidth,grid=false]{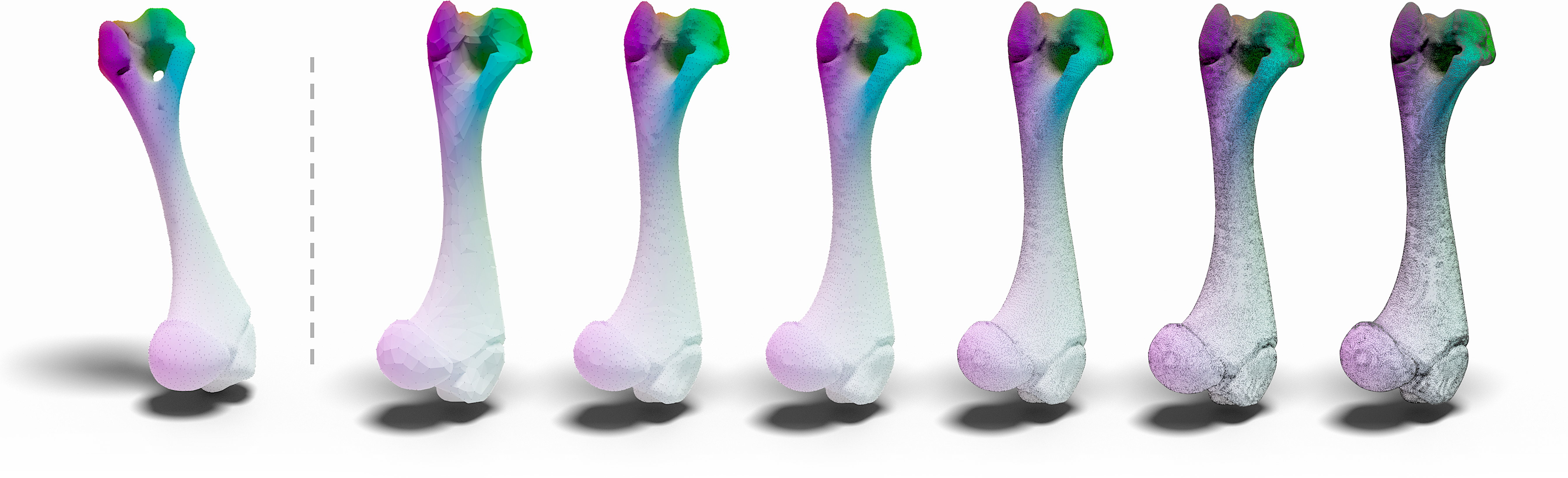}

  \put(1.5,-1){\scriptsize Source: $n=5K$}
  \put(24,-1){\scriptsize $n=1K$}
  \put(24.9,-4){\scriptsize $t=1.2$}
  \put(37.5,-1){\scriptsize $n=5K$}
  \put(38,-4){\scriptsize $t=5.7$}
  \put(49.5,-1){\scriptsize $n=10K$}
  \put(51,-4){\scriptsize $t=11$}
  \put(62,-1){\scriptsize $n=50K$}
  \put(63,-4){\scriptsize $t=55$}
  \put(74,-1){\scriptsize $n=100K$}
  \put(75,-4){\scriptsize $t=110$}
  \put(88,-1){\scriptsize $n=150K$}
  \put(88.5,-4){\scriptsize $t=169$}
  \end{overpic}
\caption{Scalability. The vertices of the bone shapes are colored black to show the resolution (zoom in for better view), while RGB colors encode the computed map, via pull-back from the source. The corresponding runtime for our upsampling, from $5\times 5$ to $50 \times 50$ \emph{without any acceleration}, is reported below each shape (in seconds).}
\label{fig:res:eg:bones}
\vspace{-0.3cm}
\end{figure}

\subsection{Performance of \name}\label{sec:perf}
We start by showing an evaluation of scalability, as well as of the stability and smoothness of our method.

\subsubsection{Scalability}
In Figure~\ref{fig:res:eval:scalability} we assess the scalability of our method using shapes of humerus bones of wild boars acquired using a 3D sensor. Each bone was scanned independently, and the ground truth was provided by domain experts as 24 consistent landmarks 
 \cite{gunz2013semilandmarks} on each shape. We show the average runtime and accuracy over 6 maps w.r.t. different target mesh resolution; 
the input descriptors (one landmark point and WKS descriptors \cite{aubry2011wave}) for the initialization are fixed. While BCICP, the current state-of-the-art method, quickly becomes prohibitively expensive at high resolution, both ICP and \name\ without acceleration have approximately linear complexity. On the other hand, the accuracy for BCICP and our method are stable w.r.t. different resolutions, while ICP is less accurate and more unstable. 
Figure~\ref{fig:res:eg:bones} also shows an example with meshes having $\sim$150K vertices.

%
Figure~\ref{fig:res:eval:subsampling} shows the results of our sub-sampling strategy for acceleration on one pair of bones, where the source has 20K vertices, and the target has 5K vertices. 
The corresponding runtime (blue curve) and average error (red curve) w.r.t. different sampling size for the source shape are reported. We can see that around 100 samples on a 20K mesh are enough to produce a refined map with similar quality to that of our method without sampling (whose average is shown as dashed black).

\begin{figure}
\centering
    \begin{minipage}{0.45\linewidth}
    \begin{center}
    \hspace{-0.5cm}
    \definecolor{mycolor1}{rgb}{0.00000,0.44700,0.74100}%
\definecolor{mycolor2}{rgb}{0.85000,0.32500,0.09800}%
\pgfplotsset{scaled x ticks=false}
\pgfplotsset{major grid style={dashed, mycolor2}}
\pgfplotsset{minor grid style={dashed}}
\pgfplotsset{
compat=1.11,
legend image code/.code={
\draw[mark repeat=2,mark phase=2]
plot coordinates {
(0cm,0cm)
(0.0cm,0cm)        
(0.3cm,0cm)         
};%
}
}
\begin{tikzpicture}
\begin{axis}[
width=1.2\linewidth,
height=1\linewidth,
xmin = 0, xmax = 500,
ymin = 0, ymax = 0.8,
axis y line* = left, 
every x tick label/.append style={font=\color{black}, font=\tiny},
every y tick label/.append style={font=\tiny\color{mycolor2}},
every y tick/.style={mycolor2},
xtick={0,100,200,300,400,500},
xlabel = {\footnotesize\# Samples},
xlabel style={at={(0.5,-0.1)}},
ylabel = {\footnotesize\textcolor{mycolor2}{Average error}},
ylabel style={at={(-0.12,0.48)}},
legend style={at={(0.72,0.95)}, legend columns=1, fill opacity=0.8, text opacity = 1, draw opacity=1}
]
\addlegendimage{only marks, mark=*,  mark options={solid, mycolor1}}
\addlegendimage{only marks, mark=*,  mark options={solid, mycolor2}}
\addlegendimage{line width=1pt, dashed, color=black}
 \addplot [color=mycolor2, line width=1pt, mark=*, mark options={solid, mycolor2},mark size=1pt]
  table[row sep=crcr]{%
5	0.476961570326499\\
10	0.685103858232299\\
15	0.424273746669431\\
20	0.430586863498764\\
25	0.553832488241968\\
30	0.55314081971706\\
35	0.422025978081773\\
40	0.446484815775622\\
45	0.422025978081773\\
50	0.441102834012709\\
55	0.093875562561079\\
60	0.0553254073548927\\
65	0.0580380881604351\\
70	0.0474307193895211\\
75	0.0425795469289073\\
80	0.0502029326314634\\
85	0.0402023792494782\\
90	0.0352456702781713\\
95	0.0317482691155184\\
100	0.0354765127230493\\
110	0.0360514004512072\\
120	0.0362375868321359\\
130	0.0327527667717354\\
140	0.032764977176903\\
150	0.0365875185780502\\
160	0.0280904900134824\\
170	0.0271593261752306\\
180	0.0256376126202158\\
190	0.0241227728856653\\
200	0.0266464329879365\\
200	0.0266464329879365\\
250	0.0248307332382092\\
300	0.0212840241770548\\
350	0.022824980884585\\
400	0.0228096008856198\\
450	0.0223580944113232\\
500	0.0222747369802069\\
};

\addplot [color=black, dashed, line width=1pt]
  table[row sep=crcr]{%
5	0.0239296102766424\\
10	0.0239296102766424\\
15	0.0239296102766424\\
20	0.0239296102766424\\
25	0.0239296102766424\\
30	0.0239296102766424\\
35	0.0239296102766424\\
40	0.0239296102766424\\
45	0.0239296102766424\\
50	0.0239296102766424\\
55	0.0239296102766424\\
60	0.0239296102766424\\
65	0.0239296102766424\\
70	0.0239296102766424\\
75	0.0239296102766424\\
80	0.0239296102766424\\
85	0.0239296102766424\\
90	0.0239296102766424\\
95	0.0239296102766424\\
100	0.0239296102766424\\
110	0.0239296102766424\\
120	0.0239296102766424\\
130	0.0239296102766424\\
140	0.0239296102766424\\
150	0.0239296102766424\\
160	0.0239296102766424\\
170	0.0239296102766424\\
180	0.0239296102766424\\
190	0.0239296102766424\\
200	0.0239296102766424\\
200	0.0239296102766424\\
250	0.0239296102766424\\
300	0.0239296102766424\\
350	0.0239296102766424\\
400	0.0239296102766424\\
450	0.0239296102766424\\
500	0.0239296102766424\\
};

\end{axis}
\begin{axis}[
width=1.2\linewidth,
height=1\linewidth,
     xmin = 0, xmax = 500,
     ymin = 0.7, ymax = 1.4,
     hide x axis,
     hide y axis]
\addplot [color=mycolor1, line width=1pt, mark=*, mark options={solid, mycolor1},mark size=1pt]
  table[row sep=crcr]{%
5	0.808810805781423\\
10	0.869764347848257\\
15	0.777591506469847\\
20	0.790248757734463\\
25	0.802769920204286\\
30	0.802078880509071\\
35	0.789994791540773\\
40	0.798469050726022\\
45	0.827588410532544\\
50	0.776666036441997\\
55	0.801351749138353\\
60	0.811637214424618\\
65	0.785080694695219\\
70	0.810261426243681\\
75	0.814488125673201\\
80	0.813496101271255\\
85	0.815681799895169\\
90	0.837235814729841\\
95	0.823259065053424\\
100	0.85276086417378\\
110	0.843581327292889\\
120	0.842003227059353\\
130	0.858761684679945\\
140	0.84020791454021\\
150	0.812741818364835\\
160	0.817512541857238\\
170	0.873807608853652\\
180	0.870598760896623\\
190	0.889918071896604\\
200	0.858291499549516\\
200	0.859113992411477\\
250	0.95507712195046\\
300	0.988957933993292\\
350	1.0281627648475\\
400	1.12908833421024\\
450	1.2547578927362\\
500	1.33499266742978\\
};

\end{axis}
\pgfplotsset{every axis y label/.append style={rotate=180,yshift=0cm}}
\pgfplotsset{major grid style={dashed, mycolor1}}
\begin{axis}[
width=1.2\linewidth,
height=1\linewidth,
every x tick label/.append style={font=\color{black}, font=\tiny},
every y tick label/.append style={font=\tiny\color{mycolor1}},
every y tick/.style={mycolor1},
xmin=0, xmax=500,
ymin=0.7, ymax=1.4,
hide x axis,
axis y line*=right,
ylabel={\footnotesize\textcolor{mycolor1}{Time (s)}},
ylabel style={at={(1.08,0.48)}}
]
\end{axis}
\end{tikzpicture}
    \vspace{-0.5cm}
    \caption{\label{fig:res:eval:subsampling}Acceleration by sampling}
    \end{center}
    \end{minipage}\hspace{0.55cm}
    \begin{minipage}{0.44\linewidth}
    \captionsetup{justification=centering,margin=0.2cm}
    \begin{center}
    \vspace{-0.5cm}
    \input{figures2/eg_stability_inset.tex}
    \caption{\label{fig:res:eval:stability}Stability of zoomOut}
    \end{center}
    \vspace{-0.5cm}
    \end{minipage}
    \vspace{-0.3cm}
\end{figure}


\subsubsection{Stability}
We also evaluate the stability of our method w.r.t. noise in the initial functional map. Here we test on a single shape pair from FAUST initialized using the approach of \cite{ren2018continuous} while fixing the size of the computed functional map to 4. Given this $4\times 4$ initial functional map, we add white noise to it and use our method to refine the map. Figure~\ref{fig:res:eval:stability} shows the average error over iterations for 100 independent random tests. This plot shows that our method is robust to noise in the input, even if the input maps can have  errors up to approximately 40\% of the shape radius. At the same time, our algorithm can efficiently filter out the noise within a small number of iterations. Note that in 94 cases out of 100 the refined maps converged to a nearly identical final result, while in the remaining 6, the refinement led to maps that are mixed with symmetric ambiguity since there is too much noise introduced into their initialization. 

\begin{figure}[t]
    \centering
%
%
\definecolor{mycolor1}{rgb}{0.00000,0.44700,0.74100}%
\definecolor{mycolor2}{rgb}{0.67843,0.92157,1.00000}%
\definecolor{mycolor3}{rgb}{0.92900,0.69400,0.12500}%
\definecolor{mycolor4}{rgb}{0.49400,0.18400,0.55600}%
\definecolor{mycolor5}{rgb}{0.46600,0.67400,0.18800}%
\definecolor{mycolor6}{rgb}{0.85098,0.32549,0.09804}%
\definecolor{mycolor7}{rgb}{0.50196,0.50196,0.50196}%
\pgfplotsset{
compat=1.11,
legend image code/.code={
\draw[mark repeat=2,mark phase=2]
plot coordinates {
(0cm,0cm)
(0.15cm,0cm)        
(0.2cm,0cm)         
};%
}
}
\begin{tikzpicture}

\begin{axis}[
width=0.4\linewidth,
height=0.32\linewidth,
at={(1.713in,1.457in)},
scale only axis,
every x tick label/.append style={font=\color{black}, font=\tiny},
every y tick label/.append style={font=\color{black}, font=\tiny},
xmin=1,
xmax=22,
xlabel style={font=\color{white!15!black}},
xlabel={\footnotesize \# iterations},
xlabel style={at={(0.5,-0.08)}},
ymin=0,
ymax=70,
xmajorgrids,
ymajorgrids,
axis background/.style={fill=white},
title style={font=\bfseries},
title style={at={(0.5,0.95)}},
title={\footnotesize Dirichlet Energy},
legend style={at={(1,0.7)}, anchor=south east, legend cell align=left, align=left, draw=white!15!black},
legend style={draw=none,legend columns=3, fill opacity=0, text opacity = 1, draw opacity=1}
]

\addplot [color=mycolor7, line width=1.5pt]
  table[row sep=crcr]{%
1	9.234095\\
2	9.234095\\
3	9.234095\\
4	9.234095\\
5	9.234095\\
6	9.234095\\
7	9.234095\\
8	9.234095\\
9	9.234095\\
10	9.234095\\
11	9.234095\\
12	9.234095\\
13	9.234095\\
14	9.234095\\
15	9.234095\\
16	9.234095\\
17	9.234095\\
18	9.234095\\
19	9.234095\\
20	9.234095\\
21	9.234095\\
22	9.234095\\
};
\addlegendentry{\scriptsize GT}


\addplot [color=mycolor2, line width=1.5pt]
  table[row sep=crcr]{%
1	19.61705\\
2	30.7174\\
3	29.2521\\
4	28.60255\\
5	28.0731\\
6	28.1012\\
7	28.2383\\
8	27.6892\\
9	27.704\\
10	27.55915\\
11	27.45485\\
12	27.37615\\
13	27.38735\\
14	27.29715\\
15	27.28515\\
16	27.3547\\
17	27.28245\\
18	27.2114\\
19	27.2198\\
20	27.2001\\
21	27.26\\
22	27.23\\
};
\addlegendentry{\scriptsize ICP$_{50}$}

\addplot [color=mycolor3, line width=1.5pt]
  table[row sep=crcr]{%
1	19.61705\\
2	44.078\\
3	38.2836\\
4	37.56885\\
5	36.93385\\
6	36.39585\\
7	36.0345\\
8	35.889\\
9	35.78565\\
10	35.5503\\
11	35.4409\\
12	35.29365\\
13	35.05985\\
14	35.07575\\
15	34.9254\\
16	34.70985\\
17	34.60335\\
18	34.5455\\
19	34.52965\\
20	34.44165\\
21	34.4044\\
22	34.36225\\
};
\addlegendentry{\scriptsize ICP$_{75}$}

\addplot [color=mycolor4, line width=1.5pt]
  table[row sep=crcr]{%
1	19.61705\\
2	60.66305\\
3	47.91225\\
4	46.3755\\
5	45.60785\\
6	45.3572\\
7	44.8914\\
8	44.79835\\
9	44.4314\\
10	44.50395\\
11	44.3098\\
12	44.0407\\
13	43.97615\\
14	43.92735\\
15	43.6349\\
16	43.51695\\
17	43.55735\\
18	43.4295\\
19	43.5116\\
20	43.5512\\
21	43.4273\\
22	43.411\\
};
\addlegendentry{\scriptsize ICP$_{100}$}

\addplot [color=mycolor5, line width=1.5pt]
  table[row sep=crcr]{%
1	19.61705\\
2	61.99795\\
3	51.26575\\
4	49.3123\\
5	48.3546\\
6	47.8477\\
7	47.6115\\
8	47.215\\
9	46.91845\\
10	47.03815\\
11	47.00475\\
12	46.78965\\
13	46.69405\\
14	46.5454\\
15	46.5071\\
16	46.61495\\
17	46.50155\\
18	46.49035\\
19	46.499\\
20	46.42755\\
21	46.25965\\
22	46.1685\\
};
\addlegendentry{\scriptsize ICP$_{120}$}

\addplot [color=mycolor6, line width=1.5pt]
  table[row sep=crcr]{%
1	19.61705\\
2	22.43201\\
3	20.405745\\
4	15.24221\\
5	14.22575\\
6	12.197095\\
7	12.26843\\
8	11.736815\\
9	12.01943\\
10	11.62895\\
11	11.704255\\
12	11.122685\\
13	11.030515\\
14	10.88188\\
15	10.88566\\
16	10.56237\\
17	10.61529\\
18	10.24778\\
19	10.10915\\
20	9.68848\\
21	9.621165\\
22	9.32935\\
};
\addlegendentry{\scriptsize \textbf{Ours}}

\end{axis}
\end{tikzpicture}%
    \begin{overpic}[trim=8cm -6cm 5cm 0cm,clip,width=0.48\linewidth,grid=flase]{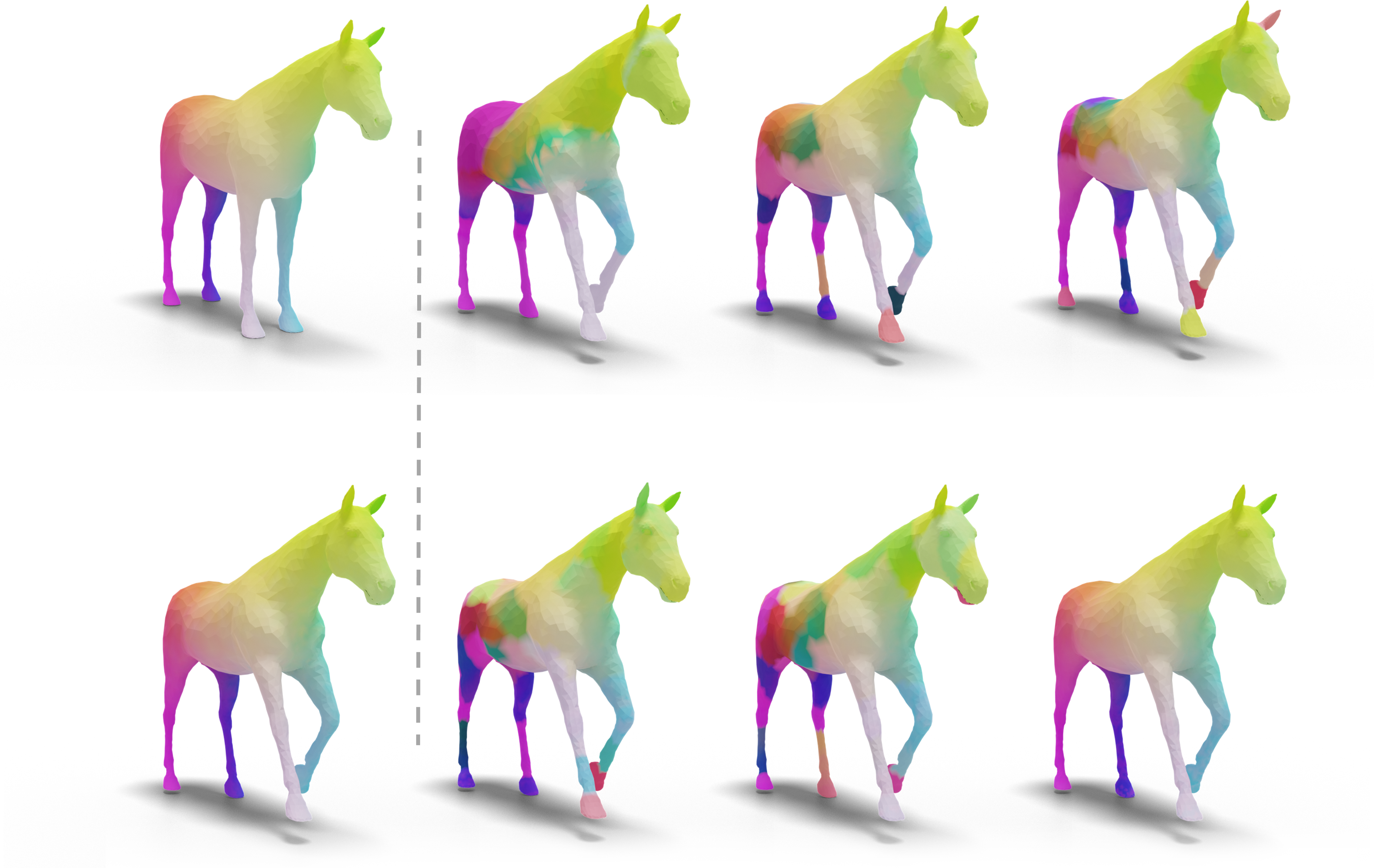}
    \put(8,85){\scriptsize{Source}}
    \put(34,85){\scriptsize{Ini}}
    \put(58,85){\scriptsize{ICP$_{50}$}}
    \put(82,85){\scriptsize{ICP$_{75}$}}
    \put(9,42){\scriptsize{GT}}
    \put(32,42){\scriptsize{ICP$_{100}$}}
    \put(57,42){\scriptsize{ICP$_{120}$}}
    \put(82,42){\scriptsize{\textbf{Ours}}}
    \end{overpic}
    \vspace{-0.2cm}
    \caption{Average Dirichlet energy of pointwise maps on 20 TOSCA pairs, starting with a computed $20 \times 20$ functional map, refined either using ICP in different dimensions or \name\ until $120 \times 120$. Our method converges to a smoother map, with Dirichlet energy closer to the ground truth.
    }
\label{fig:res:eval:smoothness}
\vspace{-0.2cm}
\end{figure}
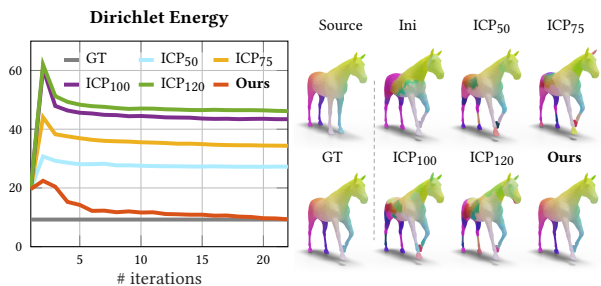

\subsubsection{Smoothness}
The maps refined with our method are typically very smooth, although this constraint is not enforced explicitly. Figure~\ref{fig:res:eval:smoothness} shows a quantitative measurement of the smoothness compared to ICP with different dimensions on 20 pairs of shapes from the TOSCA dataset \cite{TOSCA}, starting with a $20 \times 20$ functional map computed via \cite{nogneng17}. Map smoothness is measured as the mean Dirichlet energy of the normalized coordinates of the target shape mapped on the source through the given point-to-point map. Our method clearly provides smoother maps, and approaches the ground truth after a few iterations.

\subsection{Practical Applications}\label{sec:apps}
We applied our method 
across a range of application scenarios, including symmetry detection, map refinement among complete shapes, partial matching and function transfer. In each application we demonstrate a quantitative improvement as well as a significant speedup compared to the best competing method.
\revised{Note that in all experiments, we use the same initialization for all competing methods to guarantee a fair comparison.}
 %
 %

Unless otherwise stated, ICP uses the same dimension as the output dimension of \name. 
``Ours'' refers to applying {\name} on the complete meshes, while ``Ours$^*$'' refers to {\name} with sub-sampling for acceleration. In both cases, we always output dense correspondences between complete meshes.
 To measure the accuracy, we only accept {\em direct} ground-truth maps (except for the symmetry detection application, where the symmetric ground-truth maps are considered).
For texture transfer, we first convert the point-wise map to a functional map with size 300$\times$300, then we use this functional map to transfer the uv-coordinates from source to target. 

\begin{figure}[!t]
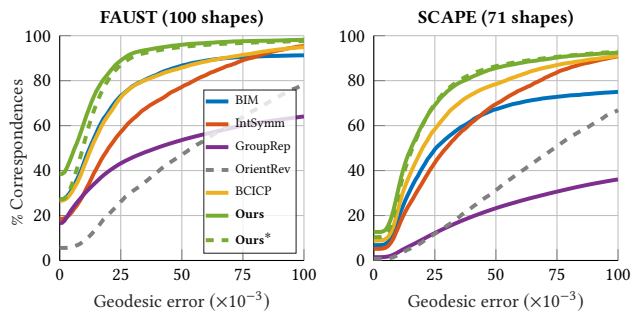

\centering
\input{./figures2/Err_symmMap_3baselines_FAUST.tex}
\input{./figures2/Err_symmMap_3baselines_SCAPE.tex}
\vspace{-0.3cm}
\caption{Error summary of symmetry detection. We compare with the recent state-of-the-art methods IntSymm \cite{Nagar_2018_ECCV} and GroupRep \cite{wang2017group}, as well as to the baseline Blended Intrinsic Maps (BIM) \cite{kim2011blended} and BCICP.}
\label{fig:res:symm:summary}
\end{figure}

\begin{table}[!t]
\caption{\textbf{Symmetry Detection}. Given approximate symmetric maps (OrientRev~\cite{ren2018continuous}), we refine them using our method or BCICP, and compare the results to several state-of-the-art methods, including BIM, IntSymm, and GroupRep. Here we report the average error and runtime over 100 FAUST shapes and 71 SCAPE shapes. We also include the results of our method with sub-sampling for acceleration (called Ours*).\vspace{-1mm}}
\label{tb:res:symm}
\centering
\footnotesize
\begin{tabular}{|c|c|cc|cc|}
\hline
\multicolumn{2}{|c|}{Measurement} & \multicolumn{2}{c|}{Average Error ($\times 10^{-3}$)} & \multicolumn{2}{c|}{Average Runtime (s)} \\ \hline
\multicolumn{2}{|c|}{Method \textbackslash~Dataset} & FAUST & SCAPE & FAUST & SCAPE \\ \hline
\multicolumn{2}{|l|}{BIM~\scriptsize{\cite{kim2011blended}}} & 65.4 & 133 & 34.6 & 41.7 \\
\multicolumn{2}{|l|}{GroupRep~\scriptsize{\cite{wang2017group}}} & 224 & 347 & 8.48 & 16.7 \\
\multicolumn{2}{|l|}{IntSymm~\scriptsize{\cite{Nagar_2018_ECCV}}} & 33.9 & 60.3 & 1.35 & 1.81 \\
\multicolumn{2}{|l|}{OrientRev (Ini) ~\scriptsize{\cite{ren2018continuous}} } & 68.0 & 110 & 0.59 & 1.07 \\
\multicolumn{2}{|l|}{Ini + BCICP~\scriptsize{\cite{ren2018continuous}}} & 29.2 & 49.7 & 195.1 & 525.6 \\ \hline
\multicolumn{2}{|c|}{\textbf{Ini + Ours}} & \textbf{16.1} & \textbf{46.2} & 22.6 & 62.7 \\
\multicolumn{2}{|c|}{\textbf{Ini + Ours*}} & 18.5 & 46.6 & \textbf{1.78} & \textbf{3.66} \\ \hline
\multirow{2}{*}{\begin{tabular}[c]{@{}c@{}}\scriptsize{\textbf{Improv}. w.r.t}\\ \scriptsize{state-of-the-art}\end{tabular}} & Ini + Ours & \textbf{44.9}\% & \textbf{7.0}\% & 8$\times$ & $8\times$ \\
 & Ini + Ours* & 36.6\% & 6.2\% & \textbf{110}$\times$ & \textbf{140}$\times$ \\ \hline
\end{tabular}
\end{table}

\subsubsection{Symmetry Detection}
We first apply our approach for computing pose-invariant
symmetries. This problem has received a lot of attention in the past and here we compare to the  most recent and  widely used  techniques. In this application we
are only given a single shape and our goal is to compute a high-quality intrinsic symmetry, such as
the left-right symmetry present in humans. This problem is slightly different from the pairwise matching scenario, since the identity solution must be ruled out. We do so by
leveraging a recent approach for encoding map orientation in functional map computations
\cite{ren2018continuous}. Namely, we compute an initial $10\times 10$ functional map by solving an
optimization problem with exactly the same parameters as in \cite{ren2018continuous} and WKS descriptors as input, but instead of
orientation-preserving, we promote orientation-{\em reversing} maps. This gives us an initial functional map which we then upsample to size $100 \times 100$. 
 Figure \ref{fig:res:symm:summary} shows the error curves on the SCAPE \cite{Anguelov05} and FAUST benchmarks (for which we have the ground truth symmetry map), while Table~\ref{tb:res:symm} reports the average error and runtime. 
 Note that the shapes in both datasets are not meshed in a symmetric way, so a successful method must be able to handle, often significant, changes in mesh structure.

Our approach 
 achieves a significant quality improvement compared to all state-of-the-art methods, and is also significantly faster. With acceleration, we achieve a speedup of more than 100x on a workstation with a 3.10GHz CPU and 64GB RAM.  Figure \ref{fig:res:symm:eg_scape_faust} further shows a qualitative comparison. Finally, we remark that for human shapes the first four Laplacian
eigenfunctions follow the same structure disambiguating top-bottom and left-right. Therefore we can
use a fixed $4\times 4$ diagonal functional map with entries $1, 1,-1,-1$ as an initial guess for human symmetry detection. Results with this initialization are shown in the supplementary materials. 

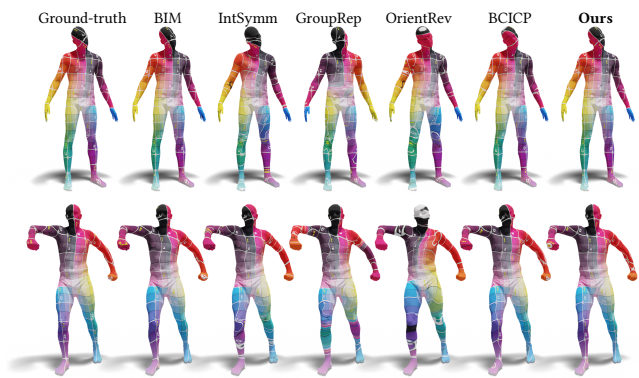
\begin{figure}[!t]
\centering
  \begin{overpic}
  [trim=6cm 0cm 10cm 0cm,clip,width=1\columnwidth,grid=false]{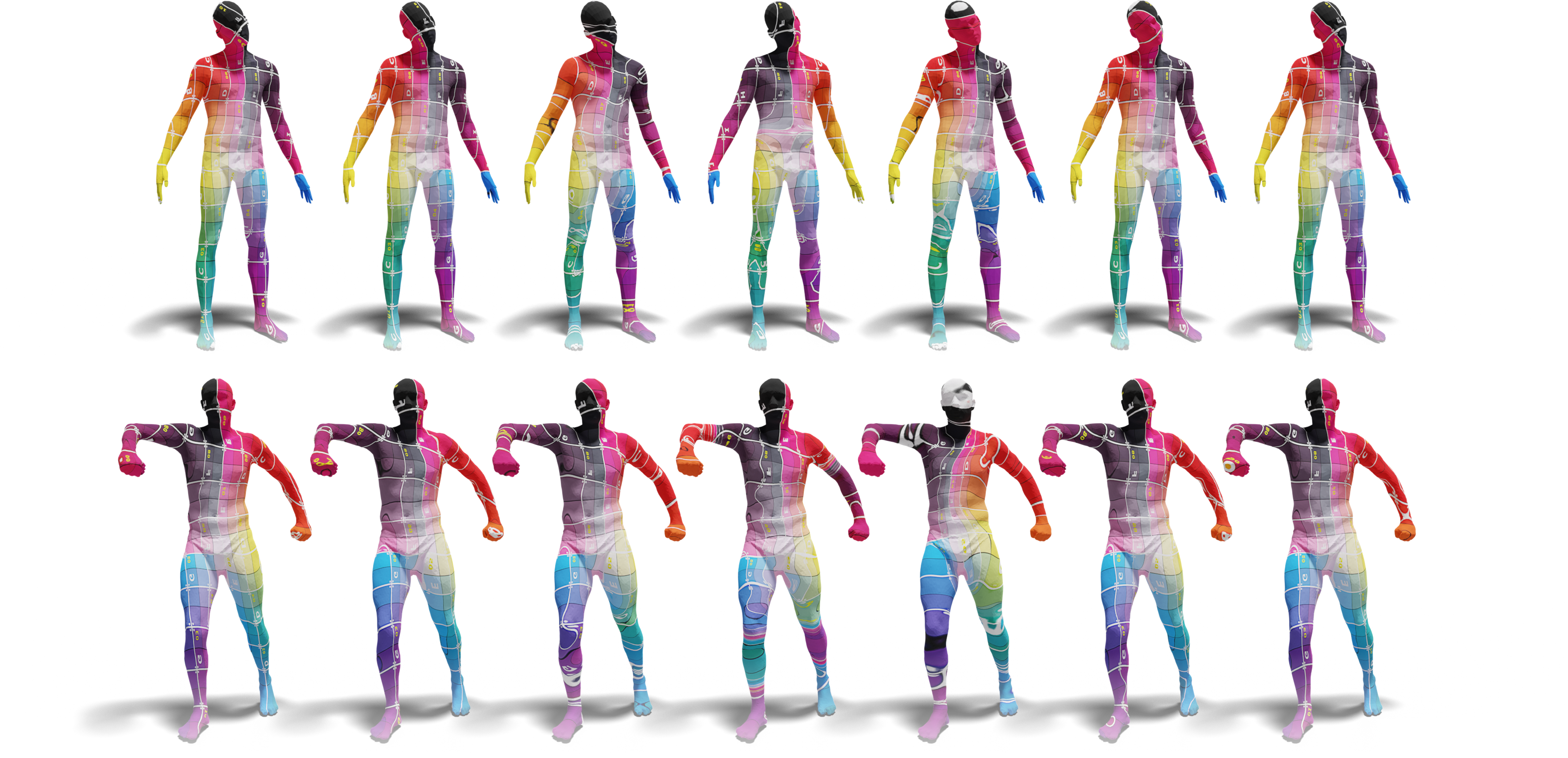}
  \put(5,56){\scriptsize{Ground-truth}}
  \put(23,56){\scriptsize{BIM}}
  \put(33,56){\scriptsize{IntSymm}}
  \put(45,56){\scriptsize{GroupRep}}
  \put(59,56){\scriptsize{OrientRev}}
  \put(75,56){\scriptsize{BCICP}}
  \put(89,56){\scriptsize{\textbf{Ours}}}
  \end{overpic}
  \vspace{-0.6cm}
  \caption{\label{fig:res:symm:eg_scape_faust}Symmetry detection. We show two examples with FAUST (first row) and SCAPE
    (second row) and visualize the symmetric maps from different methods via texture transfer. Note that our method with
    acceleration is over 100$\times$ faster than BCICP, while achieving comparable or better quality.}
  \vspace{-0.2cm}
\end{figure}

\subsubsection{Refinement for shape matching.}
We applied our technique to refine maps between pairs of shapes and compared our method with recent state-of-the-art refinement techniques, including RHM~\cite{ezuz2018reversible}, PMF~\cite{vestner2017product}, BCICP~\cite{ren2018continuous}, \revised{Deblur~{\cite{ezuz2017deblurring}}}, as well as the standard refinement ICP~\cite{ovsjanikov2012functional}.

For each dataset (FAUST and SCAPE), we consider three different versions. (1) Original: where 
all the meshes have the same triangulation. (2) Remeshed: we randomly flipped $12.5\%$ of the edges (using gptoolbox \cite{gptoolbox}) keeping the vertex positions unchanged to maintain a perfect ground-truth. (3) Remeshed + Resampled (called "Resampled" in Table~\ref{tb:res:fasut_summary}): we use the datasets provided in~\cite{ren2018continuous}, where each shape
is remeshed and resampled independently, having different number of vertices (around 5k) and often significantly different triangulation.
As such, these are more challenging than the original datasets on which near-perfect results have been reported in the past. Figure~\ref{fig:res:faust:diff_tri} shows a FAUST shape in the three versions.


\begin{figure}
\centering
\begin{minipage}{0.5\linewidth}\centering
    \begin{overpic}[trim=9cm 1cm 3cm 0cm,clip,width=0.9\linewidth,grid=false]{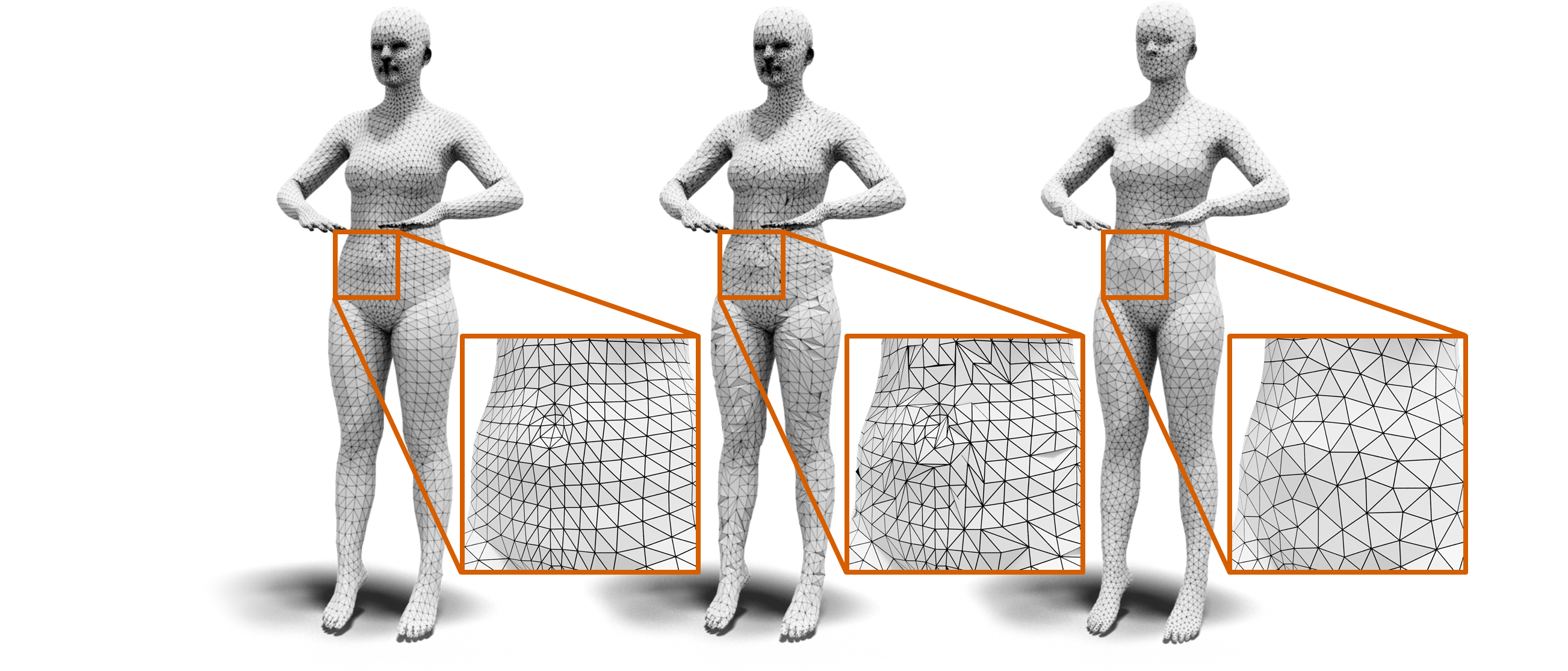}
    \put(7,53){\scriptsize{Original}}
    \put(36,53){\scriptsize{Remeshed}}
    \put(67,55.8){\scriptsize{Remeshed}}
    \put(63,50.8){\scriptsize{+ Resampled}}
    \end{overpic}
    \vspace{-0.1cm}
    \caption{Different triangulation}
    \label{fig:res:faust:diff_tri}
\end{minipage}\hfill\hspace{-0.2cm}
\begin{minipage}{0.51\linewidth}\centering
    \begin{overpic}[trim=10cm 2cm 15cm 0cm,clip,width=1\linewidth,grid=false]{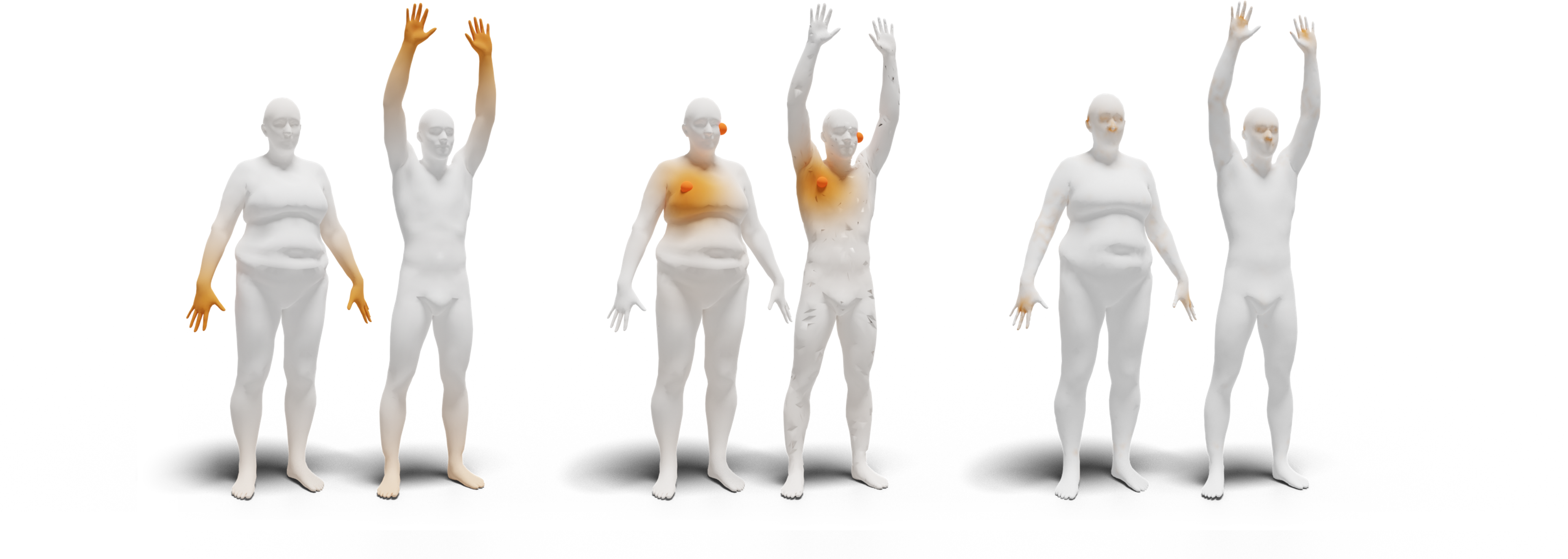}
    \put(6,46){\scriptsize{WKS desc}}
    \put(38,46){\scriptsize{2 landmarks}}
    \put(74,46){\scriptsize{Neural SHOT}}
    \end{overpic}
    \vspace{-0.6cm}
    \caption{Different descriptors}
    \label{fig:res:faust_desc}
\end{minipage}
\vspace{-0.1cm}
\end{figure}

To demonstrate that our algorithm works with different
initializations, we use three different types of descriptors to
compute the initial functional maps (with size $20\times20$) for the
three datasets: (1) WKS; (2) descriptors derived from two landmarks (see the
two spheres highlighted in the middle of
Figure~\ref{fig:res:faust_desc}); (3) Learned SHOT
descriptors~\cite{roufosse2018unsupervised}: the descriptors computed
by a non-linear transformation of SHOT, using an unsupervised deep
learning method trained on a mixed subset of the remeshed and
resampled SCAPE and FAUST dataset. For the experiments with WKS descriptors, 
we also use the orientation-preserving operators~\cite{ren2018continuous} to disambiguate
the symmetry of the WKS descriptors.

\setlength{\tabcolsep}{0.25em}
\begin{table}[!t]
\caption{\textbf{Quantitative evaluation of refinement for shape matching}. The Original and Remeshed datasets include 300 shape pairs. The Resampled dataset includes 190 FAUST pairs and 153 SCAPE pairs.}
\vspace{-0.1cm}
\label{tb:res:fasut_summary}
\centering
\footnotesize
\begin{tabular}{|cc|ccc|ccc|}
\hline
 &  & \multicolumn{3}{c|}{Average Error ($\times 10^{-3}$)} & \multicolumn{3}{c|}{Average Runtime (s)} \\ \hline
\multicolumn{2}{|c|}{Method \textbackslash~Dataset} & \scriptsize{Original} & \scriptsize{Remeshed} & \scriptsize{ Resampled} & \scriptsize{Original} & \scriptsize{Remeshed} & \scriptsize{Resampled} \\ \hline
\multicolumn{2}{|c|}{Ini} & 67.3 & 44.0 & 46.5 & - & - & - \\ \hline
\multicolumn{2}{|c|}{ICP} & 54.0 & 36.3 & 29.3 & 10.2 & 10.1 & 5.32 \\
\multicolumn{2}{|c|}{\revised{Deblur}} & \revised{61.9} & \revised{38.6} & \revised{44.4} & \revised{10.9}  & \revised{11.7} & \revised{10.4} \\
\multicolumn{2}{|c|}{RHM} & 41.9 & 33.3 & 32 & 41.4 & 42.5 & 47.4 \\
\multicolumn{2}{|c|}{PMF} & 26.4 & 25.9 & 86.4 & 736.5 & 780.2 & 311.5 \\
\multicolumn{2}{|c|}{BCICP} & 21.6 & 19.5 & 26 & 183.7 & 117.8 & 364.2 \\ \hline
\multicolumn{2}{|c|}{\textbf{Ours}} & \textbf{15.8} & \textbf{13.3} & \textbf{21.7} & 9.60 & 9.64 & 6.49 \\
\multicolumn{2}{|c|}{\textbf{Ours*}} & 17.5 & 14.5 & 24.6 & \textbf{1.14} & \textbf{1.15} & \textbf{0.68} \\ \hline
\multicolumn{1}{|c|}{\multirow{2}{*}{\textbf{Improv.}}} & Ours &\textbf{ 26.9}\% & \textbf{31.8}\% & \textbf{16.5}\% & 19$\times$ & 12$\times$ & 56$\times$ \\
\multicolumn{1}{|c|}{} & Ours* & 19.0\% & 25.6\% & 5.4\% & \textbf{160}$\times$ & \textbf{100}$\times$ & \textbf{535}$\times$ \\ \hline
\end{tabular}
\vspace{-0.2cm}
\end{table}

Table~\ref{tb:res:fasut_summary} reports the average error and runtime, 
 while the corresponding summary curves are in the supplementary materials. Figure~\ref{fig:res:eg:faust_refinement} shows a qualitative example.
Our method without acceleration achieves 26.9\%, 31.8\%, and 16.5\%
improvement in accuracy over the state-of-the-art while being 10
to 50 times faster. With acceleration, our method is more than
100-500$\times$ faster than the top existing method while
still producing accuracy improvement. Our method is also much simpler than BCICP (see Appendix~\ref{sec:appendix:code} for an overview of the source code of BCICP and our method). Interestingly, we also note that the method in \cite{roufosse2018unsupervised} overfits severely when trained directly on functional maps of size 120 and results in an average error of 97.5. In contrast, training on smaller functional maps and using our upsampling leads to average error of 21.7. Please see the supplementary for an illustration.
\revised{
We provide evaluation of other quantitative measurements such as bijectivity, coverage, and edge distortion in Appendix~{\ref{sec:appendix:measurements}}.
We also provide additional qualitative examples and comparison to the Deblur method on non-isometric shapes in Appendix~{\ref{sec:appendix:deblur}}.
}


\subsubsection{Matching different high-resolution meshes}
SHREC19 \cite{SHREC19} is a recent benchmark composed of 430 human pairs with different connectivity and mesh resolution, gathered using 44 different shapes from 11 datasets. Each shape is aligned to the SMPL
model~\cite{SMPL} using the registration pipeline of~\cite{FARM}, thus providing a dense ground truth for quantitative evaluation.
 This benchmark is challenging due to high shape variance and due to the presence of
high-resolution meshes (5K to 200K vertices, see supplementary materials for examples).  In
Table~\ref{tb:res:shrec19} we report full comparisons in terms of average error and runtime.

Since BCICP and PMF require a full geodesic distance matrix as input, we
apply them on simplified shapes (we used MATLAB's \texttt{reducepatch} for the remeshing). 
 The refined maps are then propagated back to the original meshes via nearest neighbors; please see the supplementary materials for more details.

We initialize \name\ with the $20 \times 20$ functional map provided as baseline in \cite{SHREC19}, and upsample this map to size $120 \times 120$ with a step of size $5$.
Our method achieves the best results while being over 290$\times$ faster.
We also highlight that although we have a similar accuracy as
BCICP, we better preserve the local details as shown in
Figure~\ref{fig:res:shrec:bcicp_vs_zoomOut}, since we avoid the mesh
simplification and map transfer steps.

In the supplementary materials, we further compare to methods that are applicable on full-resolution meshes directly. The experiment is conducted on a subset of SHREC19 and our method achieves a significant improvement in accuracy.


\setlength{\tabcolsep}{0.3em}
\begin{table}[!t]
\caption{\textbf{SHREC19 summary}. We compare with the refinement techniques RHM, PMF, BCICP and the baseline ICP on 430 shape pairs. We report an accuracy improvement over BCICP (the top performing method on this benchmark), and a significant gap in runtime performance over all methods.}
\label{tb:res:shrec19}
\vspace{-0.1cm}
\centering
\footnotesize
\begin{tabular}{|c|c|c|c|c|}
\hline
\multicolumn{2}{|c|}{\multirow{2}{*}{Method}} & \multirow{2}{*}{\#samples} & \multicolumn{2}{c|}{Measurement} \\ \cline{4-5} 
\multicolumn{2}{|c|}{} &  & Avg. Error ($\times 10^{-3}$) & Avg. Runtime (s) \\ \hline
\multicolumn{2}{|c|}{Initialization} & - & 60.4 & - \\ \hline
\multicolumn{2}{|c|}{ICP} & - & 47.0 & 87.3 \\ \hline
\multicolumn{2}{|c|}{\revised{Deblur}} & - & \revised{55.4} & \revised{102.1} \\ \hline
\multicolumn{2}{|c|}{RHM} & - & 42.6 & 2313 \\ \hline
\multicolumn{2}{|c|}{\multirow{3}{*}{PMF}} & 500 & 56.2 & 72.9 \\
\multicolumn{2}{|c|}{} & 1000 & 51.8 & 118.1 \\
\multicolumn{2}{|c|}{} & 5000 & 83.2 & 349.3 \\ \hline
\multicolumn{2}{|c|}{\multirow{3}{*}{BCICP}} & 500 & 40.7 & 90.0 \\
\multicolumn{2}{|c|}{} & 1000 & 33.6 & 163.7 \\
\multicolumn{2}{|c|}{} & 5000 & 30.1 & 437.9 \\ \hline
\multicolumn{2}{|c|}{\textbf{Ours*}} & \textbf{500 }& \textbf{28.8} &\textbf{ 1.5} \\ \hline
\textbf{Improv.} & Ours* & 500 & \textbf{4}\% &\textbf{290}$\times$ \\ \hline
\end{tabular}
\end{table}

\begin{figure}[!t]
\centering
  \begin{overpic}
  [trim=1cm 0cm 0cm 0cm,clip,width=1\columnwidth,grid=false]{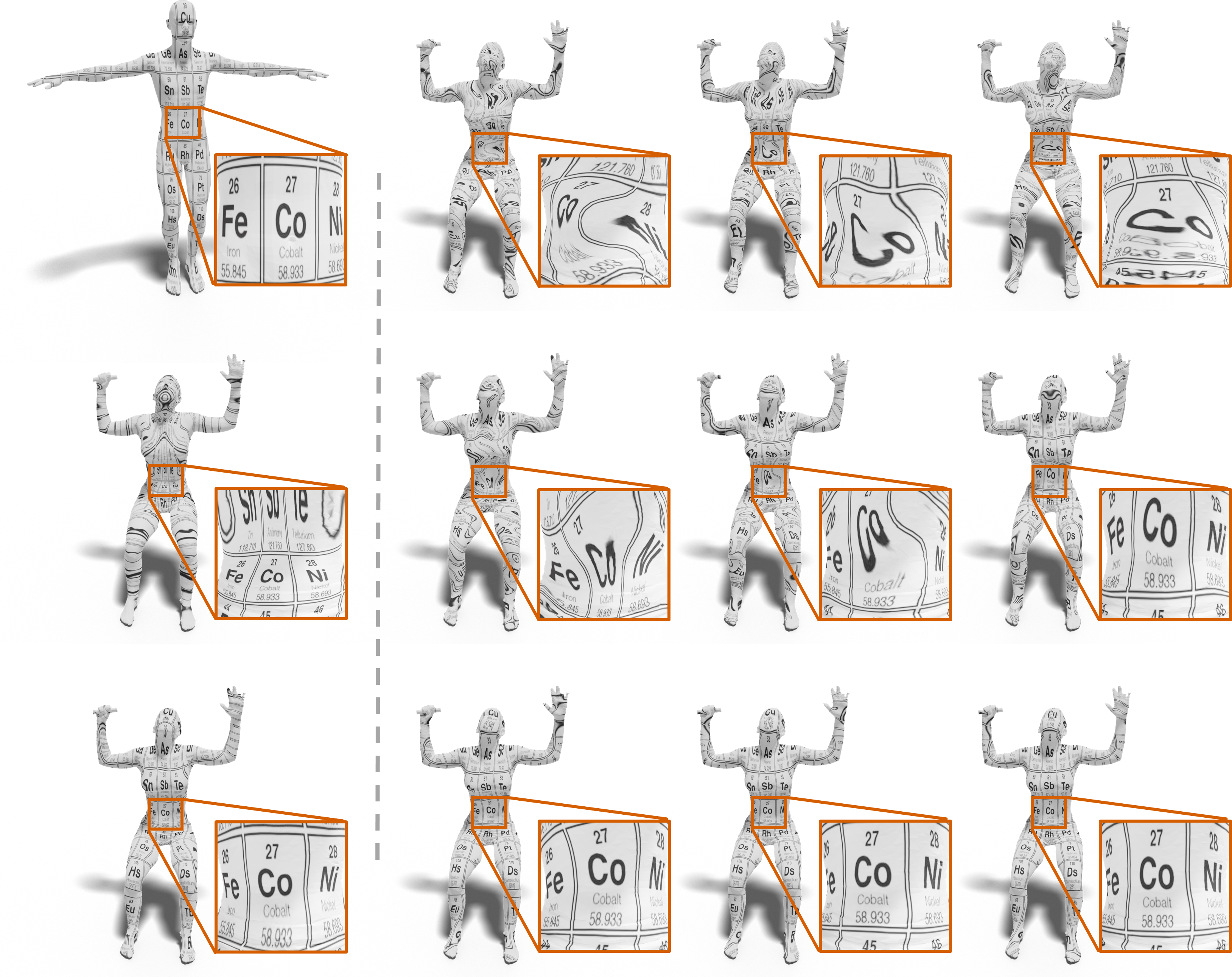}
  \put(5,81){\footnotesize{Source \scriptsize{($n=53K$)}}}
  \put(34,81){\footnotesize{PMF (500)}}
  \put(57,81){\footnotesize{PMF (1000)}}
  \put(80,81){\footnotesize{PMF (5000)}}
  
  \put(8,52){\footnotesize{Initialization}}
  \put(33,52){\footnotesize{BCICP (500)}}
  \put(56,52){\footnotesize{BCICP (1000)}}
  \put(79,52){\footnotesize{BCICP (5000)}}
  
  \put(3,25){\footnotesize{Reference \scriptsize{($n=16K$)}}}
  \put(33,25){\footnotesize{Ours (500)}}
  \put(56,25){\footnotesize{Ours (1000)}}
  \put(79,25){\footnotesize{Ours (5000)}}
  \end{overpic}
  \vspace{-0.5cm}
   \caption{Different sampling density. Here we show an example from the SHREC19 benchmark on a pair of shape with 53K and 16K vertices respectively. We compare with PMF and BCICP under different sampling density (500, 1000, and 5000 samples). The computed maps are visualized via texture transfer. 
    Our method achieves the best global accuracy while preserving the local details at the same time. Further, our method is much less dependent on the sampling density than BCICP or PMF.}
\label{fig:res:shrec:bcicp_vs_zoomOut}
\end{figure}
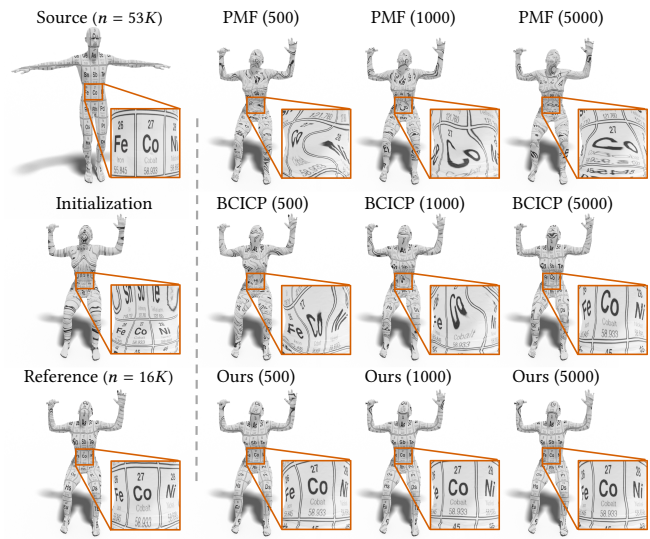

\setlength{\tabcolsep}{0.3em}
\begin{table}[!t]
\caption{\label{tab:pcl}\textbf{Quantitative evaluation on point cloud surfaces.} Our method is both more accurate and faster than ICP on average.}
\vspace{-0.1cm}
\centering
\footnotesize
\begin{tabular}{|c|c|ccc|c|c|}
\hline
\multirow{2}{*}{Measurement \textbackslash~Method} & \multirow{2}{*}{Ini} & \multirow{2}{*}{ICP} & \multirow{2}{*}{ICP$_{20}$} & \multirow{2}{*}{ICP$_{120}$} & \multirow{2}{*}{\textbf{Ours*}} & \textbf{Improv.} \\ \cline{7-7} 
 &  &  &  &  &  & Ours \\ \hline
Average Error ($\times 10^{-3}$) & 51.0 & 49.7 & 31.4 & 36.9 & \textbf{22.3} & 29.0\% \\ \hline
Average Runtime (s) & - & 29.6 & 8.3 & 305.2  & \textbf{4.0} & 2$\times$\\ \hline
\end{tabular}
\end{table}

\begin{figure}[!t]
\centering
\input{figures2/figures_simone/FAUST_TOSCA_x10_pairs.tikz}
\begin{overpic}[trim=0cm 0cm 0cm 0cm,clip,width=\columnwidth,grid=false]{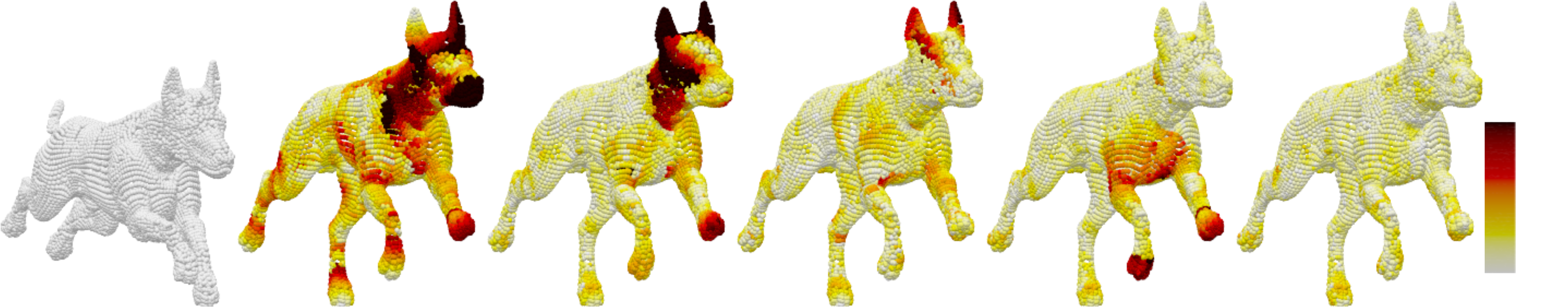}
\put(3,17.3){\footnotesize Source}
\put(20,17.3){\footnotesize Ini}
\put(36,17.3){\footnotesize ICP}
\put(50,17.3){\footnotesize ICP$_{20}$}
\put(65.6,17.3){\footnotesize ICP$_{120}$}
\put(83,17.3){\footnotesize \textbf{Ours}}
\put(97,10.3){\tiny $0.1$}
\put(97,1.3){\tiny 0}
\end{overpic}
\vspace{-0.5cm}
\caption{\label{fig:FAUST_PC}Results on non-rigid point cloud surfaces. We tested on 10 FAUST pairs and 10 TOSCA pairs. Below, we visualize geodesic error directly on the point clouds, defined as the Euclidean distance between the estimated matches and the ground truth. The heatmap grows from white (zero error) to dark red ($\ge 10\%$ deviation from ground truth).}
\vspace{-0.1cm}
\end{figure}
%

\subsubsection{Point cloud surfaces.}
Several standard methods for meshes typically fail when applied to point clouds. We tested our approach on point clouds generated from the FAUST and TOSCA datasets, by sampling points within mesh triangles uniformly at random.
We estimate the Laplace operator on point clouds as proposed in \cite{Belkin09}.
The initial $20 \times 20$ functional map is estimated with the approach of~\cite{nogneng17}, using WKS and 2 landmarks (Ini). We then upsample from 20 to 120 with steps of size 5, and compare with ICP, ICP$_{20}$ and ICP$_{120}$.
Quantitative and qualitative results are shown in Table~\ref{tab:pcl} and Figure~\ref{fig:FAUST_PC}.

\begin{figure}[tb]
  \centering
  \begin{overpic}
  [trim=0cm 0cm 0cm 0cm,clip,width=0.99\linewidth]{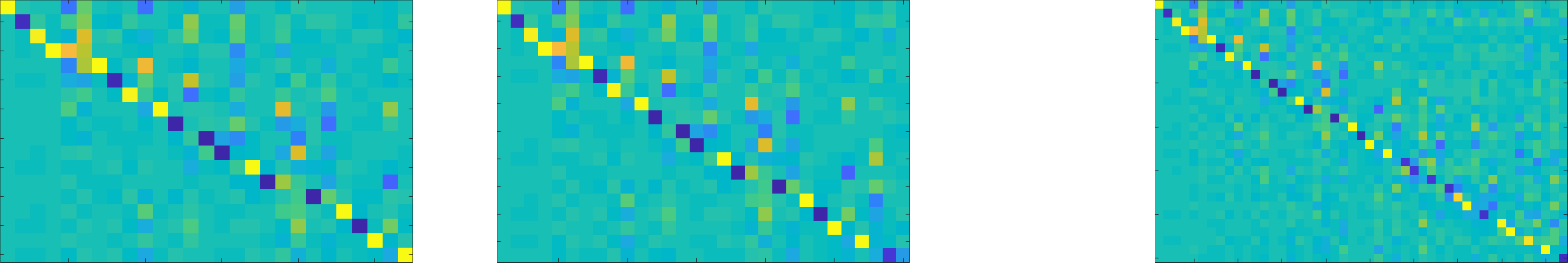}
  \put(28,7){\footnotesize $\to$}
  \put(9,-2.3){\footnotesize $18\times 27$}
  %
  \put(60,7){\footnotesize $\to \cdots \to$}
  \put(40.5,-2.3){\footnotesize $19\times 30$}
  %
  \put(83,-2.3){\footnotesize $30\times 47$}
  \end{overpic}
  \caption{\label{fig:rank}Partial matching involves functional maps $\C$ with slanted diagonal. To account for this particular structure, we iteratively increase the two dimensions of $\C$ by different amounts, see Equations~\eqref{eq:upd_km}-\eqref{eq:upd_kn}. This allows correct upsampling, as shown in this example.}
\vspace{-0.1cm}
\end{figure}
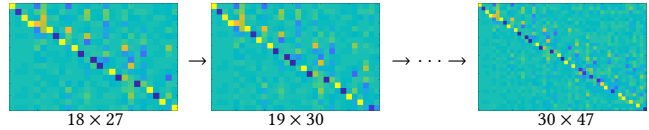
\begin{figure}[t]
  \centering
  \input{./figures/partial_cuts.tex}
  \input{./figures/partial_holes.tex}
  \begin{overpic}
  [trim=0cm 0cm 0cm 0cm,clip,width=\linewidth]{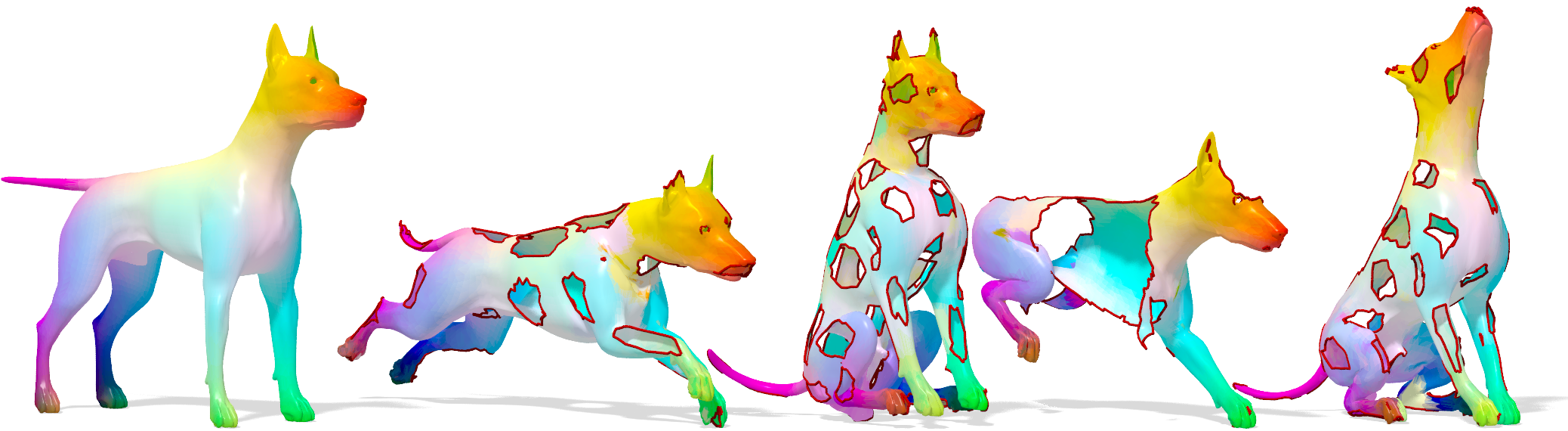}
  \put(6,19){\footnotesize Source}
  \end{overpic}
  \vspace{-0.5cm}
  \caption{\label{fig:partial}{\em Top}: Comparisons on the SHREC'16 Partiality benchmark with the state of the art method Partial Functional Maps (PFM) \cite{rodola2017partial} and with the Random Forests (RF) baseline \cite{rodola2014dense}. Average runtimes are 6sec for our method and 70sec for PFM, both initialized with a $4\times 4$ ground truth $\C$. {\em Bottom}: Qualitative results on the dog class.}
  \vspace{-0.2cm}
\end{figure}

\subsubsection{Partial Matching}
%
A particularly challenging setting of shape correspondence occurs whenever one of the two shapes has missing geometry. 
In~\cite{rodola2017partial} it was shown that, in case of partial isometries, the functional map matrix $\C$ 
has a ``slanted diagonal'' with slope proportional to the area ratio $\frac{A(\N)}{A(\M)}$ (here, $\M$ is a partial shape and $\N$ is a complete shape). 
Our spectral upsampling method can still be applied in this setting. To do so, we {\em weakly} enforce the expectation of a slanted diagonal by allowing rectangular $\C$. Namely, we define the update rules for the step size as follows:
\begin{align}
    k_\M &\mapsto k_\M+1 \label{eq:upd_km}\\
    k_\N &\mapsto k_\N + 1+ \lceil \frac{k_\N}{100}(100-r) \rceil \label{eq:upd_kn}
\end{align}
where $r$ is an estimate for $\mathrm{rank}(\C)$ obtained via the formula $r = \max_{i=1}^{k_\M} \{  i \; | \; \lambda_i^\M < \max_{j=1}^{k_\N} \lambda_j^\N \}$  after setting $k_\M=k_\N=100$ (see \cite[Eq. 9]{rodola2017partial} for details). 
%
%
%
In the classical case where both $\M$ and $\N$ are full and nearly isometric, the estimate boils down to $r = \min\{k_\M,k_\N\}=100$ 
%
and Eq.~\eqref{eq:upd_kn} reduces to $k_\N \mapsto k_\N + 1$; see Figure~\ref{fig:rank} for an illustration.

For these tests we adopt the SHREC'16 Partial Correspondence benchmark~\cite{shrec16partial}, consisting of 8 shape classes (humans and animals) undergoing partiality transformations of two kinds: regular `cuts' and irregular `holes'. All shapes are additionally resampled independently to $\sim\!10$K vertices. Evaluation is performed over 200 shape pairs in total, where each partial shape is matched to a full template of the corresponding class. Quantitative and qualitative results are reported in Figure~\ref{fig:partial}.

\subsubsection{Topological Noise}
%
We further explored the case of topological changes in the areas of self-contact (e.g., touching hands generating a geodesic shortcut). For this task, we compare with the state of the art on the SHREC'16 Topology benchmark~\cite{shrec16topology} (low-res challenge), consisting of 25 shape pairs ($\sim\!12$K vertices) undergoing nearly isometric deformations with severe topological artifacts. 
We initialize our method with a $30\times 30$ matrix $\C$ estimated via standard least squares with SHOT descriptors~\cite{shot}. 
 Since self-contact often leads to partiality, we use the rectangular update rules~\eqref{eq:upd_km}-\eqref{eq:upd_kn}.
Results are reported in Figure~\ref{fig:topology}. \revised{Figure~\ref{fig:res:topology_eg} shows some example maps computed using our method.}

\begin{figure}[bt]
  \centering
  \input{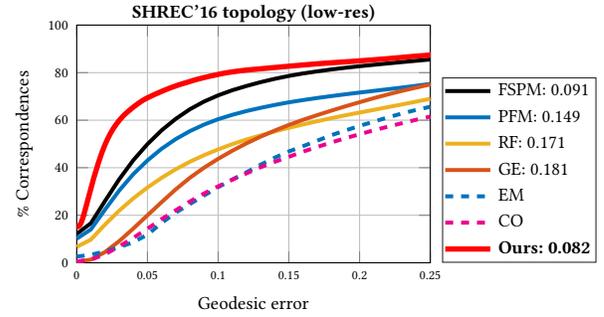}
  \vspace{-0.3cm}
  \caption{\label{fig:topology}Comparisons on the SHREC'16 Topology benchmark. Competing methods include PFM, RF, Green's Embedding (GE)~\cite{burghard2017embedding}, Expectation Maximization (EM)~\cite{sahilliouglu2012minimum}, Convex Optimization (CO)~\cite{Koltun}, and Fully Spectral Partial Matching (FSPM)~\cite{litany2017fully}. Dashed curves indicate sparse methods.}
\end{figure}

\begin{figure}[tb]
  \vspace{-0.3cm}
  \centering
  \begin{overpic}
  [trim=0cm 0cm 0cm 0cm,clip,width=1\linewidth, grid=false]{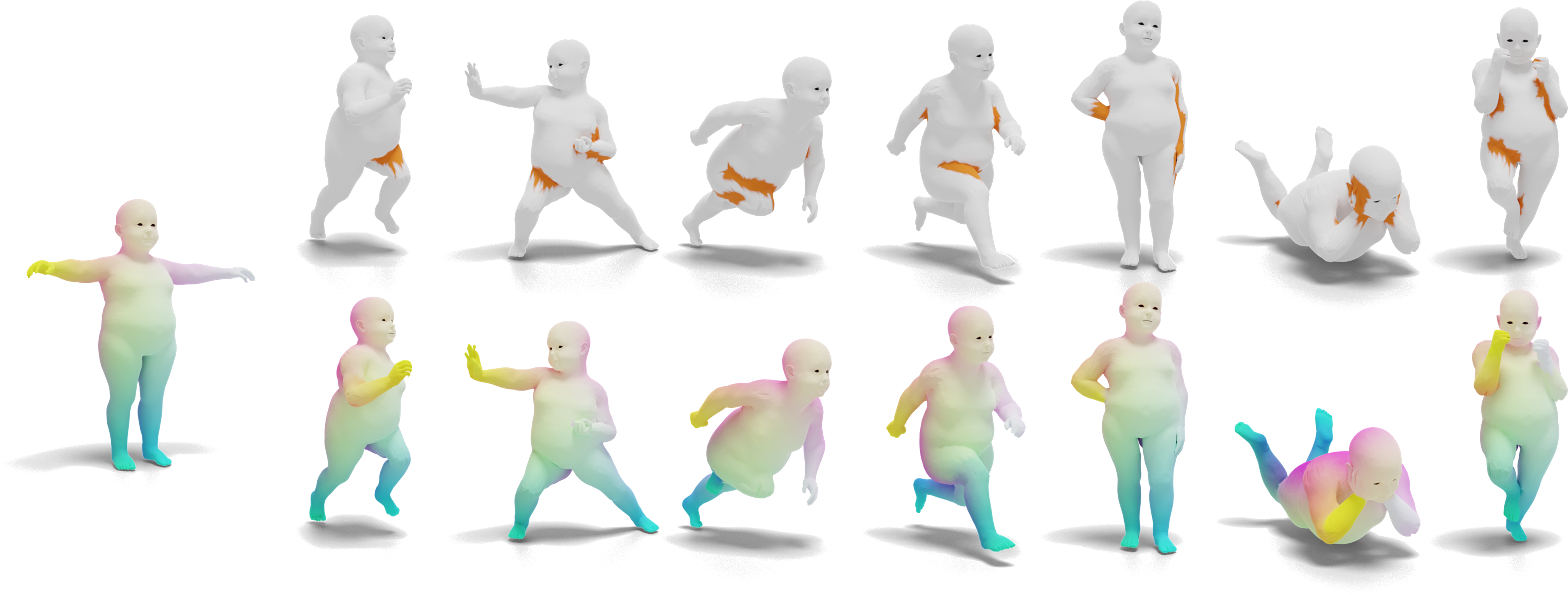}
  \put(5,27){\footnotesize Source}
  \end{overpic}
  \vspace{-0.7cm}
  \caption{\label{fig:res:topology_eg} \revised{{\em Top}: the regions with topology noise are highlighted in orange; {\em Bottom}:  maps computed using our method visualized via color transfer.}}
  \vspace{-0.4cm}
\end{figure}

\subsubsection{Different Basis}
%
In~\cite{LMH} it was proposed to address the partial setting by considering a Hamiltonian $H_\M=L_\M+V_\M$ in place of the standard manifold Laplacian, where $V_\M=\mathrm{diag}(1-v)$ is a localization potential concentrated on the support of a given (soft) indicator function $v:\M\to[0,1]$; eigenfunctions of $H_\M$ are supported on $v$. 
%
%
 We performed experiments showing that spectral upsampling can still be applied {\em as-is} to improve the quality of maps, when these are represented in this alternative basis. In these tests we initialized as in~\cite{LMH}, and evaluated on the entire dataset of~\cite{cosmo2016matching}, consisting of $150$ cluttered scenes and $3$ query models (animals). The results are reported in the supplementary materials.

\begin{table}[!t]
\vspace{-0.2cm}
\caption{Results in the transfer of different classes of functions, average on $20$ pairs from FAUST
  dataset. Initial map size is $40 \times 30$ (Ini), final size of ours is  $210 \times 200$. The methods marked
  with $\dagger$ are initialized with the initial functional maps refined by ICP. See text for details.} 
\label{tab:funtiontransfer}  
\footnotesize
\begin{tabular}{|l|c|c|c|c|c|c|c|}\hline 
\textbf{\footnotesize function} & \textbf{\footnotesize Ini} & \textbf{\footnotesize ICP} & \textbf{\footnotesize p2p$^{\dagger}$} & \textbf{\footnotesize ICP$\tiny{_{200}}$} & \textbf{\footnotesize Prod$^{\dagger}$} & \textbf{\footnotesize Ours} & \textbf{\footnotesize Ours$^{\dagger}$} \\\hline 
   \footnotesize{HeatKernel}& 0.80 & 0.18 & 0.15 & 0.17 & 0.19 & \textbf{0.10} & \textbf{0.10} \\\hline 
   \footnotesize{HeatKernel$\tiny{_{200}}$ } & 0.95 & 0.84 & 0.52 & 0.34 & 0.65 & \textbf{0.29} & \textbf{0.29} \\\hline 
    \footnotesize{HKS} & 0.66 & 0.55 & 0.21 & 0.21 & 0.28 & 0.14 & \textbf{0.13} \\\hline 
    \footnotesize{WKS} & 0.51 & 0.15 & 0.06 & 0.11 & 0.13 & \textbf{0.04} & \textbf{0.04} \\\hline 
 \footnotesize{XYZ} & 0.67 & 0.13 & 0.09 & 0.12 & 0.15 & \textbf{0.05} & \textbf{0.05} \\\hline 
 \footnotesize{Indicator} & 0.77 & 0.30 & 0.18 & 0.20 & 0.26 & \textbf{0.17} & \textbf{0.17} \\\hline 
   \footnotesize{SHOT} & 0.87 & 0.82 & 0.87 & 0.74 & 0.78 & \textbf{0.73} & \textbf{0.73} \\\hline 
   \footnotesize{AWFT} & 0.45 & 0.26 & 0.18 & 0.19 & 0.24 & \textbf{0.14} & \textbf{0.14} \\\hline 
  \footnotesize{Delta} & 0.98 & 0.93 & 0.67 & 0.43 & 0.82 & \textbf{0.38} & \textbf{0.38} \\\hline 
\end{tabular} 

\end{table} 
\begin{figure}[t!]
  \centering
  \begin{overpic}
  [trim=0cm 0cm 0cm 0cm,clip,width=\linewidth]{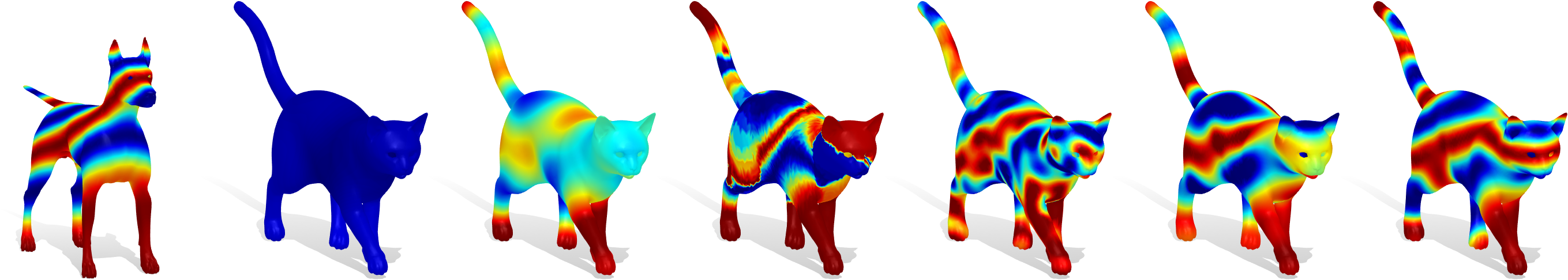}
  \put(0,16.5){\footnotesize{original $f$}}
  \put(19.8,16.5){\footnotesize{Ini}}
  \put(34.5,16.5){\footnotesize{ICP}}
  \put(48,16.5){\footnotesize{p2p$^\dagger$}}
  \put(62.5,16.5){\footnotesize{ICP$_{300}$}}
  \put(77,16.5){\footnotesize{Prod$^\dagger$}}
  \put(91.5,16.5){\textbf{\footnotesize{Ours}}}
  \end{overpic}
  \vspace{-0.6cm}
  \caption{\label{fig:dog_cat} Function transfer example on a non-isometric pair from TOSCA. We show the original function on the source shape (leftmost) and the transfer results for the different methods. The functional map is upsampled from size $40 \times 30$ to $310 \times 300$. We mark the methods initialized with ICP with $\dagger$.}
  \vspace{-0.1cm}
\end{figure}
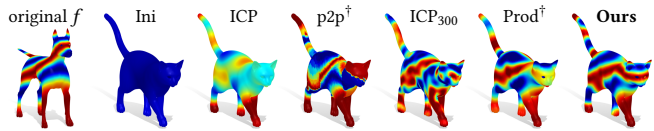

\subsubsection{Transfer of functions} 
Functional maps can be used
to transfer functions without necessarily converting to pointwise correspondences. This application,
however, can be hindered by the fact that small functional maps can only transfer low-frequency
information. A recent approach \cite{nogneng18} has tried to lift this restriction by noting that
higher frequency functions can be transferred using ``extended'' bases consisting of pointwise
products of basis functions. Our approach is similar in spirit since it also allows to extend the
expressive power of a given functional map by increasing its size and thus enabling transfer of
higher-frequency information.

We evaluated our method by directly comparing with the state of the art~\cite{nogneng18}.  For 9
different classes of functions we compute the error as the norm of the difference between 
the transferred function and the ground truth $g$ (obtained by transferring using the ground truth
pointwise map), normalized by the norm of $g$. The functions considered are: Heat Kernel computed
with $30$ and with $200$ eigenfunctions, descriptors HKS~\cite{sun2009concise},
WKS~\cite{aubry2011wave}, SHOT \cite{shot}, AWFT \cite{AWFT}, the coordinates of the 3D embedding, binary indicator of region, and the heat kernel with a very small time parameter approximating a delta function defined around a point.  The results are reported in Table~\ref{tab:funtiontransfer}.  We
use the same parameters adopted in~\cite{nogneng18}, and average over 20
 random  FAUST pairs.  We refine the initial map (Ini) of size
$40\times 30$, computed using the approach of \cite{nogneng17}, to $210\times 200$ with a step size of 1. We also compare to ICP: ICP refinement applied to Ini;
p2p: function transfer using the point-to-point map obtained by ICP; ICP$_{200}$: ICP applied to a functional map
of dimension $210 \times 200$ estimated through the same pipeline adopted for Ini; Prod: 
the method proposed in~\cite{nogneng18}. We outperform all the competitors for all the classes.

We also compare the results obtained by our method initializing the functional map after applying ICP, and the two are almost the same everywhere.  A transfer example of a high-frequency function between a dog and a cat shapes from TOSCA is visualized in Figure~\ref{fig:dog_cat}. Our refinement
achieves the best results with respect to all the competitors even in this non-isometric pair.  In
the supplementary materials we report other qualitative comparisons.

\section{Conclusion, Limitations \& Future Work}
\label{sec:conclusion}

We introduced a simple but efficient map refinement method based on iterative spectral upsampling. We presented a large variety of quantitative and qualitative results demonstrating that our method can produce similar or better quality on a wide range of shape matching problems while typically improving the speed of the matching by an order of magnitude or more. We find it remarkable that our method has such strong performance, even though it is conceptually simple and only requires a few lines of code to implement. In many cases, our method outperforms very complex frameworks that consist of multiple non-trivial algorithmic components.

Our method still comes with multiple limitations. First, while being robust to noise, its success still depends on a reasonable initialization. Starting with a bad initialization, such as random functional maps, our method would produce poor results.
Second, the method still relies on some parameters that have to be tuned for each application. Specifically, we need to
identify the number of basis functions in the initialization and the final number of basis functions. Additionally, the
step size during upsampling has to be chosen for optimal speed, but using a step size of one is always a safe
choice. Finally, our method is very robust to deviations from perfect isometries, but still will fail for significantly
non-isometric shape pairs.  \revised{See examples in Figure~\ref{fig:eg:failure} and in Appendix~{\ref{sec:appendix:deblur}.}}
In future work, we would like to investigate how to automatically compute the minimal size of the input functional map and plan to extend our work to other settings such as general graphs and images.

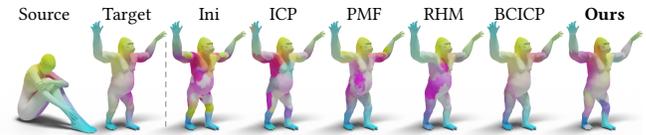
\begin{figure}[!t]
\centering
  \begin{overpic}
  [trim=6cm 0cm 2cm 0cm,clip,width=1\columnwidth,grid=false]{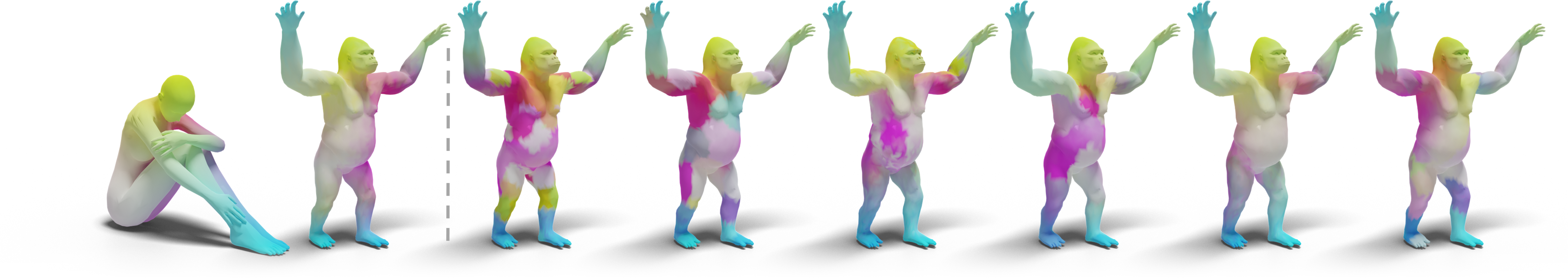}
  \put(3,19){\footnotesize{Source}}
  \put(16,19){\footnotesize{Target}}
  \put(31,19){\footnotesize{Ini}}
  \put(42,19){\footnotesize{ICP}}
  \put(54,19){\footnotesize{PMF}}
  \put(66,19){\footnotesize{RHM}}
  \put(77,19){\footnotesize{BCICP}}
  \put(91,19){\textbf{\footnotesize{Ours}}}
  \end{overpic}
  \vspace{-0.7cm}
\caption{Failure case. Here we show a challenging case where the initial map has left-to-right, back-to-front, and arm-to-leg ambiguity. When refining such a low-quality initial map, our method sometimes fails to produce a good refined map. However, our refinement still outperforms the regular ICP method with respect to the quality of the computed correspondences.}
\label{fig:eg:failure}
\vspace{-0.3cm}
\end{figure}

\begin{acks}
\revised{
The authors wish to thank the anonymous reviewers for their valuable comments and helpful suggestions, and Danielle Ezuz and Riccardo Marin for providing source code for experimental comparisons.
This work was supported by KAUST OSR Award No. CRG-2017-3426, a gift from the NVIDIA Corporation, the ERC Starting Grant StG-2017-758800 (EXPROTEA) and StG-2018-802554 (SPECGEO).
}
\end{acks}


\bibliographystyle{ACM-Reference-Format}
\bibliography{bibliography}


\begin{thebibliography}{70}


\ifx \showCODEN    \undefined \def \showCODEN     #1{\unskip}     \fi
\ifx \showDOI      \undefined \def \showDOI       #1{#1}\fi
\ifx \showISBNx    \undefined \def \showISBNx     #1{\unskip}     \fi
\ifx \showISBNxiii \undefined \def \showISBNxiii  #1{\unskip}     \fi
\ifx \showISSN     \undefined \def \showISSN      #1{\unskip}     \fi
\ifx \showLCCN     \undefined \def \showLCCN      #1{\unskip}     \fi
\ifx \shownote     \undefined \def \shownote      #1{#1}          \fi
\ifx \showarticletitle \undefined \def \showarticletitle #1{#1}   \fi
\ifx \showURL      \undefined \def \showURL       {\relax}        \fi
\providecommand\bibfield[2]{#2}
\providecommand\bibinfo[2]{#2}
\providecommand\natexlab[1]{#1}
\providecommand\showeprint[2][]{arXiv:#2}

\bibitem[\protect\citeauthoryear{Aflalo, Dubrovina, and Kimmel}{Aflalo
  et~al\mbox{.}}{2016}]%
        {aflalo2016spectral}
\bibfield{author}{\bibinfo{person}{Yonathan Aflalo}, \bibinfo{person}{Anastasia
  Dubrovina}, {and} \bibinfo{person}{Ron Kimmel}.}
  \bibinfo{year}{2016}\natexlab{}.
\newblock \showarticletitle{Spectral generalized multi-dimensional scaling}.
\newblock \bibinfo{journal}{\emph{International Journal of Computer Vision}}
  \bibinfo{volume}{118}, \bibinfo{number}{3} (\bibinfo{year}{2016}),
  \bibinfo{pages}{380--392}.
\newblock


\bibitem[\protect\citeauthoryear{Aflalo and Kimmel}{Aflalo and Kimmel}{2013}]%
        {aflalo2013spectral}
\bibfield{author}{\bibinfo{person}{Yonathan Aflalo} {and} \bibinfo{person}{Ron
  Kimmel}.} \bibinfo{year}{2013}\natexlab{}.
\newblock \showarticletitle{Spectral multidimensional scaling}.
\newblock \bibinfo{journal}{\emph{PNAS}} \bibinfo{volume}{110},
  \bibinfo{number}{45} (\bibinfo{year}{2013}), \bibinfo{pages}{18052--18057}.
\newblock


\bibitem[\protect\citeauthoryear{Anguelov, Srinivasan, Koller, Thrun, Rodgers,
  and Davis}{Anguelov et~al\mbox{.}}{2005}]%
        {Anguelov05}
\bibfield{author}{\bibinfo{person}{Dragomir Anguelov}, \bibinfo{person}{Praveen
  Srinivasan}, \bibinfo{person}{Daphne Koller}, \bibinfo{person}{Sebastian
  Thrun}, \bibinfo{person}{Jim Rodgers}, {and} \bibinfo{person}{James Davis}.}
  \bibinfo{year}{2005}\natexlab{}.
\newblock \showarticletitle{{SCAPE}: Shape Completion and Animation of People}.
\newblock \bibinfo{journal}{\emph{ACM Transactions on Graphics}}
  \bibinfo{volume}{24}, \bibinfo{number}{3} (\bibinfo{date}{July}
  \bibinfo{year}{2005}), \bibinfo{pages}{408--416}.
\newblock
\showISSN{0730-0301}


\bibitem[\protect\citeauthoryear{Aubry, Schlickewei, and Cremers}{Aubry
  et~al\mbox{.}}{2011}]%
        {aubry2011wave}
\bibfield{author}{\bibinfo{person}{Mathieu Aubry}, \bibinfo{person}{Ulrich
  Schlickewei}, {and} \bibinfo{person}{Daniel Cremers}.}
  \bibinfo{year}{2011}\natexlab{}.
\newblock \showarticletitle{The wave kernel signature: A quantum mechanical
  approach to shape analysis}. In \bibinfo{booktitle}{\emph{Computer Vision
  Workshops (ICCV Workshops), 2011 IEEE International Conference on}}. IEEE,
  \bibinfo{pages}{1626--1633}.
\newblock


\bibitem[\protect\citeauthoryear{Belkin, Sun, and Wang}{Belkin
  et~al\mbox{.}}{2009}]%
        {Belkin09}
\bibfield{author}{\bibinfo{person}{Mikhail Belkin}, \bibinfo{person}{Jian Sun},
  {and} \bibinfo{person}{Yusu Wang}.} \bibinfo{year}{2009}\natexlab{}.
\newblock \showarticletitle{Constructing Laplace Operator from Point Clouds in
  Rd}. In \bibinfo{booktitle}{\emph{Proc. Symposium on Discrete Algorithms
  (SODA)}}. \bibinfo{pages}{1031--1040}.
\newblock


\bibitem[\protect\citeauthoryear{Biasotti, Cerri, Bronstein, and
  Bronstein}{Biasotti et~al\mbox{.}}{2016}]%
        {biasotti2016recent}
\bibfield{author}{\bibinfo{person}{Silvia Biasotti}, \bibinfo{person}{Andrea
  Cerri}, \bibinfo{person}{Alex Bronstein}, {and} \bibinfo{person}{Michael
  Bronstein}.} \bibinfo{year}{2016}\natexlab{}.
\newblock \showarticletitle{Recent trends, applications, and perspectives in 3D
  shape similarity assessment}.
\newblock \bibinfo{journal}{\emph{Computer Graphics Forum}}
  \bibinfo{volume}{35}, \bibinfo{number}{6} (\bibinfo{year}{2016}),
  \bibinfo{pages}{87--119}.
\newblock


\bibitem[\protect\citeauthoryear{Bogo, Romero, Loper, and Black}{Bogo
  et~al\mbox{.}}{2014}]%
        {FAUST}
\bibfield{author}{\bibinfo{person}{Federica Bogo}, \bibinfo{person}{Javier
  Romero}, \bibinfo{person}{Matthew Loper}, {and} \bibinfo{person}{Michael~J.
  Black}.} \bibinfo{year}{2014}\natexlab{}.
\newblock \showarticletitle{{FAUST}: {D}ataset and evaluation for {3D} mesh
  registration}. In \bibinfo{booktitle}{\emph{Proc. CVPR}}.
  \bibinfo{publisher}{IEEE}, \bibinfo{address}{Columbus, Ohio},
  \bibinfo{pages}{3794--3801}.
\newblock


\bibitem[\protect\citeauthoryear{Bronstein, Bronstein, and Kimmel}{Bronstein
  et~al\mbox{.}}{2008}]%
        {TOSCA}
\bibfield{author}{\bibinfo{person}{Alex Bronstein}, \bibinfo{person}{Michael
  Bronstein}, {and} \bibinfo{person}{Ron Kimmel}.}
  \bibinfo{year}{2008}\natexlab{}.
\newblock \bibinfo{booktitle}{\emph{{N}umerical Geometry of Non-Rigid Shapes}}.
\newblock \bibinfo{publisher}{Springer}, \bibinfo{address}{New York, NY}.
\newblock


\bibitem[\protect\citeauthoryear{Burghard, Dieckmann, and Klein}{Burghard
  et~al\mbox{.}}{2017}]%
        {burghard2017embedding}
\bibfield{author}{\bibinfo{person}{Oliver Burghard}, \bibinfo{person}{Alexander
  Dieckmann}, {and} \bibinfo{person}{Reinhard Klein}.}
  \bibinfo{year}{2017}\natexlab{}.
\newblock \showarticletitle{Embedding shapes with {G}reen's functions for
  global shape matching}.
\newblock \bibinfo{journal}{\emph{Computers \& Graphics}}  \bibinfo{volume}{68}
  (\bibinfo{year}{2017}), \bibinfo{pages}{1--10}.
\newblock


\bibitem[\protect\citeauthoryear{Chen and Koltun}{Chen and Koltun}{2015}]%
        {Koltun}
\bibfield{author}{\bibinfo{person}{Qifeng Chen} {and} \bibinfo{person}{Vladlen
  Koltun}.} \bibinfo{year}{2015}\natexlab{}.
\newblock \showarticletitle{Robust Nonrigid Registration by Convex
  Optimization}. In \bibinfo{booktitle}{\emph{International Conference on
  Computer Vision (ICCV)}}. \bibinfo{publisher}{IEEE},
  \bibinfo{pages}{2039--2047}.
\newblock


\bibitem[\protect\citeauthoryear{Corman, Ovsjanikov, and Chambolle}{Corman
  et~al\mbox{.}}{2015}]%
        {corman2015continuous}
\bibfield{author}{\bibinfo{person}{Etienne Corman}, \bibinfo{person}{Maks
  Ovsjanikov}, {and} \bibinfo{person}{Antonin Chambolle}.}
  \bibinfo{year}{2015}\natexlab{}.
\newblock \showarticletitle{Continuous matching via vector field flow}.
\newblock \bibinfo{journal}{\emph{Computer Graphics Forum}}
  \bibinfo{volume}{34}, \bibinfo{number}{5} (\bibinfo{year}{2015}),
  \bibinfo{pages}{129--139}.
\newblock


\bibitem[\protect\citeauthoryear{Cosmo, Rodol\`{a}, Bronstein, Torsello,
  Cremers, and Sahillio\u{g}lu}{Cosmo et~al\mbox{.}}{2016a}]%
        {shrec16partial}
\bibfield{author}{\bibinfo{person}{Luca Cosmo}, \bibinfo{person}{Emanuele
  Rodol\`{a}}, \bibinfo{person}{Michael Bronstein}, \bibinfo{person}{Andrea
  Torsello}, \bibinfo{person}{Daniel Cremers}, {and} \bibinfo{person}{Yusuf
  Sahillio\u{g}lu}.} \bibinfo{year}{2016}\natexlab{a}.
\newblock \showarticletitle{Partial Matching of Deformable Shapes}. In
  \bibinfo{booktitle}{\emph{Proceedings of the Eurographics 2016 Workshop on 3D
  Object Retrieval}} \emph{(\bibinfo{series}{3DOR '16})}.
  \bibinfo{publisher}{Eurographics Association}, \bibinfo{pages}{61--67}.
\newblock
\showISBNx{978-3-03868-004-8}
\urldef\tempurl%
\url{https://doi.org/10.2312/3dor.20161089}
\showDOI{\tempurl}


\bibitem[\protect\citeauthoryear{Cosmo, Rodol\`a, Masci, Torsello, and
  Bronstein}{Cosmo et~al\mbox{.}}{2016b}]%
        {cosmo2016matching}
\bibfield{author}{\bibinfo{person}{Luca Cosmo}, \bibinfo{person}{Emanuele
  Rodol\`a}, \bibinfo{person}{Jonathan Masci}, \bibinfo{person}{Andrea
  Torsello}, {and} \bibinfo{person}{Michael Bronstein}.}
  \bibinfo{year}{2016}\natexlab{b}.
\newblock \showarticletitle{Matching deformable objects in clutter}. In
  \bibinfo{booktitle}{\emph{Proc. 3D Vision (3DV)}}. \bibinfo{pages}{1--10}.
\newblock


\bibitem[\protect\citeauthoryear{Dubrovina and Kimmel}{Dubrovina and
  Kimmel}{2010}]%
        {dubrovina2010matching}
\bibfield{author}{\bibinfo{person}{Anastasia Dubrovina} {and}
  \bibinfo{person}{Ron Kimmel}.} \bibinfo{year}{2010}\natexlab{}.
\newblock \showarticletitle{Matching shapes by eigendecomposition of the
  Laplace-Beltrami operator}. In \bibinfo{booktitle}{\emph{Proc. 3DPVT}},
  Vol.~\bibinfo{volume}{2}.
\newblock


\bibitem[\protect\citeauthoryear{Dubrovina and Kimmel}{Dubrovina and
  Kimmel}{2011}]%
        {dubrovina2011approximately}
\bibfield{author}{\bibinfo{person}{Anastasia Dubrovina} {and}
  \bibinfo{person}{Ron Kimmel}.} \bibinfo{year}{2011}\natexlab{}.
\newblock \showarticletitle{Approximately isometric shape correspondence by
  matching pointwise spectral features and global geodesic structures}.
\newblock \bibinfo{journal}{\emph{Advances in Adaptive Data Analysis}}
  \bibinfo{volume}{3}, \bibinfo{number}{01n02} (\bibinfo{year}{2011}),
  \bibinfo{pages}{203--228}.
\newblock


\bibitem[\protect\citeauthoryear{Dym and Lipman}{Dym and Lipman}{2017}]%
        {dym2017exact}
\bibfield{author}{\bibinfo{person}{Nadav Dym} {and} \bibinfo{person}{Yaron
  Lipman}.} \bibinfo{year}{2017}\natexlab{}.
\newblock \showarticletitle{Exact recovery with symmetries for Procrustes
  matching}.
\newblock \bibinfo{journal}{\emph{SIAM Journal on Optimization}}
  \bibinfo{volume}{27}, \bibinfo{number}{3} (\bibinfo{year}{2017}),
  \bibinfo{pages}{1513--1530}.
\newblock


\bibitem[\protect\citeauthoryear{Ezuz and Ben-Chen}{Ezuz and Ben-Chen}{2017}]%
        {ezuz2017deblurring}
\bibfield{author}{\bibinfo{person}{Danielle Ezuz} {and} \bibinfo{person}{Mirela
  Ben-Chen}.} \bibinfo{year}{2017}\natexlab{}.
\newblock \showarticletitle{Deblurring and Denoising of Maps between Shapes}.
\newblock \bibinfo{journal}{\emph{Computer Graphics Forum}}
  \bibinfo{volume}{36}, \bibinfo{number}{5} (\bibinfo{year}{2017}),
  \bibinfo{pages}{165--174}.
\newblock


\bibitem[\protect\citeauthoryear{Ezuz, Solomon, and Ben-Chen}{Ezuz
  et~al\mbox{.}}{2019}]%
        {ezuz2018reversible}
\bibfield{author}{\bibinfo{person}{Danielle Ezuz}, \bibinfo{person}{Justin
  Solomon}, {and} \bibinfo{person}{Mirela Ben-Chen}.}
  \bibinfo{year}{2019}\natexlab{}.
\newblock \showarticletitle{Reversible Harmonic Maps Between Discrete
  Surfaces}.
\newblock \bibinfo{journal}{\emph{ACM Trans. Graph.}} \bibinfo{volume}{38},
  \bibinfo{number}{2} (\bibinfo{year}{2019}), \bibinfo{pages}{15:1--15:12}.
\newblock


\bibitem[\protect\citeauthoryear{Gehre, Bronstein, Kobbelt, and Solomon}{Gehre
  et~al\mbox{.}}{2018}]%
        {gehre2018interactive}
\bibfield{author}{\bibinfo{person}{Anne Gehre}, \bibinfo{person}{Michael
  Bronstein}, \bibinfo{person}{Leif Kobbelt}, {and} \bibinfo{person}{Justin
  Solomon}.} \bibinfo{year}{2018}\natexlab{}.
\newblock \showarticletitle{Interactive curve constrained functional maps}.
\newblock \bibinfo{journal}{\emph{Computer Graphics Forum}}
  \bibinfo{volume}{37}, \bibinfo{number}{5} (\bibinfo{year}{2018}),
  \bibinfo{pages}{1--12}.
\newblock


\bibitem[\protect\citeauthoryear{Gunz and Mitteroecker}{Gunz and
  Mitteroecker}{2013}]%
        {gunz2013semilandmarks}
\bibfield{author}{\bibinfo{person}{Philipp Gunz} {and} \bibinfo{person}{Philipp
  Mitteroecker}.} \bibinfo{year}{2013}\natexlab{}.
\newblock \showarticletitle{Semilandmarks: a method for quantifying curves and
  surfaces}.
\newblock \bibinfo{journal}{\emph{Hystrix, the Italian Journal of Mammalogy}}
  \bibinfo{volume}{24}, \bibinfo{number}{1} (\bibinfo{year}{2013}),
  \bibinfo{pages}{103--109}.
\newblock


\bibitem[\protect\citeauthoryear{Huang, Wang, and Guibas}{Huang
  et~al\mbox{.}}{2014}]%
        {huang2014functional}
\bibfield{author}{\bibinfo{person}{Qixing Huang}, \bibinfo{person}{Fan Wang},
  {and} \bibinfo{person}{Leonidas Guibas}.} \bibinfo{year}{2014}\natexlab{}.
\newblock \showarticletitle{Functional map networks for analyzing and exploring
  large shape collections}.
\newblock \bibinfo{journal}{\emph{ACM Transactions on Graphics (TOG)}}
  \bibinfo{volume}{33}, \bibinfo{number}{4} (\bibinfo{year}{2014}),
  \bibinfo{pages}{36}.
\newblock


\bibitem[\protect\citeauthoryear{Huang and Ovsjanikov}{Huang and
  Ovsjanikov}{2017}]%
        {huang2017adjoint}
\bibfield{author}{\bibinfo{person}{Ruqi Huang} {and} \bibinfo{person}{Maks
  Ovsjanikov}.} \bibinfo{year}{2017}\natexlab{}.
\newblock \showarticletitle{Adjoint Map Representation for Shape Analysis and
  Matching}.
\newblock \bibinfo{journal}{\emph{Computer Graphics Forum}}
  \bibinfo{volume}{36}, \bibinfo{number}{5} (\bibinfo{year}{2017}),
  \bibinfo{pages}{151--163}.
\newblock


\bibitem[\protect\citeauthoryear{Jacobson et~al\mbox{.}}{Jacobson
  et~al\mbox{.}}{2018}]%
        {gptoolbox}
\bibfield{author}{\bibinfo{person}{Alec Jacobson} {et~al\mbox{.}}}
  \bibinfo{year}{2018}\natexlab{}.
\newblock \bibinfo{title}{{gptoolbox}: Geometry Processing Toolbox}.
\newblock
\newblock
\newblock
\shownote{http://github.com/alecjacobson/gptoolbox.}


\bibitem[\protect\citeauthoryear{Jain and Zhang}{Jain and Zhang}{2006}]%
        {jain2006}
\bibfield{author}{\bibinfo{person}{Varun Jain} {and} \bibinfo{person}{Hao
  Zhang}.} \bibinfo{year}{2006}\natexlab{}.
\newblock \showarticletitle{Robust 3D shape correspondence in the spectral
  domain}. In \bibinfo{booktitle}{\emph{Shape Modeling and Applications, 2006.
  SMI 2006. IEEE International Conference on}}. IEEE, \bibinfo{pages}{19--19}.
\newblock


\bibitem[\protect\citeauthoryear{Jain, Zhang, and van Kaick}{Jain
  et~al\mbox{.}}{2007}]%
        {jain2007}
\bibfield{author}{\bibinfo{person}{Varun Jain}, \bibinfo{person}{Hao Zhang},
  {and} \bibinfo{person}{Oliver van Kaick}.} \bibinfo{year}{2007}\natexlab{}.
\newblock \showarticletitle{Non-rigid spectral correspondence of triangle
  meshes}.
\newblock \bibinfo{journal}{\emph{International Journal of Shape Modeling}}
  \bibinfo{volume}{13}, \bibinfo{number}{01} (\bibinfo{year}{2007}),
  \bibinfo{pages}{101--124}.
\newblock


\bibitem[\protect\citeauthoryear{Kilian, Mitra, and Pottmann}{Kilian
  et~al\mbox{.}}{2007}]%
        {kilian2007geometric}
\bibfield{author}{\bibinfo{person}{Martin Kilian}, \bibinfo{person}{Niloy~J
  Mitra}, {and} \bibinfo{person}{Helmut Pottmann}.}
  \bibinfo{year}{2007}\natexlab{}.
\newblock \showarticletitle{Geometric modeling in shape space}. In
  \bibinfo{booktitle}{\emph{ACM Transactions on Graphics (TOG)}},
  Vol.~\bibinfo{volume}{26}. ACM, \bibinfo{pages}{64}.
\newblock


\bibitem[\protect\citeauthoryear{Kim, Lipman, and Funkhouser}{Kim
  et~al\mbox{.}}{2011}]%
        {kim2011blended}
\bibfield{author}{\bibinfo{person}{Vladimir~G Kim}, \bibinfo{person}{Yaron
  Lipman}, {and} \bibinfo{person}{Thomas Funkhouser}.}
  \bibinfo{year}{2011}\natexlab{}.
\newblock \showarticletitle{Blended intrinsic maps}. In
  \bibinfo{booktitle}{\emph{ACM Transactions on Graphics (TOG)}},
  Vol.~\bibinfo{volume}{30}. ACM, \bibinfo{pages}{79}.
\newblock


\bibitem[\protect\citeauthoryear{Kovnatsky, Bronstein, Bronstein, Glashoff, and
  Kimmel}{Kovnatsky et~al\mbox{.}}{2013}]%
        {kovnatsky2013coupled}
\bibfield{author}{\bibinfo{person}{Artiom Kovnatsky}, \bibinfo{person}{Michael
  Bronstein}, \bibinfo{person}{Alex Bronstein}, \bibinfo{person}{Klaus
  Glashoff}, {and} \bibinfo{person}{Ron Kimmel}.}
  \bibinfo{year}{2013}\natexlab{}.
\newblock \showarticletitle{Coupled quasi-harmonic bases}.
\newblock \bibinfo{journal}{\emph{Computer Graphics Forum}}
  \bibinfo{volume}{32}, \bibinfo{number}{2pt4} (\bibinfo{year}{2013}),
  \bibinfo{pages}{439--448}.
\newblock


\bibitem[\protect\citeauthoryear{Kovnatsky, Glashoff, and Bronstein}{Kovnatsky
  et~al\mbox{.}}{2016}]%
        {kovnatsky2016madmm}
\bibfield{author}{\bibinfo{person}{Artiom Kovnatsky}, \bibinfo{person}{Klaus
  Glashoff}, {and} \bibinfo{person}{Michael~M Bronstein}.}
  \bibinfo{year}{2016}\natexlab{}.
\newblock \showarticletitle{MADMM: a generic algorithm for non-smooth
  optimization on manifolds}. In \bibinfo{booktitle}{\emph{European Conference
  on Computer Vision}}. Springer, \bibinfo{pages}{680--696}.
\newblock


\bibitem[\protect\citeauthoryear{L\"{a}hner, Rodol\`{a}, Bronstein, Cremers,
  Burghard, Cosmo, Dieckmann, Klein, and Sahillio\u{g}lu}{L\"{a}hner
  et~al\mbox{.}}{2016}]%
        {shrec16topology}
\bibfield{author}{\bibinfo{person}{Zorah L\"{a}hner}, \bibinfo{person}{Emanuele
  Rodol\`{a}}, \bibinfo{person}{Michael Bronstein}, \bibinfo{person}{Daniel
  Cremers}, \bibinfo{person}{Oliver Burghard}, \bibinfo{person}{Luca Cosmo},
  \bibinfo{person}{Alexander Dieckmann}, \bibinfo{person}{Reinhard Klein},
  {and} \bibinfo{person}{Yusuf Sahillio\u{g}lu}.}
  \bibinfo{year}{2016}\natexlab{}.
\newblock \showarticletitle{Matching of Deformable Shapes with Topological
  Noise}. In \bibinfo{booktitle}{\emph{Proc. 3DOR}}. \bibinfo{pages}{55--60}.
\newblock


\bibitem[\protect\citeauthoryear{Litany, Rodol{\`a}, Bronstein, and
  Bronstein}{Litany et~al\mbox{.}}{2017}]%
        {litany2017fully}
\bibfield{author}{\bibinfo{person}{Or Litany}, \bibinfo{person}{Emanuele
  Rodol{\`a}}, \bibinfo{person}{Alex Bronstein}, {and} \bibinfo{person}{Michael
  Bronstein}.} \bibinfo{year}{2017}\natexlab{}.
\newblock \showarticletitle{Fully spectral partial shape matching}.
\newblock \bibinfo{journal}{\emph{Computer Graphics Forum}}
  \bibinfo{volume}{36}, \bibinfo{number}{2} (\bibinfo{year}{2017}),
  \bibinfo{pages}{247--258}.
\newblock


\bibitem[\protect\citeauthoryear{Loper, Mahmood, Romero, Pons-Moll, and
  Black}{Loper et~al\mbox{.}}{2015}]%
        {SMPL}
\bibfield{author}{\bibinfo{person}{Matthew Loper}, \bibinfo{person}{Naureen
  Mahmood}, \bibinfo{person}{Javier Romero}, \bibinfo{person}{Gerard
  Pons-Moll}, {and} \bibinfo{person}{Michael~J. Black}.}
  \bibinfo{year}{2015}\natexlab{}.
\newblock \showarticletitle{{SMPL}: A Skinned Multi-person Linear Model}.
\newblock \bibinfo{journal}{\emph{TOG}} \bibinfo{volume}{34},
  \bibinfo{number}{6} (\bibinfo{year}{2015}), \bibinfo{pages}{248:1--248:16}.
\newblock


\bibitem[\protect\citeauthoryear{Mandad, Cohen-Steiner, Kobbelt, Alliez, and
  Desbrun}{Mandad et~al\mbox{.}}{2017}]%
        {mandad2017variance}
\bibfield{author}{\bibinfo{person}{Manish Mandad}, \bibinfo{person}{David
  Cohen-Steiner}, \bibinfo{person}{Leif Kobbelt}, \bibinfo{person}{Pierre
  Alliez}, {and} \bibinfo{person}{Mathieu Desbrun}.}
  \bibinfo{year}{2017}\natexlab{}.
\newblock \showarticletitle{Variance-Minimizing Transport Plans for
  Inter-surface Mapping}.
\newblock \bibinfo{journal}{\emph{ACM Transactions on Graphics}}
  \bibinfo{volume}{36} (\bibinfo{year}{2017}), \bibinfo{pages}{14}.
\newblock


\bibitem[\protect\citeauthoryear{Marin, Melzi, Rodol\`a, and Castellani}{Marin
  et~al\mbox{.}}{2018}]%
        {FARM}
\bibfield{author}{\bibinfo{person}{Riccardo Marin}, \bibinfo{person}{Simone
  Melzi}, \bibinfo{person}{Emanuele Rodol\`a}, {and} \bibinfo{person}{Umberto
  Castellani}.} \bibinfo{year}{2018}\natexlab{}.
\newblock \bibinfo{title}{FARM: Functional Automatic Registration Method for 3D
  Human Bodies}.
\newblock
\newblock


\bibitem[\protect\citeauthoryear{Maron, Dym, Kezurer, Kovalsky, and
  Lipman}{Maron et~al\mbox{.}}{2016}]%
        {maron2016point}
\bibfield{author}{\bibinfo{person}{Haggai Maron}, \bibinfo{person}{Nadav Dym},
  \bibinfo{person}{Itay Kezurer}, \bibinfo{person}{Shahar Kovalsky}, {and}
  \bibinfo{person}{Yaron Lipman}.} \bibinfo{year}{2016}\natexlab{}.
\newblock \showarticletitle{Point registration via efficient convex
  relaxation}.
\newblock \bibinfo{journal}{\emph{ACM Transactions on Graphics (TOG)}}
  \bibinfo{volume}{35}, \bibinfo{number}{4} (\bibinfo{year}{2016}),
  \bibinfo{pages}{73}.
\newblock


\bibitem[\protect\citeauthoryear{Mateus, Horaud, Knossow, Cuzzolin, and
  Boyer}{Mateus et~al\mbox{.}}{2008}]%
        {mateus08}
\bibfield{author}{\bibinfo{person}{Diana Mateus}, \bibinfo{person}{Radu
  Horaud}, \bibinfo{person}{David Knossow}, \bibinfo{person}{Fabio Cuzzolin},
  {and} \bibinfo{person}{Edmond Boyer}.} \bibinfo{year}{2008}\natexlab{}.
\newblock \showarticletitle{Articulated Shape Matching Using {L}aplacian
  Eigenfunctions and Unsupervised Point Registration}. In
  \bibinfo{booktitle}{\emph{Proc. CVPR}}. \bibinfo{pages}{1--8}.
\newblock


\bibitem[\protect\citeauthoryear{Melzi, Marin, Rodol\`{a}, Castellani, Ren,
  Poulenard, Wonka, and Ovsjanikov}{Melzi et~al\mbox{.}}{2019}]%
        {SHREC19}
\bibfield{author}{\bibinfo{person}{Simone Melzi}, \bibinfo{person}{Riccardo
  Marin}, \bibinfo{person}{Emanuele Rodol\`{a}}, \bibinfo{person}{Umberto
  Castellani}, \bibinfo{person}{Jing Ren}, \bibinfo{person}{Adrien Poulenard},
  \bibinfo{person}{Peter Wonka}, {and} \bibinfo{person}{Maks Ovsjanikov}.}
  \bibinfo{year}{2019}\natexlab{}.
\newblock \showarticletitle{{SHREC 2019: Matching Humans with Different
  Connectivity}}. In \bibinfo{booktitle}{\emph{Eurographics Workshop on 3D
  Object Retrieval}}. \bibinfo{publisher}{The Eurographics Association}.
\newblock


\bibitem[\protect\citeauthoryear{Melzi, Rodol\`a, Castellani, and
  Bronstein}{Melzi et~al\mbox{.}}{2016}]%
        {AWFT}
\bibfield{author}{\bibinfo{person}{Simone Melzi}, \bibinfo{person}{Emanuele
  Rodol\`a}, \bibinfo{person}{Umberto Castellani}, {and}
  \bibinfo{person}{Michael Bronstein}.} \bibinfo{year}{2016}\natexlab{}.
\newblock \showarticletitle{Shape Analysis with Anisotropic Windowed Fourier
  Transform}. In \bibinfo{booktitle}{\emph{International Conference on 3D
  Vision (3DV)}}.
\newblock


\bibitem[\protect\citeauthoryear{Melzi, Rodol{\`a}, Castellani, and
  Bronstein}{Melzi et~al\mbox{.}}{2018}]%
        {LMH}
\bibfield{author}{\bibinfo{person}{Simone Melzi}, \bibinfo{person}{Emanuele
  Rodol{\`a}}, \bibinfo{person}{Umberto Castellani}, {and}
  \bibinfo{person}{Michael Bronstein}.} \bibinfo{year}{2018}\natexlab{}.
\newblock \showarticletitle{Localized Manifold Harmonics for Spectral Shape
  Analysis}.
\newblock \bibinfo{journal}{\emph{Computer Graphics Forum}}
  \bibinfo{volume}{37}, \bibinfo{number}{6} (\bibinfo{year}{2018}),
  \bibinfo{pages}{20--34}.
\newblock


\bibitem[\protect\citeauthoryear{Muja and Lowe}{Muja and Lowe}{2014}]%
        {flann}
\bibfield{author}{\bibinfo{person}{Marius Muja} {and} \bibinfo{person}{David~G.
  Lowe}.} \bibinfo{year}{2014}\natexlab{}.
\newblock \showarticletitle{Scalable Nearest Neighbor Algorithms for High
  Dimensional Data}.
\newblock \bibinfo{journal}{\emph{Pattern Analysis and Machine Intelligence,
  IEEE Transactions on}}  \bibinfo{volume}{36} (\bibinfo{year}{2014}).
\newblock


\bibitem[\protect\citeauthoryear{Nagar and Raman}{Nagar and Raman}{2018}]%
        {Nagar_2018_ECCV}
\bibfield{author}{\bibinfo{person}{Rajendra Nagar} {and}
  \bibinfo{person}{Shanmuganathan Raman}.} \bibinfo{year}{2018}\natexlab{}.
\newblock \showarticletitle{Fast and Accurate Intrinsic Symmetry Detection}. In
  \bibinfo{booktitle}{\emph{The European Conference on Computer Vision
  (ECCV)}}.
\newblock


\bibitem[\protect\citeauthoryear{Nogneng, Melzi, Rodol{\`a}, Castellani,
  Bronstein, and Ovsjanikov}{Nogneng et~al\mbox{.}}{2018}]%
        {nogneng18}
\bibfield{author}{\bibinfo{person}{Dorian Nogneng}, \bibinfo{person}{Simone
  Melzi}, \bibinfo{person}{Emanuele Rodol{\`a}}, \bibinfo{person}{Umberto
  Castellani}, \bibinfo{person}{Michael Bronstein}, {and} \bibinfo{person}{Maks
  Ovsjanikov}.} \bibinfo{year}{2018}\natexlab{}.
\newblock \showarticletitle{Improved Functional Mappings via Product
  Preservation}.
\newblock \bibinfo{journal}{\emph{Computer Graphics Forum}}
  \bibinfo{volume}{37}, \bibinfo{number}{2} (\bibinfo{year}{2018}),
  \bibinfo{pages}{179--190}.
\newblock


\bibitem[\protect\citeauthoryear{Nogneng and Ovsjanikov}{Nogneng and
  Ovsjanikov}{2017}]%
        {nogneng17}
\bibfield{author}{\bibinfo{person}{Dorian Nogneng} {and} \bibinfo{person}{Maks
  Ovsjanikov}.} \bibinfo{year}{2017}\natexlab{}.
\newblock \showarticletitle{Informative Descriptor Preservation via
  Commutativity for Shape Matching}.
\newblock \bibinfo{journal}{\emph{Computer Graphics Forum}}
  \bibinfo{volume}{36}, \bibinfo{number}{2} (\bibinfo{year}{2017}),
  \bibinfo{pages}{259--267}.
\newblock


\bibitem[\protect\citeauthoryear{Ovsjanikov, Ben-Chen, Solomon, Butscher, and
  Guibas}{Ovsjanikov et~al\mbox{.}}{2012}]%
        {ovsjanikov2012functional}
\bibfield{author}{\bibinfo{person}{Maks Ovsjanikov}, \bibinfo{person}{Mirela
  Ben-Chen}, \bibinfo{person}{Justin Solomon}, \bibinfo{person}{Adrian
  Butscher}, {and} \bibinfo{person}{Leonidas Guibas}.}
  \bibinfo{year}{2012}\natexlab{}.
\newblock \showarticletitle{Functional maps: a flexible representation of maps
  between shapes}.
\newblock \bibinfo{journal}{\emph{ACM Transactions on Graphics (TOG)}}
  \bibinfo{volume}{31}, \bibinfo{number}{4} (\bibinfo{year}{2012}),
  \bibinfo{pages}{30:1--30:11}.
\newblock


\bibitem[\protect\citeauthoryear{Ovsjanikov, Corman, Bronstein, Rodol\`{a},
  Ben-Chen, Guibas, Chazal, and Bronstein}{Ovsjanikov et~al\mbox{.}}{2017}]%
        {ovsjanikov2017computing}
\bibfield{author}{\bibinfo{person}{Maks Ovsjanikov}, \bibinfo{person}{Etienne
  Corman}, \bibinfo{person}{Michael Bronstein}, \bibinfo{person}{Emanuele
  Rodol\`{a}}, \bibinfo{person}{Mirela Ben-Chen}, \bibinfo{person}{Leonidas
  Guibas}, \bibinfo{person}{Frederic Chazal}, {and} \bibinfo{person}{Alex
  Bronstein}.} \bibinfo{year}{2017}\natexlab{}.
\newblock \showarticletitle{Computing and Processing Correspondences with
  Functional Maps}. In \bibinfo{booktitle}{\emph{ACM SIGGRAPH 2017 Courses}}.
  Article \bibinfo{articleno}{5}, \bibinfo{numpages}{5:1--5:62}~pages.
\newblock


\bibitem[\protect\citeauthoryear{Ovsjanikov, Merigot, Memoli, and
  Guibas}{Ovsjanikov et~al\mbox{.}}{2010}]%
        {ovsjanikov2010}
\bibfield{author}{\bibinfo{person}{Maks Ovsjanikov}, \bibinfo{person}{Quentin
  Merigot}, \bibinfo{person}{Facundo Memoli}, {and} \bibinfo{person}{Leonidas
  Guibas}.} \bibinfo{year}{2010}\natexlab{}.
\newblock \showarticletitle{One Point Isometric Matching with the Heat Kernel}.
\newblock \bibinfo{journal}{\emph{CGF}} \bibinfo{volume}{29},
  \bibinfo{number}{5} (\bibinfo{year}{2010}), \bibinfo{pages}{1555--1564}.
\newblock
\showISSN{1467-8659}
\urldef\tempurl%
\url{https://doi.org/10.1111/j.1467-8659.2010.01764.x}
\showDOI{\tempurl}


\bibitem[\protect\citeauthoryear{Pinkall and Polthier}{Pinkall and
  Polthier}{1993}]%
        {pinkall93}
\bibfield{author}{\bibinfo{person}{Ulrich Pinkall} {and}
  \bibinfo{person}{Konrad Polthier}.} \bibinfo{year}{1993}\natexlab{}.
\newblock \showarticletitle{{C}omputing {D}iscrete {M}inimal {S}urfaces and
  their {C}onjugates}.
\newblock \bibinfo{journal}{\emph{Experimental mathematics}}
  \bibinfo{volume}{2}, \bibinfo{number}{1} (\bibinfo{year}{1993}),
  \bibinfo{pages}{15--36}.
\newblock


\bibitem[\protect\citeauthoryear{Poulenard, Skraba, and Ovsjanikov}{Poulenard
  et~al\mbox{.}}{2018}]%
        {poulenard18}
\bibfield{author}{\bibinfo{person}{Adrien Poulenard}, \bibinfo{person}{Primoz
  Skraba}, {and} \bibinfo{person}{Maks Ovsjanikov}.}
  \bibinfo{year}{2018}\natexlab{}.
\newblock \showarticletitle{Topological Function Optimization for Continuous
  Shape Matching}.
\newblock \bibinfo{journal}{\emph{Computer Graphics Forum}}
  \bibinfo{volume}{37}, \bibinfo{number}{5} (\bibinfo{year}{2018}),
  \bibinfo{pages}{13--25}.
\newblock


\bibitem[\protect\citeauthoryear{Ren, Poulenard, Wonka, and Ovsjanikov}{Ren
  et~al\mbox{.}}{2018}]%
        {ren2018continuous}
\bibfield{author}{\bibinfo{person}{Jing Ren}, \bibinfo{person}{Adrien
  Poulenard}, \bibinfo{person}{Peter Wonka}, {and} \bibinfo{person}{Maks
  Ovsjanikov}.} \bibinfo{year}{2018}\natexlab{}.
\newblock \showarticletitle{Continuous and Orientation-preserving
  Correspondences via Functional Maps}.
\newblock \bibinfo{journal}{\emph{ACM Transactions on Graphics (TOG)}}
  \bibinfo{volume}{37}, \bibinfo{number}{6} (\bibinfo{year}{2018}).
\newblock


\bibitem[\protect\citeauthoryear{Rodol{\`a}, Cosmo, Bronstein, Torsello, and
  Cremers}{Rodol{\`a} et~al\mbox{.}}{2017}]%
        {rodola2017partial}
\bibfield{author}{\bibinfo{person}{Emanuele Rodol{\`a}}, \bibinfo{person}{Luca
  Cosmo}, \bibinfo{person}{Michael Bronstein}, \bibinfo{person}{Andrea
  Torsello}, {and} \bibinfo{person}{Daniel Cremers}.}
  \bibinfo{year}{2017}\natexlab{}.
\newblock \showarticletitle{Partial functional correspondence}.
\newblock \bibinfo{journal}{\emph{Computer Graphics Forum}}
  \bibinfo{volume}{36}, \bibinfo{number}{1} (\bibinfo{year}{2017}),
  \bibinfo{pages}{222--236}.
\newblock


\bibitem[\protect\citeauthoryear{Rodol\`{a}, Moeller, and Cremers}{Rodol\`{a}
  et~al\mbox{.}}{2015}]%
        {rodola-vmv15}
\bibfield{author}{\bibinfo{person}{Emanuele Rodol\`{a}},
  \bibinfo{person}{Michael Moeller}, {and} \bibinfo{person}{Daniel Cremers}.}
  \bibinfo{year}{2015}\natexlab{}.
\newblock \showarticletitle{Point-wise Map Recovery and Refinement from
  Functional Correspondence}. In \bibinfo{booktitle}{\emph{Proc. Vision,
  Modeling and Visualization (VMV)}}.
\newblock


\bibitem[\protect\citeauthoryear{Rodol{\`a}, Rota~Bul{\`o}, Windheuser,
  Vestner, and Cremers}{Rodol{\`a} et~al\mbox{.}}{2014}]%
        {rodola2014dense}
\bibfield{author}{\bibinfo{person}{Emanuele Rodol{\`a}},
  \bibinfo{person}{Samuel Rota~Bul{\`o}}, \bibinfo{person}{Thomas Windheuser},
  \bibinfo{person}{Matthias Vestner}, {and} \bibinfo{person}{Daniel Cremers}.}
  \bibinfo{year}{2014}\natexlab{}.
\newblock \showarticletitle{Dense non-rigid shape correspondence using random
  forests}. In \bibinfo{booktitle}{\emph{IEEE Conference on Computer Vision and
  Pattern Recognition (CVPR)}}. \bibinfo{publisher}{IEEE},
  \bibinfo{pages}{4177--4184}.
\newblock


\bibitem[\protect\citeauthoryear{Roufosse, Sharma, and Ovsjanikov}{Roufosse
  et~al\mbox{.}}{2018}]%
        {roufosse2018unsupervised}
\bibfield{author}{\bibinfo{person}{Jean-Michel Roufosse},
  \bibinfo{person}{Abhishek Sharma}, {and} \bibinfo{person}{Maks Ovsjanikov}.}
  \bibinfo{year}{2018}\natexlab{}.
\newblock \showarticletitle{Unsupervised Deep Learning for Structured Shape
  Matching}.
\newblock \bibinfo{journal}{\emph{arXiv preprint arXiv:1812.03794}}
  (\bibinfo{year}{2018}).
\newblock


\bibitem[\protect\citeauthoryear{Rustamov, Ovsjanikov, Azencot, Ben-Chen,
  Chazal, and Guibas}{Rustamov et~al\mbox{.}}{2013}]%
        {rustamov2013map}
\bibfield{author}{\bibinfo{person}{Raif~M Rustamov}, \bibinfo{person}{Maks
  Ovsjanikov}, \bibinfo{person}{Omri Azencot}, \bibinfo{person}{Mirela
  Ben-Chen}, \bibinfo{person}{Fr{\'e}d{\'e}ric Chazal}, {and}
  \bibinfo{person}{Leonidas Guibas}.} \bibinfo{year}{2013}\natexlab{}.
\newblock \showarticletitle{Map-based exploration of intrinsic shape
  differences and variability}.
\newblock \bibinfo{journal}{\emph{ACM Transactions on Graphics (TOG)}}
  \bibinfo{volume}{32}, \bibinfo{number}{4} (\bibinfo{year}{2013}),
  \bibinfo{pages}{72}.
\newblock


\bibitem[\protect\citeauthoryear{Sahillio{\u{g}}lu and Yemez}{Sahillio{\u{g}}lu
  and Yemez}{2012}]%
        {sahilliouglu2012minimum}
\bibfield{author}{\bibinfo{person}{Yusuf Sahillio{\u{g}}lu} {and}
  \bibinfo{person}{Y{\"u}cel Yemez}.} \bibinfo{year}{2012}\natexlab{}.
\newblock \showarticletitle{Minimum-distortion isometric shape correspondence
  using EM algorithm}.
\newblock \bibinfo{journal}{\emph{IEEE transactions on pattern analysis and
  machine intelligence}} \bibinfo{volume}{34}, \bibinfo{number}{11}
  (\bibinfo{year}{2012}), \bibinfo{pages}{2203--2215}.
\newblock


\bibitem[\protect\citeauthoryear{Scott and Longuet-Higgins}{Scott and
  Longuet-Higgins}{1991}]%
        {scott1991}
\bibfield{author}{\bibinfo{person}{Guy~L Scott} {and}
  \bibinfo{person}{Hugh~Christopher Longuet-Higgins}.}
  \bibinfo{year}{1991}\natexlab{}.
\newblock \showarticletitle{An algorithm for associating the features of two
  images}.
\newblock \bibinfo{journal}{\emph{Proc. R. Soc. Lond. B}}
  \bibinfo{volume}{244}, \bibinfo{number}{1309} (\bibinfo{year}{1991}),
  \bibinfo{pages}{21--26}.
\newblock


\bibitem[\protect\citeauthoryear{Shoham, Vaxman, and Ben-Chen}{Shoham
  et~al\mbox{.}}{2019}]%
        {Shoham2019hierarchicalFmap}
\bibfield{author}{\bibinfo{person}{Meged Shoham}, \bibinfo{person}{Amir
  Vaxman}, {and} \bibinfo{person}{Mirela Ben-Chen}.}
  \bibinfo{year}{2019}\natexlab{}.
\newblock \showarticletitle{{Hierarchical Functional Maps between Subdivision
  Surfaces}}.
\newblock \bibinfo{journal}{\emph{Computer Graphics Forum}}
  (\bibinfo{year}{2019}).
\newblock
\showISSN{1467-8659}
\urldef\tempurl%
\url{https://doi.org/10.1111/cgf.13789}
\showDOI{\tempurl}


\bibitem[\protect\citeauthoryear{Solomon, Peyr{\'e}, Kim, and Sra}{Solomon
  et~al\mbox{.}}{2016}]%
        {solomon2016entropic}
\bibfield{author}{\bibinfo{person}{Justin Solomon}, \bibinfo{person}{Gabriel
  Peyr{\'e}}, \bibinfo{person}{Vladimir~G Kim}, {and} \bibinfo{person}{Suvrit
  Sra}.} \bibinfo{year}{2016}\natexlab{}.
\newblock \showarticletitle{Entropic metric alignment for correspondence
  problems}.
\newblock \bibinfo{journal}{\emph{ACM Transactions on Graphics (TOG)}}
  \bibinfo{volume}{35}, \bibinfo{number}{4} (\bibinfo{year}{2016}),
  \bibinfo{pages}{72}.
\newblock


\bibitem[\protect\citeauthoryear{Sun, Ovsjanikov, and Guibas}{Sun
  et~al\mbox{.}}{2009}]%
        {sun2009concise}
\bibfield{author}{\bibinfo{person}{Jian Sun}, \bibinfo{person}{Maks
  Ovsjanikov}, {and} \bibinfo{person}{Leonidas Guibas}.}
  \bibinfo{year}{2009}\natexlab{}.
\newblock \showarticletitle{A concise and provably informative multi-scale
  signature based on heat diffusion}.
\newblock \bibinfo{journal}{\emph{Computer graphics forum}}
  \bibinfo{volume}{28}, \bibinfo{number}{5} (\bibinfo{year}{2009}),
  \bibinfo{pages}{1383--1392}.
\newblock


\bibitem[\protect\citeauthoryear{Tam, Cheng, Lai, Langbein, Liu, Marshall,
  Martin, Sun, and Rosin}{Tam et~al\mbox{.}}{2013}]%
        {tam2013registration}
\bibfield{author}{\bibinfo{person}{Gary~KL Tam}, \bibinfo{person}{Zhi-Quan
  Cheng}, \bibinfo{person}{Yu-Kun Lai}, \bibinfo{person}{Frank~C Langbein},
  \bibinfo{person}{Yonghuai Liu}, \bibinfo{person}{David Marshall},
  \bibinfo{person}{Ralph~R Martin}, \bibinfo{person}{Xian-Fang Sun}, {and}
  \bibinfo{person}{Paul~L Rosin}.} \bibinfo{year}{2013}\natexlab{}.
\newblock \showarticletitle{Registration of {3D} point clouds and meshes: a
  survey from rigid to nonrigid}.
\newblock \bibinfo{journal}{\emph{IEEE TVCG}} \bibinfo{volume}{19},
  \bibinfo{number}{7} (\bibinfo{year}{2013}), \bibinfo{pages}{1199--1217}.
\newblock


\bibitem[\protect\citeauthoryear{Tombari, Salti, and Di~Stefano}{Tombari
  et~al\mbox{.}}{2010}]%
        {shot}
\bibfield{author}{\bibinfo{person}{Federico Tombari}, \bibinfo{person}{Samuele
  Salti}, {and} \bibinfo{person}{Luigi Di~Stefano}.}
  \bibinfo{year}{2010}\natexlab{}.
\newblock \showarticletitle{Unique signatures of histograms for local surface
  description}. In \bibinfo{booktitle}{\emph{Proc. ECCV}}. Springer,
  \bibinfo{pages}{356--369}.
\newblock


\bibitem[\protect\citeauthoryear{Umeyama}{Umeyama}{1988}]%
        {umeyama1988}
\bibfield{author}{\bibinfo{person}{Shinji Umeyama}.}
  \bibinfo{year}{1988}\natexlab{}.
\newblock \showarticletitle{An eigendecomposition approach to weighted graph
  matching problems}.
\newblock \bibinfo{journal}{\emph{IEEE transactions on pattern analysis and
  machine intelligence}} \bibinfo{volume}{10}, \bibinfo{number}{5}
  (\bibinfo{year}{1988}), \bibinfo{pages}{695--703}.
\newblock


\bibitem[\protect\citeauthoryear{Van~Kaick, Zhang, Hamarneh, and
  Cohen-Or}{Van~Kaick et~al\mbox{.}}{2011}]%
        {van2011survey}
\bibfield{author}{\bibinfo{person}{Oliver Van~Kaick}, \bibinfo{person}{Hao
  Zhang}, \bibinfo{person}{Ghassan Hamarneh}, {and} \bibinfo{person}{Daniel
  Cohen-Or}.} \bibinfo{year}{2011}\natexlab{}.
\newblock \showarticletitle{A survey on shape correspondence}.
\newblock \bibinfo{journal}{\emph{Computer Graphics Forum}}
  \bibinfo{volume}{30}, \bibinfo{number}{6} (\bibinfo{year}{2011}),
  \bibinfo{pages}{1681--1707}.
\newblock


\bibitem[\protect\citeauthoryear{Vestner, L{\"a}hner, Boyarski, Litany,
  Slossberg, Remez, Rodol\`a, Bronstein, Bronstein, and Kimmel}{Vestner
  et~al\mbox{.}}{2017a}]%
        {vestner2017efficient}
\bibfield{author}{\bibinfo{person}{Matthias Vestner}, \bibinfo{person}{Zorah
  L{\"a}hner}, \bibinfo{person}{Amit Boyarski}, \bibinfo{person}{Or Litany},
  \bibinfo{person}{Ron Slossberg}, \bibinfo{person}{Tal Remez},
  \bibinfo{person}{Emanuele Rodol\`a}, \bibinfo{person}{Alex Bronstein},
  \bibinfo{person}{Michael Bronstein}, {and} \bibinfo{person}{Ron Kimmel}.}
  \bibinfo{year}{2017}\natexlab{a}.
\newblock \showarticletitle{Efficient deformable shape correspondence via
  kernel matching}. In \bibinfo{booktitle}{\emph{3D Vision (3DV), 2017
  International Conference on}}. IEEE, \bibinfo{pages}{517--526}.
\newblock


\bibitem[\protect\citeauthoryear{Vestner, Litman, Rodol{\`a}, Bronstein, and
  Cremers}{Vestner et~al\mbox{.}}{2017b}]%
        {vestner2017product}
\bibfield{author}{\bibinfo{person}{Matthias Vestner}, \bibinfo{person}{Roee
  Litman}, \bibinfo{person}{Emanuele Rodol{\`a}}, \bibinfo{person}{Alex
  Bronstein}, {and} \bibinfo{person}{Daniel Cremers}.}
  \bibinfo{year}{2017}\natexlab{b}.
\newblock \showarticletitle{Product Manifold Filter: Non-rigid Shape
  Correspondence via Kernel Density Estimation in the Product Space}. In
  \bibinfo{booktitle}{\emph{Proc. CVPR}}. \bibinfo{pages}{6681--6690}.
\newblock


\bibitem[\protect\citeauthoryear{Wang, Huang, and Guibas}{Wang
  et~al\mbox{.}}{2013}]%
        {wang2013}
\bibfield{author}{\bibinfo{person}{Fan Wang}, \bibinfo{person}{Qixing Huang},
  {and} \bibinfo{person}{Leonidas~J. Guibas}.} \bibinfo{year}{2013}\natexlab{}.
\newblock \showarticletitle{Image Co-segmentation via Consistent Functional
  Maps}. In \bibinfo{booktitle}{\emph{Proc. ICCV}}. \bibinfo{pages}{849--856}.
\newblock


\bibitem[\protect\citeauthoryear{Wang and Huang}{Wang and Huang}{2017}]%
        {wang2017group}
\bibfield{author}{\bibinfo{person}{Hui Wang} {and} \bibinfo{person}{Hui
  Huang}.} \bibinfo{year}{2017}\natexlab{}.
\newblock \showarticletitle{Group representation of global intrinsic
  symmetries}. In \bibinfo{booktitle}{\emph{Computer Graphics Forum}},
  Vol.~\bibinfo{volume}{36}. Wiley Online Library, \bibinfo{pages}{51--61}.
\newblock


\bibitem[\protect\citeauthoryear{Wang, Gehre, Bronstein, and Solomon}{Wang
  et~al\mbox{.}}{2018a}]%
        {wang2018kernel}
\bibfield{author}{\bibinfo{person}{Larry Wang}, \bibinfo{person}{Anne Gehre},
  \bibinfo{person}{Michael Bronstein}, {and} \bibinfo{person}{Justin Solomon}.}
  \bibinfo{year}{2018}\natexlab{a}.
\newblock \showarticletitle{Kernel Functional Maps}.
\newblock \bibinfo{journal}{\emph{Computer Graphics Forum}}
  \bibinfo{volume}{37}, \bibinfo{number}{5} (\bibinfo{year}{2018}),
  \bibinfo{pages}{27--36}.
\newblock


\bibitem[\protect\citeauthoryear{Wang and Singer}{Wang and Singer}{2013}]%
        {wang2013exact}
\bibfield{author}{\bibinfo{person}{Lanhui Wang} {and} \bibinfo{person}{Amit
  Singer}.} \bibinfo{year}{2013}\natexlab{}.
\newblock \showarticletitle{Exact and stable recovery of rotations for robust
  synchronization}.
\newblock \bibinfo{journal}{\emph{Information and Inference: A Journal of the
  IMA}} \bibinfo{volume}{2}, \bibinfo{number}{2} (\bibinfo{year}{2013}),
  \bibinfo{pages}{145--193}.
\newblock


\bibitem[\protect\citeauthoryear{Wang, Liu, Zhou, and Tong}{Wang
  et~al\mbox{.}}{2018b}]%
        {wang2018vector}
\bibfield{author}{\bibinfo{person}{Y Wang}, \bibinfo{person}{B Liu},
  \bibinfo{person}{K Zhou}, {and} \bibinfo{person}{Y Tong}.}
  \bibinfo{year}{2018}\natexlab{b}.
\newblock \showarticletitle{Vector Field Map Representation for Near Conformal
  Surface Correspondence}.
\newblock \bibinfo{journal}{\emph{Computer Graphics Forum}}
  \bibinfo{volume}{37}, \bibinfo{number}{6} (\bibinfo{year}{2018}),
  \bibinfo{pages}{72--83}.
\newblock


\end{thebibliography}

\appendix
\section{Theoretical Analysis}

\paragraph{Proof of Theorem \ref{thm:energy}}
We will prove this theorem with the help of the following well-known lemma, for which we give the
proof in the Supplementary Material for completeness:

\begin{lemma}
\label{lemma:diag}
Let us be given a pair of shapes $\M,\N$ each having non-repeating
Laplacian eigenvalues, which are the same. A point-to-point map
$T:\M \rightarrow \N$ is an isometry if and only if the corresponding
functional map $\C$ in the complete Laplacian basis is both diagonal
and orthonormal.
\end{lemma}
\begin{proof}
  To prove Theorem \ref{thm:energy} first suppose that the map $T$ is
  an isometry, and thus, thanks to Lemma \ref{lemma:diag}, the
  functional map $\C = \Phi^{+}_{\M} \Pi \Phi_{\N}$ is diagonal and
  orthonormal. From this, it immediately follows that every principal
  submatrix of $\C$ must also be orthonormal implying $E(\C) = 0$.

  To prove the converse, suppose that $\C \in \mathcal{P}$. Then
  $E(\C) = 0$ implies that every principal submatrix of $\C$ is
  orthonormal. By induction on $k$ this implies that $\C$ must also be 
  diagonal. Finally since $\C \in \mathcal{P}$, again using Lemma
  \ref{lemma:diag} we obtain that the corresponding pointwise map must
  be an isometry.
\end{proof}

\subsection{Map Recovery}
Our goal is to prove that Eq. \eqref{eq:ourprob3} with the regularizer
$\mathcal{R}(\Pi) = \| (I - \Phi^{k}_{\M} (\Phi^{k}_{\M})^{+}) \Pi
\Phi^k_{\N} \C_K^T \|_{A_\M}^2$ is equivalent to solving $ \min_{\Pi} \|   \Pi
\Phi^k_{\N} \C_k^T - \Phi^{k}_{\M} \|^2_{F}$: In other words:
\begin{align}
  \notag
  &\argmin_{\Pi} \|   (\Phi^{k}_{\M})^{+} \Pi \Phi^k_{\N} \C_k^T - I_k
  \|^2_F + \| (I - \Phi^{k}_{\M} (\Phi^{k}_{\M})^{+}) \Pi \Phi^k_{\N}
    \C_K^T \|_{A_\M}^2\\
  \label{eq:ourprob_append}
  = &\argmin_{\Pi} \|   \Pi \Phi^k_{\N} \C_k^T - \Phi^{k}_{\M} \|^2_{F}.
\end{align}
For this we use the following result: for any matrix $X$ and basis $B$ that is orthonormal with
respect to a symmetric positive definite matrix $A$, i.e. $B^T A B = Id$, and thus
$B^{+} = B^T A$, if we let $\|X\|^2_A = tr(X^T A X)$ then:
$\|X\|^2_A = \|B^{+} X\|_F^2 + \| (I - B B^{+}) X\|_A^2$. To see this, observe that
$\|B^{+} X\|_F^2 = tr(X^T A B B^T A X)$ while $\| (I - B B^{+}) X\|_A^2 = tr\left(X^T  (I - A B
  B^{T}) A  (I - B B^T A) X \right) = tr\left(X^T  (A - A B B^{T} A) X \right)$ since
$B^T A B$. 
We now use this result with $X = \Pi \Phi^k_{\N} \C_k^T -
\Phi^{k}_{\M}$, $B = \Phi^k_{\M}$ and $A = A_{\M}$. This gives:
\begin{align*}
  \| \Pi \Phi^k_{\N} \C_k^T - \Phi^{k}_{\M} \|^2_{A_{\M}} =
  \| (\Phi^{k}_{\M})^{+} \left(\Pi \Phi^k_{\N} \C_k^T - \Phi^{k}_{\M}\right) \|^2_F \\
   + \| \left(I - \Phi^{k}_{\M} (\Phi^{k}_{\M})^{+}\right) \left(\Pi \Phi^k_{\N} \C_k^T - \Phi^{k}_{\M}\right) \|^2_{A_{\M}} \\
  = \| (\Phi^{k}_{\M})^{+} \Pi \Phi^k_{\N} \C_k^T - I_k \|^2_F 
  + \| \left(I - \Phi^{k}_{\M} (\Phi^{k}_{\M})^{+}\right) \Pi \Phi^k_{\N} \C_k^T \|^2_{A_{\M}}.
\end{align*}
It remains to show that $\argmin_{\Pi} \| X \|^2_{A_{\M}}  =
\argmin_{\Pi} \| X \|_F^2$ with $X = \Pi \Phi^k_{\N} \C_k^T -
\Phi^{k}_{\M}$. For this note simply that since $\Pi$ represents a
pointwise map, both problems reduce to finding the row of $\Phi^k_{\N}
\C_k^T$ that is closest to each of the rows of $\Phi^{k}_{\M}$.

Note that in supplementary material we derive both an alternative
approach to \name\ and, as mentioned in Section 4.3. provide a link
between our approach and PMF.

\section{Implementation}
\label{sec:appendix:code}
This Appendix lists standard Matlab code for our method and BCICP,
which is the most competitive method to ours while being orders of
magnitude slower. Note that a fully-working version of \name~\textbf{can be
implemented in just 5 lines of code}, while BCICP relies on the
computation of all pairs of geodesics distances on both shapes, and even
after pre-computation, is more than 250 lines of code relying on
numerous parameters and spread across a main procedure and 4 utility
functions.

\lstset{language=Matlab,%
    breaklines=true,%
    morekeywords={matlab2tikz},
    keywordstyle=\color{blue},%
    morekeywords=[2]{1}, keywordstyle=[2]{\color{black}},
    identifierstyle=\color{black},%
    stringstyle=\color{mylilas},
    commentstyle=\color{mygreen},%
    showstringspaces=false,
    numbers=left,%
    numberstyle={\tiny \color{black}},
    numbersep=9pt, 
    emph=[1]{for,end,break},emphstyle=[1]\color{red}, 
    basicstyle=\linespread{0.8},
  }

\subsection{Source Code - \name}
\begin{lstlisting}
function [C,P]=ZoomOut(M,N,C,k_final)

for k=size(C,1):k_final-1
 x = knnsearch(N.Phi(:,1:k)*C',M.Phi(:,1:k));
 P = sparse(1:M.n,x,1,M.n,N.n);
 C = M.Phi(:,1:k+1)'*M.A*P*N.Phi(:,1:k+1);
end
\end{lstlisting}

\subsection{Source Code - BCICP \cite{ren2018continuous}}
\begin{footnotesize}
\begin{lstlisting}
Precompute: 
(1) complete pairwise geodesic distance matrix of each shape; 
(2) vertex one-ring neighbor; 
(3) edge list of each mesh

function [T12,T21]=BCICP(S1,S2,T12,T21,K)
B1 = S1.Phi(:,1:50); 
B2 = S2.Phi(:,1:50);
for k=1:K
    [T21,T12]=refine_pMap(T21,T12,S1,S2);
    C12 = B2\B1(T21,:);
    C21 = B1\B2(T12,:);
    T12 = knnsearch(B2*C21',B1);
    T21 = knnsearch(B1*C12',B2);
    [T21,T12]=refine_pMap(T21,T12,S1,S2);
    C1 = B1\B1(T21(T12),:);
    C1 = mat_projection(C1);
    C2 = B2\B2(T21(T12),:);
    C2 = mat_projection(C2);
    T21 = knnsearch(B2(T12,:)*C2',B2);
    T12 = knnsearch(B1(T21,:)*C1',B1);
end
end

function [T21,T12]=refine_pMap(T21,T12,S1,S2)
for k=1:4
    T12 = improve_coverage(T12,S1,S2);
    T21 = improve_coverage(T21,S1,S2);
    T12 = improve_smoothness(T12,S1,S2);
    T21 = improve_smoothness(T21,S1,S2);
    T12 = fix_outliers(T12,S1,S2);
    T21 = fix_outliers(T21,S1,S2);
end
end

function [T12]=improve_coverage(T12,S1,S2)
% around 120 lines of code
function [T12]=improve_smoothness(T12,S1,S2)
% around 50 lines of code
function [T12]=fix_outliers(T12,S1,S2)
% around 50 lines of code
...
\end{lstlisting}
  \end{footnotesize}

\section{Additional measurements}\label{sec:appendix:measurements}

\revised{
Given 10 random shape pairs from the FAUST original dataset, Table~{\ref{tb:append:measurement}} shows the performance summary of different refinement methods w.r.t. the following measurements as used in~{\cite{ren2018continuous}}. Specifically, we evaluate the maps $T_{12}$ and $T_{21}$ between a pair of shapes $S_1$ and $S_2$:

\begin{itemize}
    \item \textbf{Accuracy}. We measure the geodesic distance between $T_{12}$ (and $T_{21}$) and the given ground-truth correspondences. 
    \item \textbf{Un-Coverage}. The percentage of vertices/areas that are NOT covered by the map $T_{12}$ (or $T_{21}$) on shape $S_2$ (or $S_1$)
    \item \textbf{Bijectivity}. The composite map $T_{21}\circ T_{12}$ (or $T_{12}\circ T_{21}$) gives a map from the shape
      to itself. Thus, we measure the geodesic distance between this composite map and the identity.
    \item \textbf{Edge distortion}. We measure how each edge in $S_1$ (or $S_2$) is distorted by the map $T_{12}$ (or $T_{21}$) as follows:
    $$ e_{v_i \sim v_j} = \Big(\frac{d_{S_2}\big(T_{12}(v_i)), T_{12}(v_j)\big)}{d_{S_1}(v_i, v_j)} - 1\Big)^2$$
    We then average the distortion error over all the edges as a measure for the map smoothness.
\end{itemize}

Note that the PMF method optimizes for a permutation matrix directly, that is why the computed maps covered all the
vertices and give almost zero bijectivity error in Table~\ref{tb:append:measurement} (the bijiectivity error is not
strictly zero because the maps $T_{12}$ and $T_{21}$ are computed independently). The method BCICP includes heuristics
to explicitly improve the coverage, bijectivity, and smoothness. Even though our method is not designed to optimize
these measurements, it still achieves reasonable performance. Note that our method gives the smallest edge distortion, which suggests that our method not only gives the most accurate map but also the smoothest map w.r.t. all the competing methods.
}

\setlength{\tabcolsep}{0.25em}
\begin{table}[!t]
\caption{
\revised{
\textbf{Additional measurements.} Besides the map accuray, we also measure the coverage, bijectivity, and edge
distortion as a smoothness measure on 10 random shape pairs from the original FAUST dataset.
}
}
\vspace{-0.2cm}
\label{tb:append:measurement}
\centering
\footnotesize
\begin{tabular}{|c|c|cccccc|}
\hline
\multirow{2}{*}{Measurement\textbackslash~Method} & \multirow{2}{*}{Ini} & \multicolumn{6}{c|}{Refinement methods} \\ \cline{3-8} 
 &  & ICP & PMF (gauss) & RHM & BCICP & \textbf{ours} & \textbf{ours$^*$} \\ \hline
\textbf{Accuracy} ($\times 10^{-3}$) & 98.4 & 85.8 & 36.3 & 63.9 & 49.9 & \textbf{33.3} & 36.8 \\
\textbf{Un-Coverage} (\%) & 72.3 & 42.4 & \textbf{0} & 44.5 & 15.9 & 23.6 & 30.7 \\
\textbf{Bijectivity} ($\times 10^{-3}$) & 104 & 89.6 & \textbf{1.90} & 24.6 & 5.48 & 15.6 & 14.8 \\
\textbf{Edge distortion} & 10.9 & 26.4 & 37.3 & 3.69 & 5.49 & \textbf{1.16} & 4.37 \\ \hline
\end{tabular}
\end{table}

\section{Comparison to Deblur}\label{sec:appendix:deblur}

\begin{figure}[h]
\vspace{-0.1cm}
  \centering
  \begin{overpic}
  [trim=0cm 0cm 0cm 0cm,clip,width=0.95\linewidth, grid=false]{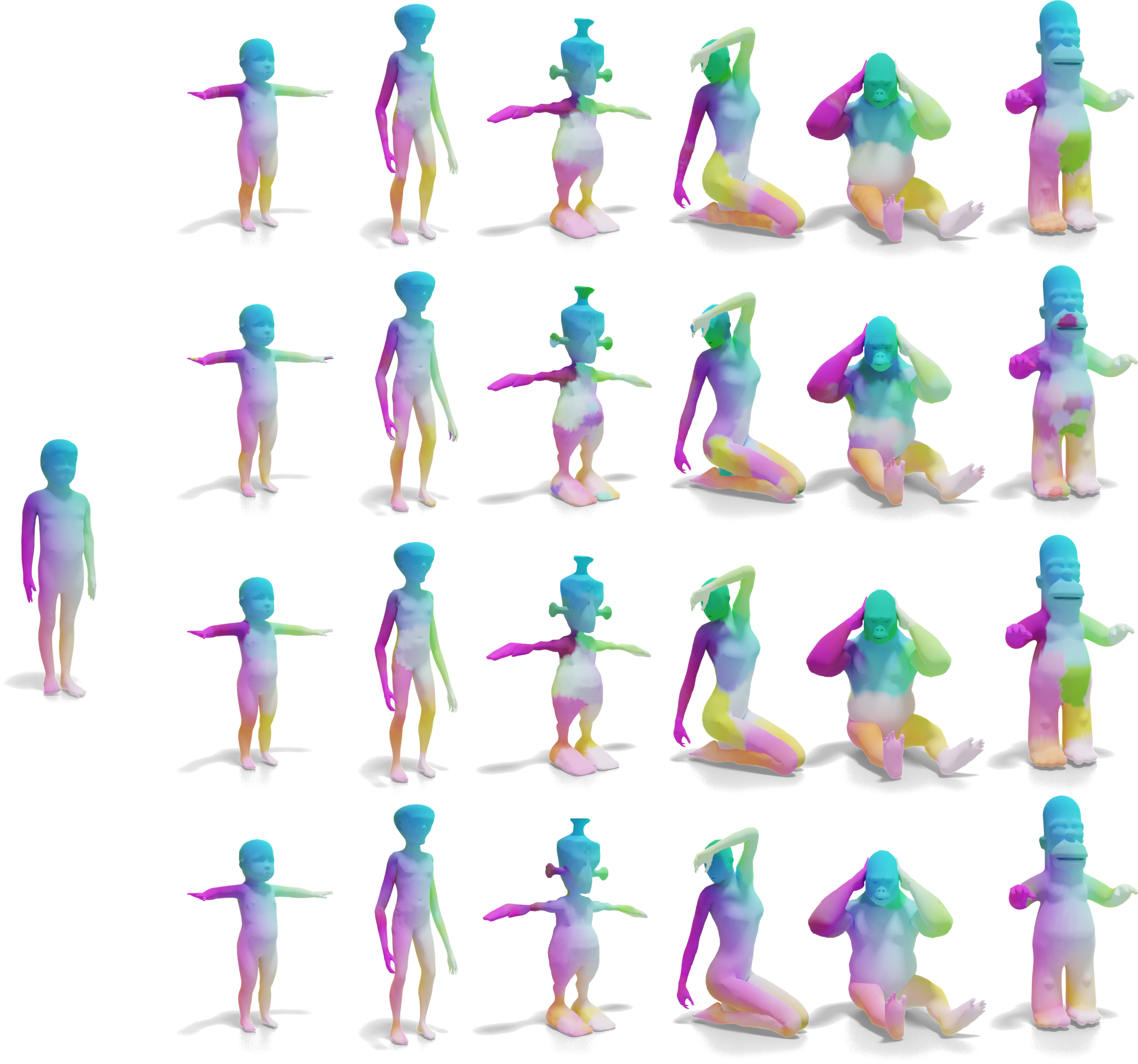}
  \put(-1,56){ Source}
  \put(100,92){ \rotatebox{-90}{Initialization}}
  \put(100,63){ \rotatebox{-90}{ICP}}
  \put(101,40){\rotatebox{-90}{Deblur}}
  \put(100,16){ \rotatebox{-90}{\textbf{Ours}}}
  \end{overpic}
  \caption{\revised{Comparison to~{\cite{ezuz2017deblurring}} on non-isometric shape pairs.} }
  \label{fig:append:deblur}
  \vspace{-2mm}
\end{figure}

\revised{ {The work of \cite{ezuz2017deblurring}} also provides an approach for recovering a
  point-wise map from a functional map, based on a different energy. Specifically, our energy
  defined in Eq.~\eqref{eq:ourprob} and the resulting point-wise map conversion step defined in
  Eq.~\eqref{eq:ourprob3} are different from the deblurring energy defined in Eq. (4)
  in~{\cite{ezuz2017deblurring}}. Figure~\ref{fig:append:deblur} shows a qualitative comparison
  between our method and this method, which we call ``Deblur.'' We use 5 landmarks to compute the
  initial maps (first row), and we then apply ICP, Deblur, and ours to refine the initial maps. Note
  that in these examples, we rescale the target shapes to the same surface area as the source 
  shape. We can see that even when the shape pair is far from isometry, our method can still produce
  reasonable maps, even though our theory relies on the isometry assumption.  }

\end{document}